\documentclass[a4paper,11pt]{article}
\usepackage{authblk}

\usepackage[
  paperwidth=199.8mm,
  paperheight=287mm,
  centering,
  hmargin=22.5mm,
  vmargin=2cm
]{geometry}
\usepackage[T1]{fontenc}
\usepackage{amsmath,amssymb,bm,mathtools,amsthm}
\usepackage{graphicx}
\usepackage{microtype}
\usepackage{enumitem}
\usepackage[dvipsnames]{xcolor}
\usepackage[normalem]{ulem}
\usepackage{tikz}
\usepackage[colorlinks=true,citecolor=blue!55!black,linkcolor=blue!55!black,urlcolor=blue!55!black]{hyperref}
\usepackage{comment}
\usetikzlibrary{arrows.meta,calc,decorations.pathreplacing,positioning,shapes.geometric}

\newcommand{\ee}{\mathrm{e}}
\newcommand{\ii}{\mathrm{i}}
\newcommand{\one}{\mathbf{1}}
\newcommand{\dd}{\mathrm{d}}
\newcommand{\bbE}{\mathbb{E}}
\newcommand{\bbP}{\mathbb{P}}
\newcommand{\cA}{\mathcal{A}}

\newcommand{\cG}{\mathcal{G}}
\newcommand{\cM}{\mathcal{M}}
\newcommand{\cQ}{\mathcal{Q}}
\newcommand{\cS}{\mathcal{S}}
\newcommand{\cT}{\mathcal{T}}
\newcommand{\order}{\mathcal{O}}
\newcommand{\poly}{\operatorname{poly}}
\newcommand{\polylog}{\operatorname{polylog}}
\newcommand{\TV}{\operatorname{TV}}
\DeclareMathOperator{\Tr}{Tr}
\DeclareMathOperator{\supp}{supp}
\DeclareMathOperator{\conv}{conv}
\DeclareMathOperator{\sgn}{sgn}
\DeclarePairedDelimiter{\norm}{\lVert}{\rVert}

\newtheorem{theorem}{Theorem}
\newtheorem{lemma}[theorem]{Lemma}
\newtheorem{proposition}[theorem]{Proposition}
\newtheorem{corollary}[theorem]{Corollary}
\theoremstyle{definition}

\newcommand{\hzy}[1]{}

\newcommand{\zyc}[1]{}

\begin{document}

\title{Weakly interacting fermionic Gibbs states are\\ Gaussian mixtures and classically simulable}

\author[1]{Zhengyi Han\thanks{These authors contributed equally to this work.}}
\author[2]{Yuanchen Zhao\protect\footnotemark[1]}
\author[1]{Zi-Wen Liu\thanks{\texttt{zwliu0@tsinghua.edu.cn}}}
\affil[1]{Yau Mathematical Sciences Center, Tsinghua University, Beijing 100084, China}
\affil[2]{State Key Laboratory of Low Dimensional Quantum Physics, Department of Physics, Tsinghua University, Beijing, 100084, China}

\date{\today}

\maketitle

\begin{abstract}
 Identifying regimes in which interacting fermionic systems admit efficient classical simulation is a crucial task in quantum many-body physics. In this work, we rigorously establish the Gaussianity and classical simulability of sufficiently weakly interacting fermions. For the Gibbs state of a local interacting fermionic Hamiltonian $H=H_0+V$ at any finite temperature and in any spatial dimension, where $H_0$ is quadratic and $V$ is the interaction term, we show that there exists a positive interaction strength below which the Gibbs state is a convex mixture of fermionic Gaussian states. Moreover, based on this structure, we construct a polynomial-time classical sampler for the Gibbs state via Metropolis updates on a tree whose leaves are Gaussian states.
\end{abstract}

\setcounter{tocdepth}{2} 
\tableofcontents

\section{Introduction}
\label{sec:introduction}
{
Simulating interacting fermions at finite temperature is an important computational physics problem, with wide applications in condensed matter physics, quantum chemistry, nuclear physics, and quantum gravity~\cite{QinEtAl2022,LeePhamReichman2022,Alhassid2001,MaldacenaStanford2016}. Although general interacting fermionic problems are computationally hard in the worst case~\cite{TroyerWiese2005,SchuchVerstraete2009,OGormanEtAl2022}, physically relevant regimes may nevertheless admit efficient classical
simulation. The practical efficiency of widely used methods often
depends on favorable model structures. Tensor network methods are particularly effective in one dimension, but their efficiency generally does not extend to higher-dimensional interacting systems.
Quantum Monte Carlo can treat higher dimensions, but the fermion sign problem can lead to exponential computational costs~\cite{TroyerWiese2005}. Certain specific models can circumvent the sign problem, but require additional symmetry or other structures.
This raises the natural question of which parameter regimes allow efficient classical simulation of interacting fermionic systems at finite temperature without relying on low dimensionality or special symmetries.

Free fermions provide a natural starting point. Their Gibbs states are Gaussian at any temperature. Fermionic Gaussian states are completely specified by their covariance matrices and obey Wick's theorem, enabling efficient classical simulation~\cite{TerhalDiVincenzo2002,Bravyi2005,SuraceTagliacozzo2022}.
Probabilistic mixtures of Gaussian states are called convex-Gaussian states~\cite{DeMeloCwiklinskiTerhal2013}. Their components can be entangled, so this representation goes beyond mixtures of spatially uncorrelated states. Convex-Gaussian states also support efficient classical simulation, e.g., when their sample space has polynomial size \cite{DeMeloCwiklinskiTerhal2013,OszmaniecGuttKus2014}. Therefore, efficiently sampleable Gaussian decompositions provide a route to the classical simulation of interacting fermionic systems.

High temperature provides a regime where both the Gaussian structure and efficient sampling can be established rigorously.
Recent work showed that high-temperature Gibbs states of local fermionic Hamiltonians are convex-Gaussian and constructed an efficient classical sampler~\cite{RamkumarCaiTongJiang2026}, paralleling the product state structure established for spin systems~\cite{BakshiLiuMoitraTang2025}.

However, the low-temperature regime matters more for physical applications, as interaction-driven quantum many-body phenomena are typically most pronounced when thermal fluctuations are suppressed. 
Beyond the high-temperature setting, rigorous classical algorithms have also been obtained for sufficiently weak interactions at finite temperature. Recent work established a polynomial-time classical importance sampling algorithm for estimating the log-partition function and local observables of weakly interacting fermions~\cite{ChenEtAl2025}. However, their algorithm computes specified quantities rather than sampling a representation of the Gibbs state, limiting access to quantities like nonlocal high-order correlation functions and nonlinear properties such as entanglement entropies.

In this work, we rigorously establish a universal weak interaction regime that supports convex-Gaussian structure as well as efficient
classical sampling. For local fermionic Hamiltonians $H=H_0+V$ in any spatial dimension, with a quadratic term \(H_0\) and an interaction term \(V\), we show that at any finite temperature there is a positive interaction threshold $g_{\rm mix} = e^{-\Theta(\beta)}>0 $ below which the Gibbs state is exactly convex-Gaussian. Moreover, below a possibly smaller threshold value $g_{\rm sample}= e^{-\Theta(\beta)}$, we construct a classically efficient Gibbs state sampler which provides access to the entire Gibbs ensemble. Although the decomposition may contain exponentially many Gaussian states, our sampler uses the Metropolis method on a polynomial-depth sampling tree without enumerating all components, so that the sampler requires only polynomial time cost. This state-level sampling enables estimation of not only local observables and free energy, but also higher-order correlation functions, full counting statistics of subsystem particle number~\cite{HumeniukBuchler2017}, and nonlinear functionals of reduced density matrices. Our results establish weak interactions as a broad regime for efficient classical simulation of interacting fermions without requiring high temperatures, low spatial dimensions, or special symmetries. 
}

\section{Main results and physical examples}
\label{sec:main-result}

\subsection{Setting and theorem}

Consider $n$ fermionic modes arranged on a $D$-dimensional lattice, with a constant number of modes per lattice site, and Majorana operators $\gamma_1,\ldots,\gamma_{2n}$.
We write
\begin{equation}
 H=H_0+V,\qquad
 H_0=\frac14\bm\gamma^{\mathsf T}K_0\bm\gamma,
 \label{eq:H}
\end{equation}
where $H_0$ is even and quadratic and admits a bounded-local-norm
decomposition on balls of radius at most $r_0$.
The matrix $K_0$ may break particle-number conservation.
We write $V=\sum_B V_B$, where each $V_B$ is even, supported in a ball of fixed radius $r_0$, and obeys $\norm{V_B}\le g$.
The local quadratic norm is fixed to order one.
Let $\cG_n$ denote all normalized fermionic Gaussian density operators, { and $\conv(\cG_n)$ denote all convex combinations of states in $\cG_n$.} 

{

\begin{theorem}[Informal]
\label{thm:main}
Fix finite $\beta>0$, spatial dimension $D$, and a constant $r_0$.
There exists $g_{\rm mix}(\beta,D,r_0)>0$, independent of $n$, such
that for every Hamiltonian~\eqref{eq:H} with $g\le g_{\rm mix}$,
\begin{equation}
 \rho_\beta=\frac{\ee^{-\beta H}}{\Tr\ee^{-\beta H}}
 \in\conv(\cG_n).
 \label{eq:convex}
\end{equation}
Moreover, there exists $g_{\rm samp}(\beta,D,r_0)$ satisfying
$0<g_{\rm samp}\le g_{\rm mix}$ such that, for every $\epsilon>0$ and
$g\le g_{\rm samp}$, a randomized classical algorithm runs in
$\poly(n,\log(1/\epsilon))$ time and outputs $\tau\in\cG_n$ satisfying
\begin{equation}
 \norm{\bbE\tau-\rho_\beta}_1\le\epsilon,
 \label{eq:sampling}
\end{equation}
where the expectation is over the internal randomness of the algorithm.
\end{theorem}
}

{
The theorem establishes both the convex-Gaussianity and classical simulability of geometrically local, weakly interacting fermionic systems. The local interaction terms are only required to have bounded support and need not be restricted to quartic interactions. Notably, the sampling threshold may be smaller than the threshold for convex-Gaussianity, consistent with the intuition that not every convex-Gaussian state admits an efficient classical simulation. As will be shown below, the sample space associated with the convex decomposition of $\rho_\beta$ is exponentially large, and we employ a Metropolis update to sample from it efficiently. The final output of the algorithm is an ensemble of Gaussian states whose average approximates $\rho_\beta$.
}

The Gaussian state sampler also gives access to observables, full
counting statistics, and nonlinear functions of reduced states.  The
following corollary summarizes these consequences.
\begin{corollary}[Informal]
\label{cor:informal-physical-quantities}
Assume the sampling conditions of Theorem~\ref{thm:main}, and run the
sampler with error parameter $\epsilon_{\mathrm{samp}}/2$.  For statistical accuracy
$\delta>0$, the following quantities are classically accessible in
time polynomial in $n$, $1/\delta$, and
$\log(1/\epsilon_{\mathrm{samp}})$; the moment estimator also has
polynomial dependence on its order $k$.

\begin{enumerate}[label=(\arabic*)]
\item
Let $O$ be an operator given by a polynomially sparse, numerically
accessible Majorana expansion.  Repeated samples estimate
$\Tr(O\rho_\beta)$ with additive error at most
$(\epsilon_{\mathrm{samp}}+\delta)\norm O$.

\item
For any subset of fermionic modes, the simultaneous
projective measurement of the occupation numbers on $\rho_\beta$
can be sampled with total variation error at most
$\epsilon_{\mathrm{samp}}/2$.
Consequently, one can efficiently sample the full counting statistics as well as the fermion parity.

\item
Let $A$ be any subsystem. For any integer $k\ge2$,
$\Tr(\rho_{\beta,A}^k)$ can be estimated with additive error at most
$k\epsilon_{\mathrm{samp}}+\delta$.

\item
The von Neumann entropy of a subsystem containing $\order(\log n)$
fermionic modes can be estimated to inverse-polynomial additive
accuracy in polynomial time.
\end{enumerate}
\end{corollary}
Notably, Cor.~\ref{cor:informal-physical-quantities}(1) imposes no
restriction on the support of $O$. The observable may be highly
nonlocal and may have support on an extensive number of sites. Cor.~\ref{cor:informal-physical-quantities}(2) is not a simple application of Cor.~\ref{cor:informal-physical-quantities}(1), since the projection operator of occupation number might contain exponentially many Majorana monomials.
Also, the subsystem $A$ in
Cor.~\ref{cor:informal-physical-quantities}(3) may itself be extensive. For the evaluation of von Neumann entropy in Cor.~\ref{cor:informal-physical-quantities}(4), the subsystem size is restricted to $\mathcal O(\log n)$.
Cor.~\ref{cor:informal-physical-quantities}(3) and (4) further give access
to nonlinear functions of reduced density matrices, substantially
extending the class of efficiently accessible quantities beyond those
covered by the algorithm of Ref.~\cite{ChenEtAl2025}.

\subsection{Representative models}

As a first example, consider the Fermi-Hubbard model~\cite{QinEtAl2022}
\begin{align}
 H_0={}&-\sum_{\langle i,j\rangle,\sigma}
 t_{ij,\sigma}(c_{i\sigma}^\dagger c_{j\sigma}+\mathrm{H.c.})
 -\sum_{i,\sigma}\mu_{i\sigma}n_{i\sigma},
 \nonumber\\
 V={}&g\sum_i
 \left(n_{i\uparrow}-\frac12\right)
 \left(n_{i\downarrow}-\frac12\right).
 \label{eq:Hubbard-example}
\end{align}
The Fermi-Hubbard model provides a standard setting for studying Mott
insulating behavior~\cite{JordensEtAl2008}, antiferromagnetic
ordering~\cite{MazurenkoEtAl2017}, and unconventional
superconductivity~\cite{DongEtAl2022}. It has also been realized with
ultracold fermionic atoms in optical lattices~\cite{SchneiderEtAl2008}.
For finite temperature $\beta$ and bounded local hopping $|t_{ij,\sigma}| \leq t$, Theorem~\ref{thm:main} applies when $|g|$ is below a finite threshold, such that the weakly interacting thermal density operator is an average of sampleable free fermion states.
In particular, it requires no translation invariance and thus applies even in the presence of disorder.
 
Our theorem does not necessarily require particle conservation.
A second example is a Bogoliubov-de Gennes (BdG) background,
\begin{equation}
 H_0=\sum_{i,j}h_{ij}c_i^\dagger c_j
 +\frac12\sum_{i,j}(\Delta_{ij}c_i^\dagger c_j^\dagger+\mathrm{H.c.}),
 \label{eq:BdG-example}
\end{equation}
perturbed by weak local density, exchange, or other even many-fermion interactions.
BdG Hamiltonians provide the standard mean-field
framework for describing pairing and quasiparticle excitations in
superconductors~\cite{Bogoliubov1958,DeGennes1964}, and are also central
to the study of topological superconducting phases~\cite{ReadGreen2000,Kitaev2001}.
Here the pairing matrix $\Delta$ can be order one and is incorporated exactly. { This setting captures interacting superconducting systems whose quadratic description is given by a BdG Hamiltonian, with additional local fermion interactions. It therefore provides access to finite-temperature equilibrium properties of such systems, including anomalous pairing correlations \cite{Gorkov1959} and interaction-induced corrections \cite{Carbotte1990,KhveshchenkoEtAl2001} to observables beyond the quadratic BdG description, to be simulated efficiently within our regime.}

\section{Proof strategy}
\label{sec:proof-strategy}
We show that the Gibbs state is a convex mixture of Gaussian states and obtain a polynomial-time sampler.

A direct high-temperature expansion treats the full Hamiltonian perturbatively by expanding the Gibbs state in powers of \(\beta\). At low temperatures, however, the quadratic contribution can be large even when the interaction is weak. We therefore treat the quadratic evolution exactly and remove Majorana labels from the interaction support one at a time. At each step, the coefficient sum is controlled by the local interaction strength, with a bound independent of the system size.

\subsection{From local interaction to Gaussian leaves}

Set $T=\beta/2$.  Factoring out the quadratic evolution gives
\begin{equation}
G=\ee^{-TH_0},
\qquad
X=\ee^{TH_0}\ee^{-TH}
=
\cT\exp\left[-\int_0^T
\ee^{uH_0}V\ee^{-uH_0}\,\dd u\right].
\label{eq:X}
\end{equation}
and hence
\begin{equation}
\ee^{-\beta H}=GXX^\dagger G.
\label{eq:whiten}
\end{equation}
The operator $G$ is Gaussian, while $XX^\dagger$ contains the
interaction and is positive.  It therefore suffices to decompose
$XX^\dagger$ into nonnegative Gaussian operators.

Conjugation by \(e^{uH_0}\) spreads local interaction terms over the lattice. For any fixed Majorana label, however, the weighted sum of the absolute coefficients of terms containing that label is bounded independently of the system size. If
\begin{equation}
\ee^{uH_0}V\ee^{-uH_0}
=
\sum_Rv_R(u)P_R,
\end{equation} Appendix~\ref{sec:SM-mixture-proof} proves
\begin{equation}
\max_a\sum_{R\ni a}|v_R(u)|q^{|R|}
\le
q^dC_Vg\,\ee^{dC_Hu}.
\label{eq:strategy-local-coefficient}
\end{equation}
The argument below requires only this coefficient bound.

The nonidentity coefficient weight in the full expansion of $X$
generally grows with the system size. We therefore factor $X$ by
removing the Majorana labels one at a time. For a set $S$ of active
Majorana labels, let $X^S$ be the evolution obtained by keeping only
interaction monomials supported in $S$. For $i\in S$, define
$M_{S,i}$ by
\begin{equation}
X^S=X^{S\setminus\{i\}}M_{S,i}.
\label{eq:ratio}
\end{equation}
Only interaction terms containing $i$ distinguish the two evolutions.
The local bound therefore keeps the nonidentity coefficient weight
of $M_{S,i}-\one$ and $M_{S,i}^\dagger-\one$ small, uniformly in $S$,
$i$, and the system size.

At each removal step, we keep the middle operator in
$X^SC(X^S)^\dagger$ in the form
\begin{equation}
C=(\one+\alpha\Gamma)\sigma.
\label{eq:strategy-partial-state}
\end{equation}
Here $\Gamma$ is a Hermitian even Majorana monomial that has not yet
been reduced to quadratic form, while $\sigma$ is a product of
commuting quadratic operators. Expanding $M_{S,i}$ and its adjoint
gives an exact convex decomposition into operators of the same form.
The weights in this decomposition can be chosen so that each term
satisfies $\alpha\le2^{-|\supp\Gamma\cap S|}$. A monomial of degree $s$
can introduce at most $s$ active labels, and the factor $2^{-s}$ in
its coefficient compensates for this increase. In particular,
$\alpha\le1$, so $\one+\alpha\Gamma$ is nonnegative.

When $\Gamma$ contains no active label, write
$\Gamma=AB$, where $B=\ii\gamma_a\gamma_b$ is quadratic and is
disjoint from $A$.  The identity
\begin{equation}
\one+\alpha AB
=
\frac12(\one+\alpha A)(\one+B)
+
\frac12(\one-\alpha A)(\one-B)
\label{eq:strategy-hard-pinning}
\end{equation}
separates one quadratic pair from $\Gamma$ and transfers it to
$\sigma$.  Alternating label removal with this splitting eventually
gives
\begin{equation}
XX^\dagger
=
\sum_zp_z\sigma_z,
\qquad
p_z\ge0,
\qquad
\sum_zp_z=1,
\label{eq:tree-mean}
\end{equation}
where every $\sigma_z$ is nonnegative and Gaussian.  Together with
Eq.~\eqref{eq:whiten}, this proves the Gaussian-mixture theorem.

However, this theorem does not by itself give an efficient sampling algorithm. We also need to estimate the normalization
factors that determine the sampling probabilities. For these
estimates, we replace the singular operators $\one\pm B$ in
Eq.~\eqref{eq:strategy-hard-pinning} by the strictly positive
operators $\one\pm tB$, where $0<t<1$. We call this step \emph{soft pinning}, which keeps
$\sigma$ strictly positive. The modified recursion maintains
$\alpha\le1/2$, and hence
\begin{equation}
\frac12\sigma
\preceq
(\one+\alpha\Gamma)\sigma
\preceq
\frac32\sigma.
\label{eq:strategy-live-comparison}
\end{equation}
Thus the intermediate operators are also strictly positive, while the leaves remain Gaussian.

\subsection{Sampling the trace-weighted leaves}

Let $p_z$ be the probability of reaching a leaf $z$ in the sampling
tree, and write
\begin{equation}
L_z=G\sigma_zG,
\qquad
F_z=\Tr L_z,
\qquad
\tau_z=\frac{L_z}{F_z}.
\end{equation}
The tree identity gives
\begin{equation}
\rho_\beta
=
\sum_zq_z\tau_z,
\qquad
q_z
=
\frac{p_zF_z}{\Tr\ee^{-\beta H}}.
\label{eq:weights}
\end{equation}
The tree produces the probabilities $p_z$, but the normalized Gaussian
states must be sampled with probabilities proportional to $p_zF_z$.
The unknown trace $F_z$ is therefore the main obstacle in turning the
decomposition into an algorithm.

For a node $v$, let $W_v$ be the expected leaf trace below $v$.  If
$w$ ranges over the children of $v$, then
\begin{equation}
W_v
=
\sum_wp(w\,|\,v)W_w.
\end{equation}
Choosing $w$ with probability $p(w\,|\,v)W_w/W_v$ would produce the
distribution in Eq.~\eqref{eq:weights}.  Computing $W_v$ from this
definition, however, requires summing over the entire subtree below
$v$.

To estimate $W_v$ without summing over the subtree below $v$, we
use the exact decomposition identities, which give
\[
W_v=\Tr\left[GX^SC(X^S)^\dagger G\right].
\]
The comparison between $C$ and $\sigma$ in
Eq.~\eqref{eq:strategy-live-comparison} reduces the problem to
estimating
\begin{equation}
Z_{S,\sigma}
=
\Tr\left[GX^S\sigma(X^S)^\dagger G\right].
\label{eq:internalZ}
\end{equation}
Indeed, $Z_{S,\sigma}/2\le W_v\le3Z_{S,\sigma}/2$, so an estimate
of $\log Z_{S,\sigma}$ with constant additive error gives a constant-factor approximation to $W_v$.
To estimate $Z_{S,\sigma}$, we use an interaction expansion in
the Gaussian state proportional to $G\sigma G$, as described in
Appendix~\ref{sec:SM-oracle}. The normalization $\Tr(G\sigma G)$
and the two-point function of this state can be computed from
$N\times N$ transfer matrices.

Each soft pinning factor $\one\pm t\ii\gamma_a\gamma_b$ in $\sigma$
involves a pair of labels $\{a,b\}$, and different factors share
no labels. The two labels in a pair may be far apart, so spatial
decay of the two-point function cannot be assumed. For sufficiently
small $t$, the disjointness of these pairs gives uniform bounds
on the absolute row and column sums of the two-point function.
These bounds are independent of the number of pairs and the
distances between their labels.

The coefficients in the interaction expansion of $\log Z_{S,\sigma}$
are expressed through connected cumulants in this Gaussian state.
Wick's theorem writes the Gaussian moments as Pfaffians, and the
BKAR forest formula expresses each connected coefficient as a sum
over spanning trees. The covariance entries on the tree edges are
summed using the row and column bounds, while the remaining entries
stay inside a single Pfaffian of absolute value at most one.
Consequently,
\begin{equation}
|c_s|
\le
C_{\beta,D,r_0}Ng\,
\bigl(C'_{\beta,D,r_0}g\bigr)^{s-1},
\qquad
C'_{\beta,D,r_0}g<1.
\label{eq:cluster}
\end{equation}
Truncating the series at order $\order(\log(N/\eta))$ gives an
additive error of at most $\eta$ in $\log Z_{S,\sigma}$.
An efficient estimator for the truncated series is constructed
in Appendix~\ref{sec:SM-oracle}.

The sampling rule for $M_{S,i}$ involves continuous time variables
and paths of unbounded length. We truncate the paths using their
exponential tail and discretize the time variables, as described in
Appendix~\ref{sec:SM-finite}. The accumulated error remains small
over the $\order(n^2)$ levels of the tree.

Evaluating all child weights at every step is still too expensive.
We therefore use a lazy Metropolis chain on the finite tree. The
trace oracle estimates the expected leaf trace only at visited
nodes. A downward proposal samples a child from the original
child distribution, while an upward proposal selects the parent.
For a proposed move $v\to w$ between adjacent nodes, the acceptance
probability is
\begin{equation}
\min\left\{1,\frac{\varphi_w}{\varphi_v}\right\}.
\label{eq:strategy-Metropolis}
\end{equation}
After mixing, conditioning on the leaves gives a distribution
close to the target. The appendix bounds the errors from
truncation, trace estimation, and finite running time.

At the selected leaf $z$, the algorithm outputs
\begin{equation}
\tau_z
=
\frac{G\sigma_zG}{\Tr(G\sigma_zG)}.
\label{eq:tau}
\end{equation}
The algorithm achieves
$\norm{\bbE\tau-\rho_\beta}_1\le\epsilon$ in time polynomial in
$n$ and $\log(1/\epsilon)$, as proved in
Appendix~\ref{sec:SM-finite}.

\section{Discussion}
\label{sec:discussion}
\subsection{Relation to previous work}
Diagrammatic Monte Carlo expands around free fermions to estimate local observables, with polynomial cost in the inverse error under convergence and sampling assumptions~\cite{Rossi2017,RossiEtAl2017}. Using a different sampling approach, Ref.~\cite{ChenEtAl2025} rigorously established polynomial-time estimation of the log-partition function and local observables for weakly interacting fermions without Markov chain mixing assumptions. Our sampler produces Gaussian states whose average approximates the entire Gibbs state in trace norm, further enabling nonlocal correlations, occupation-number sampling, and nonlinear quantities of the reduced density matrix (see Cor.~\ref{cor:informal-physical-quantities}).

Fermionic tensor network methods numerically approximate thermal states of the two-dimensional Hubbard model~\cite{SinhaEtAl2022} and interacting spinless fermions~\cite{DeMeyerEtAl2026}. However, their accuracy is assessed through bond dimension convergence and observable benchmarks, without the rigorous runtime and global trace norm guarantees established here.

Gaussian quantum Monte Carlo uses nonnegative weights over Gaussian operators and stochastic differential equations to estimate thermal and dynamical observables, including for the Hubbard model~\cite{CorneyDrummond2004,CorneyDrummond2006}. However, its components need not be Hermitian or positive semidefinite, and these works do not establish a polynomial time guarantee. We prove a decomposition into physical Gaussian states together with efficient sampling for sufficiently weak interactions.

Ref.~\cite{RamkumarCaiTongJiang2026} established convex-Gaussianity and efficient sampling at high temperature on bounded-degree interaction graphs. We obtain both results at any fixed positive temperature for sufficiently weak nonquadratic interactions, requiring finite-range interactions on a finite-dimensional lattice.

\subsection{Outlook}

{ In this work, we established the convex-Gaussianity and classical simulability of weakly interacting fermions for interaction strengths below positive thresholds. These thresholds provide sufficient conditions for our main results and may decrease rapidly with $\beta$. A natural open question is whether these bounds can be further improved, or, more fundamentally, whether one can determine the optimal interaction strength thresholds for convex-Gaussianity and classical simulability. Our results suggest that convex-Gaussianity can support classical
simulation beyond Gaussian states. However, this property alone
does not guarantee efficient sampling. States outside the set
may still be classically tractable.  In addition, our results concern geometrically local systems, and the sampling proof relies explicitly on this assumption.
Whether the results extend to long-range interactions remains open.

The current results focus on a certain microscopic Hamiltonian and establish simulability only within a restricted parameter regime. A possible direction is to incorporate renormalization methods to extend our approach to a broader parameter regime. Moreover, although the simulable regime considered here is continuously connected to the quadratic Hamiltonian and thus belongs to the same many-body phase, techniques that reorganize the original problem around a different reference quadratic Hamiltonian may enable the study of different phases. Exploring this possibility is another interesting direction for future work.}

\section*{Acknowledgments}
 Z.H., Y.Z., and Z.-W.L.\ are supported in part by NSFC under Grant No.~12475023, Dushi Program, and a startup funding from YMSC.

\section*{AI disclosure}
The authors set the overall problem and strategy, and generative AI tools provided substantial assistance in developing the proof details under extensive instructions from the authors. The authors take full responsibility for the results and presentation.

\appendix
\setcounter{section}{0}

\numberwithin{equation}{section}
\numberwithin{figure}{section}
\numberwithin{table}{section}
\numberwithin{theorem}{section}

\renewcommand{\theHequation}{\thesection.\arabic{equation}}
\renewcommand{\theHfigure}{\thesection.\arabic{figure}}
\renewcommand{\theHtable}{\thesection.\arabic{table}}
\renewcommand{\theHtheorem}{\thesection.\arabic{theorem}}

\section{Setting and main theorems}
\label{app:detailed-proof}
\label{sec:SM-model}

\subsection{Preliminaries}

Let $\gamma_1,\ldots,\gamma_N$, $N=2n$, be Majorana operators satisfying
\begin{equation}
 \gamma_a^\dagger=\gamma_a,
 \qquad
 \{\gamma_a,\gamma_b\}=2\delta_{ab}\one.
\end{equation}
For every even subset
$R=\{r_1<\cdots<r_{|R|}\}\subseteq[N]$, choose the canonical Hermitian
Majorana monomial
\begin{equation}
 P_R=\ii^{|R|/2}\gamma_{r_1}\cdots\gamma_{r_{|R|}},
 \qquad
 P_R^\dagger=P_R,
 \qquad
 P_R^2=\one.
 \label{eq:SM-PR}
\end{equation}
Products satisfy
\begin{equation}
 P_RP_T=\omega(R,T)P_{R\triangle T},
 \qquad
 |\omega(R,T)|=1.
 \label{eq:SM-clifford}
\end{equation}

The fermions live on a finite subset
$\Lambda\subset\mathbb Z^D$, with at most $\nu$ complex fermionic modes
per site, where $\nu$ is independent of the system size. We write
$\operatorname{dist}(a,b)$ for the lattice distance between the sites
supporting $\gamma_a$ and $\gamma_b$.

We consider
\begin{equation}
 H=H_0+V,
 \qquad
 H_0=\frac14\bm\gamma^{\mathsf T}K_0\bm\gamma,
 \label{eq:SM-H}
\end{equation}
where $K_0^{\mathsf T}=-K_0$ and $K_0^\dagger=K_0$.
The local decompositions below are supported on sets $B$ contained
in balls of radius at most $r_0$. Each $B$ contains at most $d$
Majorana modes, and each lattice site belongs to at most $q_B$ such
sets. The constants $d$ and $q_B$ are independent of the system size.

The quadratic part has a local decomposition
\begin{equation}
 H_0=\sum_B H_{0,B},
 \qquad
 \supp H_{0,B}\subseteq B,
 \qquad
 H_{0,B}\ \text{quadratic},
 \qquad
 \norm{H_{0,B}}\le1.
 \label{eq:SM-H0-local}
\end{equation}
The interaction is decomposed as
\begin{equation}
 V=\sum_B V_B,
 \qquad
 \supp V_B\subseteq B,
 \qquad
 \norm{V_B}\le g,
 \label{eq:SM-V}
\end{equation}
where $g\ge0$ denotes the local interaction strength.
Each $V_B$ is even, but particle number conservation is not assumed.
We omit additive constants in $H$, which do not affect $\rho_\beta$.
The normalization $\norm{H_{0,B}}\le1$ fixes the local energy scale.
A general bound $\norm{H_{0,B}}\le J_0$ can be reduced to this
convention by replacing $\beta$ with $\beta J_0$ and $g$ with
$g/J_0$.

The Majorana monomials are orthogonal with respect to the normalized trace:
\begin{equation}
 2^{-n}\Tr(P_RP_T)=\delta_{R,T}.
 \label{eq:SM-Majorana-orthogonality}
\end{equation}
Thus, if $O=\sum_R o_RP_R$, then
\begin{equation}
 |o_R|
 =
 \left|2^{-n}\Tr(P_RO)\right|
 \le\norm O.
 \label{eq:SM-coefficient-bound}
\end{equation}
The coefficient of $P_{\{a,b\}}=\ii\gamma_a\gamma_b$ in $H_0$ is
$-\ii(K_0)_{ab}/2$. For a fixed $a$, each local quadratic term
contributes at most $d-1$ such coefficients, each of absolute value
at most one, and at most $q_B$ local terms contain $a$. It follows that
\begin{equation}
 \max\left\{
 \norm{K_0}_{1\to1},
 \norm{K_0}_{\infty\to\infty}
 \right\}
 \le 2q_B(d-1)
 =:C_H.
 \label{eq:SM-K0-local-bound}
\end{equation}
Writing $V=\sum_Rv_RP_R$, the same coefficient bound gives
\begin{equation}
 \max_a\sum_{R\ni a}|v_R|
 \le q_B2^dg
 =:C_Vg.
 \label{eq:SM-V-local-bound}
\end{equation}

Let $\cG_n$ denote the set of even fermionic Gaussian states on $n$
modes. These states need not be full rank; in particular, pure
Gaussian states are included~\cite{Bravyi2005,SuraceTagliacozzo2022}.

\subsection{Main results}

Set
\begin{equation}
    I_\beta:=\int_0^{\beta/2}\ee^{dC_Hu}\,\dd u,
\end{equation}
\begin{equation}
    \eta_{\mathrm{mix}}:=
 \min\left\{
 \frac1d,
 \frac12\log\frac{13}{12}
 \right\},
 \qquad
 g_{\mathrm{mix}}
 :=
 \frac{\eta_{\mathrm{mix}}}{2^dC_VI_\beta}.
\end{equation}

\begin{theorem}
\label{thm:SM-mixture}
Fix $\beta>0$, $D$, and $r_0$.
For any system size and any Hamiltonian satisfying
Eqs.~\eqref{eq:SM-H}--\eqref{eq:SM-V}, if
$g\le g_{\mathrm{mix}}$, then its Gibbs state satisfies
\begin{equation}
 \rho_\beta
 =
 \frac{\ee^{-\beta H}}{\Tr\ee^{-\beta H}}
 \in\conv(\cG_n).
 \label{eq:SM-main-struct}
\end{equation}
\end{theorem}
We also construct a classical sampler for the Gibbs state.  The
sampler requires a stronger bound on $g$.

Let $t_0=t_0(\beta,D,r_0)>0$ and
$g_{\mathrm{log}}=g_{\mathrm{log}}(\beta,D,r_0)>0$ be the constants
from the internal log-weight estimate below.  Let $Q$ be the smallest
power of two satisfying
\begin{equation}
 Q\ge
 \max\left\{
 8,\frac{2}{\sqrt{t_0}}
 \right\},
 \qquad
 t=\frac4{Q^2}\le t_0.
 \label{eq:SM-Qt}
\end{equation}
Here $t$ is the pinning parameter.  Set
\begin{equation}
 \kappa_0:=\frac1{32\sqrt2},
 \qquad
 \eta_{\mathrm{pin}}
 :=
 \min\left\{
 \frac{Q^2}{4d},
 \frac12\log(1+\kappa_0)
 \right\},
 \qquad
 g_{\mathrm{pin}}
 :=
 \frac{\eta_{\mathrm{pin}}}{Q^dC_VI_\beta}.
 \label{eq:SM-gpin}
\end{equation}
The sampler applies when
\begin{equation}
 g\le g_{\mathrm{samp}},
 \qquad
 g_{\mathrm{samp}}
 :=
 \min\left\{
 g_{\mathrm{mix}},
 g_{\mathrm{log}},
 g_{\mathrm{pin}},
 1
 \right\}.
 \label{eq:SM-gsample}
\end{equation}

\begin{theorem}[Gibbs sampler]
\label{thm:SM-sampler}
Fix $\beta>0$, as well as $D$ and $r_0$.
Let $H$ be any Hamiltonian satisfying
Eqs.~\eqref{eq:SM-H}--\eqref{eq:SM-V} with
$g\le g_{\mathrm{samp}}$.

Suppose that the supports $B$ and the coefficients of the local terms
$H_{0,B}$ and $V_B$ in the Majorana basis are given explicitly.  We
assume that $\beta$ and these coefficients can be evaluated to
precision $2^{-p}$ in time polynomial in $n$ and $p$.  Then, for every
$\epsilon>0$, there is a randomized classical algorithm running in
\begin{equation}
 \poly\left(n,\log(1/\epsilon)\right)
\end{equation}
time that outputs a classical description of a normalized Gaussian
state $\tau\in\cG_n$ such that
\begin{equation}
 \norm{\bbE\tau-\rho_\beta}_1\le\epsilon.
 \label{eq:SM-main-alg}
\end{equation}
The expectation is over the randomness of the algorithm.
\end{theorem}

\subsection{Examples}

As a first example, consider the spin-$1/2$ Fermi--Hubbard model
\begin{align}
 H_0={}&-\sum_{\langle i,j\rangle,\sigma}
 t_{ij,\sigma}
 \left(c_{i\sigma}^\dagger c_{j\sigma}+\mathrm{H.c.}\right)
 -\sum_{i,\sigma}\mu_{i\sigma}n_{i\sigma},
 \nonumber\\
 V={}&g_{\mathrm H}\sum_i
 \left(n_{i\uparrow}-\frac12\right)
 \left(n_{i\downarrow}-\frac12\right).
 \label{eq:SM-Hubbard}
\end{align}
Each site has two fermionic modes, corresponding to four Majorana
operators.  The coefficients $t_{ij,\sigma}$ and $\mu_{i,\sigma}$ may
vary with the sites and bonds but are uniformly bounded independently of
the system size. Thus $H_0$ is a finite-range quadratic Hamiltonian
with bounded local norm.

Each term in $V$ is an even quartic operator supported on one site.
In the normalization $\norm{H_{0,B}}\le1$ used above, its local
interaction strength is \(g=\frac{|g_{\mathrm H}|}{4}.\)
Hence Theorem~\ref{thm:SM-mixture} shows that $\rho_\beta$ is a convex
mixture of fermionic Gaussian states whenever
$|g_{\mathrm H}|\le4g_{\mathrm{mix}}$.
Under the stronger condition
$|g_{\mathrm H}|\le4g_{\mathrm{samp}}$,
Theorem~\ref{thm:SM-sampler} provides a classical Gibbs sampler.

As a second example, consider a Bogoliubov--de Gennes Hamiltonian of
finite range,
\begin{equation}
 H_0=\sum_{i,j}h_{ij}c_i^\dagger c_j
 +\frac12\sum_{i,j}
 \left(
 \Delta_{ij}c_i^\dagger c_j^\dagger+\mathrm{H.c.}
 \right).
 \label{eq:SM-BdG}
\end{equation}
Both the hopping terms proportional to $h_{ij}$ and the pairing terms
proportional to $\Delta_{ij}$ are quadratic, although the latter need
not conserve particle number. 

We assume that $H_0=\sum_BH_{0,B}$ with
$\supp H_{0,B}\subseteq B$ and $\norm{H_{0,B}}\le1$.
The interaction $V=\sum_BV_B$ satisfies
$\supp V_B\subseteq B$ and $\norm{V_B}\le g$, and may contain even
local terms of degree four or higher. Under these assumptions, $\rho_\beta$ is a convex mixture of fermionic
Gaussian states when $g\le g_{\mathrm{mix}}$.
Under the stronger condition $g\le g_{\mathrm{samp}}$, the
corresponding Gaussian mixture can also be sampled classically.

\section{Proof of the Gaussian-mixture theorem}
\label{sec:SM-mixture-proof}

Our goal is to write the Gibbs state as a convex mixture of Gaussian
states.  To do this, we first separate the quadratic part of the
Hamiltonian in a way that leaves a positive operator in the middle.
Set $T=\beta/2$.  Then
\begin{equation}
 \ee^{-\beta H}=GXX^\dagger G,
 \qquad
 G=\ee^{-TH_0},
 \qquad
 X=\ee^{TH_0}\ee^{-TH}.
 \label{eq:SM-half-contour}
\end{equation}
Indeed, $GX=X^\dagger G=\ee^{-TH}$.

The factorization reduces the theorem to the positive middle operator
$XX^\dagger$: since $G$ is Gaussian, a convex decomposition of
$XX^\dagger$ into nonnegative Gaussian operators induces a convex
Gaussian decomposition of the Gibbs state.

\begin{lemma}
\label{lem:SM-Gaussian-sandwich}
Let $G$ be an invertible positive Gaussian operator, and suppose that
a nonzero positive operator $B$ can be written as
\begin{equation}
 B=\sum_z\pi_z\sigma_z,
 \qquad
 \pi_z\ge0,
 \qquad
 \sum_z\pi_z=1,
\end{equation}
where every $\sigma_z$ is a nonnegative Gaussian operator.  Then
\begin{equation}
 \frac{GBG}
 {\Tr(GBG)}
 \in\conv(\cG_n).
\end{equation}
\end{lemma}

{Since $X$ preserves fermion parity, $XX^\dagger$ is a real
linear combination of Hermitian even Majorana monomials.} {The following elementary calculation explains when such an
expansion gives a convex combination of nonnegative Gaussian
operators.  Consider} 
{
\begin{equation}
 Y=c_0\one+\sum_Rc_R\Gamma_R,
 \qquad
 c_0>0,
\end{equation}
}
{where every $\Gamma_R$ is a nonidentity Hermitian even
Majorana monomial and every $c_R$ is real.  If}
{
\begin{equation}
 \sum_R|c_R|\le c_0,
 \label{eq:SM-overview-global-condition}
\end{equation}
}
then
{
\begin{equation}
 Y
 =
 \left(c_0-\sum_R|c_R|\right)\one
 +
 \sum_R|c_R|
 \left(\one+\sgn(c_R)\Gamma_R\right).
 \label{eq:SM-overview-positive-mixture}
\end{equation}
}
{Since $\Gamma_R^2=\one$, every operator
$\one+\sgn(c_R)\Gamma_R$ is nonnegative.  Thus
Eq.~\eqref{eq:SM-overview-positive-mixture} is a nonnegative
combination of nonnegative operators.}

Each $\one\pm\Gamma_R$ is itself a convex combination of
nonnegative Gaussian operators.  To see this, write
{
\begin{equation}
 \Gamma_R=F_1\cdots F_\ell,
\end{equation}
}
{Here the $F_j$ are commuting Hermitian quadratic Majorana
monomials with disjoint supports.  For $s\in\{\pm1\}$,}
{
\begin{equation}
 \one+s\Gamma_R
 =
 2^{1-\ell}
 \sum_{\substack{\eta_1,\ldots,\eta_\ell\in\{\pm1\}\\
                  \eta_1\cdots\eta_\ell=s}}
 \prod_{j=1}^{\ell}
 \left(\one+\eta_jF_j\right).
 \label{eq:SM-overview-monomial-mixture}
\end{equation}
}
{Each factor $\one+\eta_jF_j$ is nonnegative and Gaussian.
Since the factors have disjoint supports, each product on the right is
also nonnegative and Gaussian.  Together with
Eq.~\eqref{eq:SM-overview-positive-mixture}, this expresses $Y$ as a
sum of nonnegative Gaussian operators with nonnegative coefficients.}

{However, the Majorana expansion of $XX^\dagger$ is not
guaranteed to satisfy
Eq.~\eqref{eq:SM-overview-global-condition}.  Instead, writing} 
{\
\begin{equation}
 A(u)=\sum_Rv_R(u)P_R,
\end{equation}
}
{Lemma~\ref{lem:SM-interaction-coefficients} gives a bound for
the terms containing each Majorana operator $\gamma_i$:}
{
\begin{equation}
 \max_i\sum_{R\ni i}|v_R(u)|2^{|R|}
 \le
 2^dC_Vg\,\ee^{dC_Hu}.
 \label{eq:SM-overview-local-coefficient}
\end{equation}
}

{For $S\subseteq[N]$, let $A^S(u)$ contain the monomials of
$A(u)$ supported in $S$, and let $X^S$ denote the corresponding
evolution.  Then}
{
\begin{equation}
 X^{[N]}=X,
 \qquad
 X^\varnothing=\one.
\end{equation}
}
{For $i\in S$, the difference
$A^S(u)-A^{S\setminus\{i\}}(u)$ consists precisely of the monomials
containing $\gamma_i$ and hence satisfies the coefficient bound in
Eq.~\eqref{eq:SM-overview-local-coefficient}.}  {Define the relative factor $M_{S,i}$ by}

{
\begin{equation}
 X^S=X^{S\setminus\{i\}}M_{S,i}.
 \label{eq:SM-overview-relative}
\end{equation}
}
{Although $M_{S,i}$ may contain Majorana monomials of high
degree, the next lemma bounds the weighted sum of its nonidentity
coefficients independently of $S$, $i$, and the system size.}
{
\begin{lemma}
\label{lem:SM-relative-coefficient}
If $g\le g_{\mathrm{mix}}$, there is a constant
$r\le1/12$, independent of $S$, $i$, and the system size, such that
for every $S\subseteq[N]$ and $i\in S$,
\begin{equation}
 M_{S,i}-\one=\sum_Rm_RP_R,
 \qquad
 M_{S,i}^\dagger-\one=\sum_R\widetilde m_RP_R,
\end{equation}
with
\begin{equation}
 \sum_R|m_R|2^{|R|}\le r,
 \qquad
 \sum_R|\widetilde m_R|2^{|R|}\le r.
 \label{eq:SM-overview-ratio-bound}
\end{equation}
\end{lemma}
}
{The bound $r\le1/12$ will ensure that the operators appearing
in the convex combination remain nonnegative.}
{After some Majorana operators have been removed, the
corresponding term has the form} 
{
\begin{equation}
 \mathcal E_S(C):=X^SC(X^S)^\dagger,
 \label{eq:SM-overview-partial}
\end{equation}
}
{The factors already separated from $X^S$ are collected in
$C$.  We start from}
\begin{equation}
 \mathcal E_{[N]}(\one)=XX^\dagger.
\end{equation}
{For $i\in S$, set $S'=S\setminus\{i\}$.  The relation
$X^S=X^{S'}M_{S,i}$ gives}
{
\begin{equation}
 \mathcal E_S(C)
 =
 X^{S'}M_{S,i}CM_{S,i}^\dagger(X^{S'})^\dagger.
\end{equation}
}
{The next subsection proves that the middle operator can be
written as an exact convex combination}
{
\begin{equation}
 M_{S,i}CM_{S,i}^\dagger
 =
 \sum_kp_kC_k,
 \qquad
 p_k\ge0,
 \qquad
 \sum_kp_k=1,
 \label{eq:SM-overview-local-replacement}
\end{equation}
}
where every $C_k$ is nonnegative and can again be written as 
$(\one+\alpha_k\Gamma_k)\sigma_k$ where $\sigma_k$ denotes a Gaussian state.  Hence
\begin{equation}
 \mathcal E_S(C)
 =
 \sum_kp_k\mathcal E_{S'}(C_k).
\end{equation}

Repeating the above identity until $S=\varnothing$ writes
$XX^\dagger$ as a convex combination of the final operators $C$.
Any remaining factor $\one+\alpha\Gamma$ is also written as a convex
combination of products of commuting quadratic factors.  Hence the
final operators are nonnegative and Gaussian, and
\begin{equation}
XX^\dagger
=
\sum_z\pi_z\sigma_z,
\qquad
\pi_z\ge0,
\qquad
\sum_z\pi_z=1,
\end{equation}
with each $\sigma_z$ Gaussian and nonnegative.

Proposition~\ref{prop:SM-hard-decomposition} derives the Gaussian
mixture from the coefficient bound in
Lemma~\ref{lem:SM-relative-coefficient}.  We prove that bound by
combining the local estimate for $A(u)$ with a commutator expansion.
The Gaussian-sandwich lemma then transfers the decomposition of
$XX^\dagger$ to the normalized Gibbs state.

\subsection{\texorpdfstring{$XX^\dagger$ is a Gaussian mixture}{A Gaussian mixture for XXdagger}}

\begin{proposition}
\label{prop:SM-hard-decomposition}
If $g\le g_{\mathrm{mix}}$, then $XX^\dagger$ is a convex combination
of nonnegative Gaussian operators.
\end{proposition}

\begin{proof}

Assume Lemma~\ref{lem:SM-relative-coefficient}.  We recursively remove
the interaction terms from $X^S$ while preserving a nonnegative middle
operator.  This will give
\begin{equation}
 XX^\dagger=\bbE_z\,\sigma_z,
 \label{eq:SM-hard-goal}
\end{equation}
for a distribution over nonnegative Gaussian operators $\sigma_z$.

For $S\subseteq[N]$, write
\begin{equation}
 \mathcal E_S(C)=X^SC(X^S)^\dagger,
 \label{eq:SM-hard-partial}
\end{equation}
and maintain
\begin{equation}
 C=(\one+\alpha\Gamma)\sigma,
 \label{eq:SM-hard-state}
\end{equation}
where $\alpha\ge0$ and $\Gamma$ is a Hermitian even Majorana monomial.
The Gaussian part has the form
\begin{equation}
 \sigma
 =
 c\prod_{e=(a,b)\in\cM}
 \left(\one+s_e\ii\gamma_a\gamma_b\right),
 \qquad
 c\ge0,
 \qquad
 s_e\in\{\pm1\},
 \label{eq:SM-hard-sigma}
\end{equation}
where $\cM$ is a collection of disjoint pairs in $S^c$, and $\Gamma$
is disjoint from these pairs.  The quadratic factors are commuting and
nonnegative, so $\sigma$ is a nonnegative Gaussian operator.  Moreover,
$\sigma$ commutes with $M_{S,i}$ because their supports lie in $S^c$
and $S$, respectively.

At the root, $S=[N]$ and $C=\one$, so
\begin{equation}
 \mathcal E_{[N]}(\one)=XX^\dagger.
\end{equation}

{For $i\in S$, set}
\begin{equation}
 S'=S\setminus\{i\},
 \qquad
 M=(X^{S'})^{-1}X^S.
\end{equation}

The identity $X^S=X^{S'}M$ confines the difference between
the two partial evolutions to the middle operator:
\begin{equation}
 \mathcal E_S(C)
 =
 X^{S'}MCM^\dagger(X^{S'})^\dagger.
\end{equation}
It therefore suffices to write
\begin{equation}
 MCM^\dagger=\sum_kp_kC_k,
 \qquad
 p_k\ge0,
 \qquad
 \sum_kp_k=1,
 \label{eq:SM-hard-local-replacement}
\end{equation}
where every $C_k$ has the form in
Eq.~\eqref{eq:SM-hard-state}.  Indeed,
\begin{equation}
 \mathcal E_S(C)
 =\sum_kp_k\mathcal E_{S'}(C_k).
\end{equation}
Thus one term in the convex combination is replaced by another convex
combination with the same value.

We now construct the local convex identity in
Eq.~\eqref{eq:SM-hard-local-replacement} while keeping every $C_k$
nonnegative and in the form required for the next removal.

Since $M$ is supported in $S$, it commutes with $\sigma$.  Hence
\begin{equation}
 MCM^\dagger
 =
 M(\one+\alpha\Gamma)M^\dagger\sigma.
\end{equation}

The coefficient bound represents the nonidentity part of $M$
as the expectation of a single Majorana monomial.  Write
\begin{equation}
 M-\one=\sum_Rm_RP_R,
 \qquad
 \sum_R|m_R|2^{|R|}\le r.
\end{equation}
For every $R$ with $m_R\ne0$, set
\begin{equation}
 p_R=\frac{|m_R|2^{|R|}}{r}.
\end{equation}
Since $\sum_Rp_R\le1$, define a random variable $Z$ as follows.  With
probability $p_R$, set
\begin{equation}
 Z=\frac{m_R}{|m_R|}\,r2^{-|R|}P_R,
\end{equation}
and set $Z=0$ with the remaining probability.  Its expectation is
\begin{align}
 \bbE Z
 &=
 \sum_R
 \frac{|m_R|2^{|R|}}{r}
 \frac{m_R}{|m_R|}\,
 r2^{-|R|}P_R
 \nonumber\\
 &=\sum_Rm_RP_R
 =M-\one.
\end{align}
Thus the coefficients are redistributed among the terms of
the average, but none of them is discarded.  Applying the same
definition independently to $M^\dagger-\one$ gives random variables
$Z_1$ and $Z_2$ satisfying
\begin{equation}
 \bbE(\one+Z_1)=M,
 \qquad
 \bbE(\one+Z_2)=M^\dagger.
\end{equation}
Then
\begin{equation}
 M(\one+\alpha\Gamma)M^\dagger
 =
 \bbE_{Z_1,Z_2}
 \left[
 (\one+Z_1)(\one+\alpha\Gamma)(\one+Z_2)
 \right].
 \label{eq:SM-hard-random-product}
\end{equation}
For fixed $Z_1$ and $Z_2$, write
\begin{equation}
 (\one+Z_1)(\one+\alpha\Gamma)(\one+Z_2)
 =
 \one+\sum_{k=1}^7T_k,
\end{equation}
where
\begin{align}
 T_1&=\alpha\Gamma,
 &
 T_2&=Z_1,
 &
 T_3&=Z_2,
 \nonumber\\
 T_4&=Z_1\alpha\Gamma,
 &
 T_5&=\alpha\Gamma Z_2,
 &
 T_6&=Z_1Z_2,
 &
 T_7&=Z_1\alpha\Gamma Z_2.
 \label{eq:SM-hard-seven-terms}
\end{align}
Choose probabilities $p_1,\ldots,p_7>0$ satisfying
$\sum_{k=1}^7p_k=1$, and define
\begin{equation}
 H_k=\operatorname{Herm}(p_k^{-1}T_k),
 \qquad
 \operatorname{Herm}(Y)=\frac{Y+Y^\dagger}{2}.
 \label{eq:SM-hard-HK}
\end{equation}
The factor $p_k^{-1}$ compensates for selecting the $k$th term with
probability $p_k$. Each $H_k$ is either zero or a real multiple of a Hermitian Majorana monomial.

For each $k$, set
\begin{equation}
 C_k=(\one+H_k)\sigma.
\end{equation}
For fixed $Z_1$ and $Z_2$, these operators satisfy
\begin{align}
 \sum_{k=1}^7p_kC_k
 &=
 \left[
 \one+\sum_{k=1}^7\operatorname{Herm}(T_k)
 \right]\sigma
 \nonumber\\
 &=
 \operatorname{Herm}\left[
 (\one+Z_1)(\one+\alpha\Gamma)(\one+Z_2)
 \right]\sigma.
 \label{eq:SM-hard-seven-mean}
\end{align}
Averaging over the independent variables $Z_1$ and $Z_2$ gives
\begin{align}
 \bbE_{Z_1,Z_2}
 \left[
 \sum_{k=1}^7p_kC_k
 \right]
 &=
 \operatorname{Herm}\left[
 M(\one+\alpha\Gamma)M^\dagger
 \right]\sigma
 \nonumber\\
 &=
 M(\one+\alpha\Gamma)M^\dagger\sigma
 \nonumber\\
 &=MCM^\dagger.
 \label{eq:SM-hard-child-mean}
\end{align}

The identity above holds for any positive probabilities
$p_1,\ldots,p_7$ whose sum is one.  The individual operators $C_k$
need not be nonnegative for an arbitrary choice, because the
coefficient of $T_k$ in $H_k$ is multiplied by $p_k^{-1}$.  The
following lemma chooses the probabilities and the removed label so
that every $C_k$ remains nonnegative and can be used at the next
removal.

For $\alpha>0$, let
$m=|\supp\Gamma\cap S|$.  The condition used throughout the removal is
$\alpha\le2^{-m}$.  It implies $\alpha\le1$, which makes
$(\one+\alpha\Gamma)\sigma$ nonnegative.  Its dependence on $m$ is
needed because a selected monomial of degree $s$ can introduce at most
$s$ labels into $\supp\Gamma\cap S$, while its coefficient contains the
factor $2^{-s}$.

There are three cases.  If $\alpha=0$, there is no remaining
monomial.  If $\alpha>0$ and $m>0$, we remove a label in
$\supp\Gamma\cap S$.  If $\alpha>0$ and $m=0$, then $\Gamma$ is
supported in $S^c$ and we first separate quadratic factors from it.
The following lemma handles the first two cases.

\begin{lemma}
\label{lem:SM-hard-invariant}
Let $C=(\one+\alpha\Gamma)\sigma$.  When $\alpha>0$, let
$m=|\supp\Gamma\cap S|$ and assume that
\begin{equation}
 m>0,
 \qquad
 \alpha\le2^{-m}.
 \label{eq:SM-hard-invariant}
\end{equation}
No condition on $\Gamma$ is needed when $\alpha=0$.

Choose $i\in\supp\Gamma\cap S$ when $\alpha>0$, and choose any
$i\in S$ when $\alpha=0$.  Set $S'=S\setminus\{i\}$ and take
\begin{equation}
 p_1=\frac12,
 \qquad
 p_2=\cdots=p_7=\frac1{12}.
 \label{eq:SM-hard-probabilities}
\end{equation}

For every realization of $Z_1,Z_2$ and every $k$, the operator $C_k$
has the form required above with $S$ replaced by $S'$.  If $H_k$ is
nonzero, write $H_k=\alpha_k\Gamma_k$, where $\alpha_k\ge0$, and let
$m_k=|\supp\Gamma_k\cap S'|$.  Then
\begin{equation}
 \alpha_k\le2^{-m_k}.
 \label{eq:SM-hard-child-invariant}
\end{equation}
If $H_k=0$, then $C_k=\sigma$ is nonnegative.  Otherwise,
Eq.~\eqref{eq:SM-hard-child-invariant} gives $\alpha_k\le1$.
Since $\Gamma_k^2=\one$ and $\Gamma_k$ commutes with $\sigma$, this
also implies that $C_k=(\one+\alpha_k\Gamma_k)\sigma$ is nonnegative.
\end{lemma}
\begin{proof}
Every Majorana monomial in $Z_1$ or $Z_2$ is supported in $S$.
The quadratic factors in $\sigma$ are supported in $S^c$, and
$\Gamma$ is disjoint from these factors.  Hence every $H_k$ commutes
with $\sigma$ and is disjoint from its quadratic factors.  Moreover,
$S^c\subseteq(S')^c$, so $\sigma$ has the required form after $S$ is
replaced by $S'$.

For a nonzero $Z_j$, write
\begin{equation}
 Z_j=\zeta_j r2^{-s_j}P_{R_j},
 \qquad
 s_j=|R_j|,
 \qquad
 |\zeta_j|=1.
\end{equation}
We use $r\le1/12$.  Taking the Hermitian part replaces the coefficient
of a Majorana monomial by its real part and therefore cannot increase
its absolute value.  It is thus enough to bound the coefficients of
$p_k^{-1}T_k$.

Suppose first that $\alpha>0$.  For $T_1=\alpha\Gamma$, we have
\begin{equation}
 H_1=2\alpha\Gamma.
\end{equation}
The chosen label $i$ belongs to $\supp\Gamma\cap S$.  Hence $\Gamma$
has $m-1$ labels in $S'$, and
\begin{equation}
 2\alpha
 \le
 2^{-(m-1)}.
\end{equation}
Thus Eq.~\eqref{eq:SM-hard-child-invariant} holds for $k=1$.

Consider next $T_2$ and $T_3$.  If the selected monomial has degree
$s$, its coefficient after division by $p_k=1/12$ is at most
\begin{equation}
 12r2^{-s}\le2^{-s}.
\end{equation}
Its support contains at most $s$ labels in $S'$, so
$m_k\le s$ and
\begin{equation}
 \alpha_k\le2^{-s}\le2^{-m_k}.
\end{equation}

For $T_4$ and $T_5$, the coefficient is at most
\begin{equation}
 12r\alpha2^{-s}\le\alpha2^{-s}.
\end{equation}
After $i$ is removed, $\Gamma$ has $m-1$ labels in $S'$.  Multiplication
by a monomial of degree $s$ gives
\begin{equation}
 m_k\le m-1+s.
\end{equation}
Using $\alpha\le2^{-m}$, we obtain
\begin{equation}
 \alpha_k
 \le
 \alpha2^{-s}
 \le
 2^{-(m+s)}
 \le
 2^{-m_k}.
\end{equation}

For $T_6$, let $s_1$ and $s_2$ be the degrees of the monomials in
$Z_1$ and $Z_2$.  Its coefficient is at most
\begin{equation}
 12r^2 2^{-(s_1+s_2)}
 \le
 2^{-(s_1+s_2)}.
\end{equation}
Since $m_k\le s_1+s_2$, this gives
$\alpha_k\le2^{-m_k}$.  The same argument for $T_7$ gives
\begin{equation}
 \alpha_k
 \le
 12r^2\alpha2^{-(s_1+s_2)}
 \le
 \alpha2^{-(s_1+s_2)}
 \le
 2^{-m_k},
\end{equation}
where
\begin{equation}
 m_k\le m-1+s_1+s_2.
\end{equation}

If $\alpha=0$, the terms $T_1,T_4,T_5$, and $T_7$ vanish.  The
arguments for $T_2,T_3$, and $T_6$ do not depend on the choice of
$i$, and therefore remain valid for any $i\in S$.

We have proved Eq.~\eqref{eq:SM-hard-child-invariant} in every case.
For a nonzero $H_k$, it gives $\alpha_k\le2^{-m_k}\le1$.  Hence
$\one+\alpha_k\Gamma_k$ is nonnegative.  This factor commutes with the
nonnegative operator $\sigma$, so $C_k$ is nonnegative as well.
\end{proof}
The lemma removes labels from $\supp\Gamma\cap S$.  Once this set is
empty, $\Gamma$ is supported in $S^c$.  If $\alpha>0$ and $\Gamma$ is
not a scalar, choose
two labels $a,b\in\supp\Gamma$ and write
\begin{equation}
 \Gamma=\Gamma^-Q,
 \qquad
 Q=\ii\gamma_a\gamma_b,
\end{equation}
where $\Gamma^-$ and $Q$ are disjoint commuting Hermitian monomials.
The following identity is the one-pair form of
Eq.~\eqref{eq:SM-overview-monomial-mixture}.
We use the identity
\begin{equation}
 \one+\alpha\Gamma^-Q
 =
 \frac12(\one+\alpha\Gamma^-)(\one+Q)
 +
 \frac12(\one-\alpha\Gamma^-)(\one-Q).
 \label{eq:SM-hard-split}
\end{equation}
Since $m=0$, the condition $\alpha\le2^{-m}$ gives
$\alpha\le1$.  The four factors on the right hand side are therefore
nonnegative.  Each product is also nonnegative because its two factors
commute.

Using Eq.~\eqref{eq:SM-hard-split}, replace $C$ by either
\begin{align}
 C_+&=(\one+\alpha\Gamma^-)(\one+Q)\sigma,
 \nonumber\\
 C_-&=(\one-\alpha\Gamma^-)(\one-Q)\sigma,
\end{align}
each with probability $1/2$.  Their average is $C$.  In the first
case, include $\one+Q$ in $\sigma$ and replace $\Gamma$ by
$\Gamma^-$.  In the second case, include $\one-Q$ in $\sigma$ and
replace $\Gamma$ by $-\Gamma^-$.  The coefficient $\alpha$ is
unchanged.

The labels $a$ and $b$ lie in $S^c$ and do not occur in
$\Gamma^-$.  Thus the new quadratic factor is disjoint from the
factors already present in $\sigma$, and the remaining Majorana
monomial has degree two less than $\Gamma$.  Repeating this procedure
eventually reduces $\Gamma$ to $\pm\one$.  The remaining factor is then
the nonnegative scalar $1\pm\alpha$, which is absorbed into $c$.

Starting from $S=[N]$ and $C=\one$, so that
$\mathcal E_{[N]}(\one)=XX^\dagger$, we apply the two identities above
until no label remains in $S$.  When $\alpha=0$, remove any
$i\in S$.  When $\alpha>0$ and
$\supp\Gamma\cap S\ne\varnothing$, remove a label in this
intersection and apply Lemma~\ref{lem:SM-hard-invariant}.  When
$\alpha>0$ and $\supp\Gamma\cap S=\varnothing$, first use
Eq.~\eqref{eq:SM-hard-split} to separate all remaining quadratic
factors from $\Gamma$, and then continue removing labels.
A label removal decreases $|S|$ by one, while each use of
Eq.~\eqref{eq:SM-hard-split} decreases the degree of $\Gamma$ by two.
Moreover, every quadratic factor separated from $\Gamma$ is supported
on a pair of labels that does not occur in $\sigma$ or in the remaining
monomial.  Thus all labels are removed after finitely many steps, and
the final operator has the form $C=\sigma$ with $S=\varnothing$.

For a label removal, we have already shown that
\begin{equation}
 \mathcal E_S(C)
 =\sum_kp_k\mathcal E_{S'}(C_k).
\end{equation}
For a quadratic separation, Eq.~\eqref{eq:SM-hard-split} gives
\begin{equation}
 C=\frac12C_++\frac12C_-,
\end{equation}
and hence
\begin{equation}
 \mathcal E_S(C)
 =
 \frac12\mathcal E_S(C_+)
 +\frac12\mathcal E_S(C_-).
\end{equation}
Applying these identities until $S=\varnothing$ expresses
$\mathcal E_{[N]}(\one)=XX^\dagger$ as
\begin{equation}
 XX^\dagger
 =
 \sum_z\pi_z\sigma_z,
 \qquad
 \pi_z\ge0,
 \qquad
 \sum_z\pi_z=1.
 \label{eq:SM-hard-mean}
\end{equation}
Every nonzero $\sigma_z$ is a nonnegative product of commuting
quadratic factors and is therefore Gaussian.  This proves the desired
convex decomposition of $XX^\dagger$.
\end{proof}

We have proved Proposition~\ref{prop:SM-hard-decomposition}
assuming Lemma~\ref{lem:SM-relative-coefficient}.  The next two
subsections prove this lemma, beginning with the coefficient bound for
$A(u)$.

\subsection{\texorpdfstring{{Coefficient bound for $A(u)$}}{Coefficient bound for A(u)}}
 Define
\begin{equation}
 X(u)=\ee^{uH_0}\ee^{-uH},
 \qquad
 A(u)=\ee^{uH_0}V\ee^{-uH_0},
 \qquad
 0\le u\le T.
 \label{eq:SM-Xu-Au}
\end{equation}
When $V=0$, one has $X(u)=\one$, so $X(u)$ describes the change in
the imaginary-time evolution due to the interaction.  Direct
differentiation gives
\begin{equation}
 \frac{\dd}{\dd u}X(u)=-A(u)X(u),
 \qquad
 X(0)=\one.
\end{equation}
Therefore,
\begin{equation}
 X=X(T)
 =
 \cT\exp\left[-\int_0^T A(u)\,\dd u\right].
 \label{eq:SM-A}
\end{equation}

\par
Quadratic conjugation generally spreads a local interaction
monomial across the system.  The coefficient weight through any fixed
Majorana operator nevertheless grows by a factor independent of the
system size.  We begin with the evolution of a single Majorana
operator. Since $H_0$ is quadratic,
\begin{equation}
 [H_0,\gamma_a]
 =
 -\sum_b(K_0)_{ab}\gamma_b,
\end{equation}
and therefore
\begin{equation}
 \ee^{uH_0}\gamma_a\ee^{-uH_0}
 =
 \sum_bQ(u)_{ab}\gamma_b,
 \qquad
 Q(u)=\ee^{-uK_0}.
 \label{eq:SM-Q}
\end{equation}
Thus, conjugation by $\ee^{uH_0}$ may spread a Majorana monomial over
the lattice, but it does not change its degree. Since $Q(u)=\ee^{-uK_0}$, the submultiplicativity of the induced norms
gives
\begin{equation}
 \max\left\{
 \norm{Q(u)}_{1\to1},
 \norm{Q(u)}_{\infty\to\infty}
 \right\}
 \le \ee^{C_Hu}.
 \label{eq:SM-Qbound}
\end{equation}

The bound above describes the spreading of a single Majorana operator.
The row sums control the expansion of a fixed $\gamma_a$,
while the column sums enter when the terms containing a fixed
$\gamma_b$ are collected.
To apply it to the Majorana expansion of $A(u)$, we use the following
weighted coefficient sums.  For $q\ge1$ and an operator
$Y=\sum_R y_RP_R$, define
\begin{equation}
 L_q(Y)
 :=
 \sum_R|y_R|q^{|R|},
 \qquad
 J_q(Y)
 :=
 \max_a\sum_{R\ni a}|y_R|q^{|R|}.
 \label{eq:SM-Lq}
\end{equation}
\begin{lemma}
\label{lem:SM-interaction-coefficients}
For every $q\ge1$ and $0\le u\le T$,
\begin{equation}
 J_q(A(u))
 \le
 q^dC_Vg\,\ee^{dC_Hu}.
 \label{eq:SM-Jbound}
\end{equation}
In particular, the bound is independent of the system size.
\end{lemma}

\begin{proof}
To estimate $J_q(A(u))$, consider first the evolution of a single
Majorana monomial.  Since conjugation by $\ee^{uH_0}$ acts linearly
on the Majorana operators,
\begin{equation}
 \ee^{uH_0}P_R\ee^{-uH_0}
 =
 \sum_{\substack{S\subseteq[N]\\ |S|=|R|}}
 \det Q(u)_{R,S}\,P_S,
 \label{eq:SM-exterior-power}
\end{equation}
where $Q(u)_{R,S}$ denotes the corresponding submatrix of $Q(u)$.
Thus, the evolution may change the support of a Majorana monomial,
but not its degree.

We next bound the coefficients of $A(u)$.  Expanding the determinant
in Eq.~\eqref{eq:SM-exterior-power}, we have
\begin{equation}
 \det Q(u)_{R,S}
 =
 \sum_{\pi:R\to S}
 \sgn(\pi)\prod_{a\in R}Q(u)_{a,\pi(a)},
 \label{eq:SM-determinant-expansion}
\end{equation}
where the sum is over bijections from $R$ to $S$.  Fix an output
Majorana operator $\gamma_b$.  In every bijection
$\pi:R\to S$ with $b\in S$, there is a unique $a\in R$ such that
$\pi(a)=b$.  We first choose this $a$ and then sum over the possible
outputs of the remaining elements of $R$.  Allowing these remaining
outputs to be chosen independently gives the upper bound
\begin{equation}
 \sum_{\substack{S\ni b\\ |S|=|R|}}
 \left|\det Q(u)_{R,S}\right|
 \le
 \sum_{a\in R}
 \left[
 |Q(u)_{ab}|
 \prod_{c\in R\setminus\{a\}}
 \left(
 \sum_j|Q(u)_{cj}|
 \right)
 \right].
 \label{eq:SM-minor-bound}
\end{equation}
The right-hand side may also count choices in which two elements of
$R$ have the same output, but this only enlarges the sum.

Write the Majorana expansion of $A(u)$ as
\begin{equation}
 A(u)=\sum_Sv_S(u)P_S.
\end{equation}
Equation~\eqref{eq:SM-exterior-power} gives
\begin{equation}
 v_S(u)
 =
 \sum_{\substack{R\subseteq[N]\\ |R|=|S|}}
 v_R\det Q(u)_{R,S}.
\end{equation}
Using the triangle inequality and
Eq.~\eqref{eq:SM-minor-bound}, we obtain
\begin{align}
 J_q(A(u))
 &\le
 \max_b
 \sum_R |v_R|q^{|R|}
 \sum_{a\in R}
 \left[
 |Q(u)_{ab}|
 \prod_{c\in R\setminus\{a\}}
 \left(
 \sum_j|Q(u)_{cj}|
 \right)
 \right].
 \label{eq:SM-Jbound-intermediate}
\end{align}

Every monomial in $V$ has degree at most $d$.  Hence
$q^{|R|}\le q^d$, and Eq.~\eqref{eq:SM-Qbound} bounds each of the
$|R|-1$ row sums by $\ee^{C_Hu}$.  It follows that
\begin{align}
 J_q(A(u))
 &\le
 q^d\ee^{(d-1)C_Hu}
 \max_b
 \left\{
 \sum_a |Q(u)_{ab}|
 \left(
 \sum_{R\ni a}|v_R|
 \right)
 \right\}
 \nonumber\\
 &\le
 q^dC_Vg\,\ee^{dC_Hu}.
\end{align}
\end{proof}

\subsection{\texorpdfstring{{Coefficient bound for $M_{S,i}-\one$}}{Coefficient bound for M(S,i)-1}}

For $S\subseteq[N]$, {write
$A(u)=\sum_Rv_R(u)P_R$ and define}
\begin{equation}
 A^S(u):=\sum_{R\subseteq S}v_R(u)P_R,
\end{equation}
{so that $A^S(u)$ contains exactly the monomials supported in
$S$.  Define}
\begin{equation}
 X^S(u)
 :=
 \cT\exp\left[-\int_0^u A^S(v)\,\dd v\right],
 \qquad
 X^S:=X^S(T).
 \label{eq:SM-XS}
\end{equation}
At the two endpoints,
\begin{equation}
 X^{[N]}=X,
 \qquad
 X^\varnothing=\one.
\end{equation}

For $i\in S$, set $S'=S\setminus\{i\}$.  The interaction terms present in $A^S(u)$ but not in $A^{S'}(u)$ are
\begin{equation}
 D_{S,i}(u)
 :=
 A^S(u)-A^{S'}(u)
 =
 \sum_{\substack{R\subseteq S\\ i\in R}}
 v_R(u)P_R.
 \label{eq:SM-DSi}
\end{equation}
We describe the corresponding change in the evolution by
\begin{equation}
 M_{S,i}(u)
 :=
 \bigl(X^{S'}(u)\bigr)^{-1}X^S(u).
 \label{eq:SM-Mu}
\end{equation}
Then
\begin{equation}
 M_{S,i}(0)=\one,
 \qquad
 M_{S,i}:=M_{S,i}(T),
 \qquad
 X^S=X^{S'}M_{S,i}.
 \label{eq:SM-M}
\end{equation}
We will derive a bound on $L_q(M_{S,i}-\one)$ that is independent of
$S$, $i$, and the system size.

To state the bound, set
\begin{equation}
 \eta_q
 :=
 \int_0^T J_q(A(u))\,\dd u.
 \label{eq:SM-etaq}
\end{equation}
Lemma~\ref{lem:SM-interaction-coefficients} gives
\begin{equation}
 \eta_q\le q^dC_VgI_\beta.
 \label{eq:SM-etaq-bound}
\end{equation}

Conjugation by $X^{S'}(u)$ produces nested commutators.  We
use the following bound to control their coefficients.
\begin{lemma}
\label{lem:SM-comm}
Let $q\ge1$.  Suppose that $A$ and $Y$ are even and that every
Majorana monomial in $Y$ has degree at most $k$.  Then
\begin{equation}
 L_q([A,Y])
 \le
 \frac{2k}{q^2}J_q(A)L_q(Y).
 \label{eq:SM-comm}
\end{equation}
\end{lemma}

\begin{proof}
Write
\begin{equation}
 A=\sum_Ra_RP_R,
 \qquad
 Y=\sum_Ty_TP_T.
\end{equation}
For even sets $R$ and $T$, the commutator $[P_R,P_T]$ vanishes unless
$|R\cap T|$ is odd.  In particular, a nonzero commutator requires
$R\cap T\ne\varnothing$.  Since
\begin{equation}
 P_RP_T=\omega(R,T)P_{R\triangle T},
\end{equation}
we then have
\begin{equation}
 q^{|R\triangle T|}
 \le
 q^{|R|+|T|-2}.
\end{equation}
It follows that
\begin{align}
 L_q([A,Y])
 &\le
 \frac{2}{q^2}
 \sum_T |y_T|q^{|T|}
 \sum_{a\in T}
 \sum_{R\ni a}|a_R|q^{|R|}
 \nonumber\\
 &\le
 \frac{2k}{q^2}J_q(A)L_q(Y).
\end{align}
\end{proof}

{The resulting bound for the relative factor is as follows.}
\begin{lemma}
\label{lem:SM-ratio}
If \(\frac{2d\eta_q}{q^2}<1\),
then, for every $S\subseteq[N]$ and $i\in S$,
\begin{equation}
 L_q(M_{S,i}-\one)
 \le
 \exp\left(
 \frac{\eta_q}{1-2d\eta_q/q^2}
 \right)-1.
 \label{eq:SM-ratio-est}
\end{equation}
The same bound holds for $M_{S,i}^\dagger-\one$.
\end{lemma}

\begin{proof}
Since
\begin{equation}
 \frac{\dd}{\dd u}X^S(u)
 =
 -A^S(u)X^S(u),
 \qquad
 \frac{\dd}{\dd u}X^{S'}(u)
 =
 -A^{S'}(u)X^{S'}(u),
\end{equation}
direct differentiation gives
\begin{equation}
 \frac{\dd}{\dd u}M_{S,i}(u)
 =
 -B_{S,i}(u)M_{S,i}(u),
 \qquad
 M_{S,i}(0)=\one,
 \label{eq:SM-M-differential}
\end{equation}
where
\begin{equation}
 B_{S,i}(u)
 :=
 \bigl(X^{S'}(u)\bigr)^{-1}
 D_{S,i}(u)
 X^{S'}(u),
 \qquad
 D_{S,i}(u)
 :=
 A^S(u)-A^{S'}(u).
 \label{eq:SM-B}
\end{equation}

The operator $D_{S,i}(u)$ consists of the monomials in $A^S(u)$
that contain $\gamma_i$.  Hence
\begin{equation}
 L_q(D_{S,i}(u))=
 \sum_{\substack{R\subseteq S\\i\in R}}
 |v_R(u)|q^{|R|}\le
 \sum_{R\ni i}|v_R(u)|q^{|R|}
 \le
 J_q(A(u)).
\end{equation}

It remains to bound $B_{S,i}(u)$.  For fixed $u$, expanding the
conjugation by $X^{S'}(u)$ gives
\begin{align}
 B_{S,i}(u)
 ={}&
 D_{S,i}(u)
 \nonumber\\
 &+
 \sum_{m\ge1}
 \int_{0<t_1<\cdots<t_m<u}
 \Bigl[
 A^{S'}(t_1),
 \Bigl[
 A^{S'}(t_2),
 \ldots,
 \Bigl[
 A^{S'}(t_m),
 D_{S,i}(u)
 \Bigr]
 \ldots
 \Bigr]
 \Bigr]
 \,\dd\bm t.
 \label{eq:SM-B-nested}
\end{align}
For the estimates below, we write
$\operatorname{ad}_A(Y):=[A,Y]$; the nested commutator in
Eq.~\eqref{eq:SM-B-nested} is then
\[
 \operatorname{ad}_{A^{S'}(t_1)}
 \cdots
 \operatorname{ad}_{A^{S'}(t_m)}
 D_{S,i}(u).
\]

Every monomial in $A^{S'}(t)$ and $D_{S,i}(u)$ has degree at most
$d$, since the quadratic evolution preserves Majorana degree.  If
$P_R$ and $P_T$ have even degree, their commutator can be nonzero only
if $R\cap T\ne\varnothing$.  In that case,
\begin{equation}
 |R\triangle T|
 =
 |R|+|T|-2|R\cap T|
 \le
 |T|+d-2.
\end{equation}
Thus, each additional nonzero commutator increases the degree by at
most $d-2$.  After $j-1$ commutators, the degree is therefore at most
\begin{equation}
 d+(j-1)(d-2)\le jd.
 \label{eq:SM-nested-degree}
\end{equation}

We now apply $L_q$ to \eqref{eq:SM-B-nested}.  The triangle inequality gives
\begin{align}
 L_q(B_{S,i}(u))
 \le{}&
 L_q(D_{S,i}(u))
 \nonumber\\
 &+
 \sum_{m\ge1}
 \int_{0<t_1<\cdots<t_m<u}
 L_q\left(
 \operatorname{ad}_{A^{S'}(t_1)}
 \cdots
 \operatorname{ad}_{A^{S'}(t_m)}
 D_{S,i}(u)
 \right)
 \,\dd\bm t.
 \label{eq:SM-B-Lq-expansion}
\end{align}
Applying Lemma~\ref{lem:SM-comm} successively, and using
$J_q(A^{S'}(t))\le J_q(A(t))$, gives
\begin{align}
 &L_q\left(
 \operatorname{ad}_{A^{S'}(t_1)}
 \cdots
 \operatorname{ad}_{A^{S'}(t_m)}
 D_{S,i}(u)
 \right)
 \nonumber\\
 &\qquad\le
 m!\left(\frac{2d}{q^2}\right)^m
 J_q(A(u))
 \prod_{j=1}^m J_q(A(t_j)).
 \label{eq:SM-nested-bound}
\end{align}
The factor $m!$ comes from the successive degree bounds
$d,2d,\ldots,md$ in the commutator estimate.

Using this bound in Eq.~\eqref{eq:SM-B-Lq-expansion} and integrating
over $u$, we obtain
\begin{align}
 \int_0^T L_q(B_{S,i}(u))\,\dd u
 \le
 \sum_{m\ge0}
 m!\left(\frac{2d}{q^2}\right)^m
 \int_0^T\dd u
 \int_{0<t_1<\cdots<t_m<u}
 J_q(A(u))
 \prod_{j=1}^mJ_q(A(t_j))
 \,\dd\bm t.
 \label{eq:SM-B-integrated}
\end{align}
Here the term with $m=0$ is the contribution from
$L_q(D_{S,i}(u))\le J_q(A(u))$.

The integrand is symmetric in the $m+1$ time variables
$t_1,\ldots,t_m,u$.  The region
$0<t_1<\cdots<t_m<u<T$ is one of their $(m+1)!$ possible orderings.
Therefore,
\begin{align}
 &\int_0^T\dd u
 \int_{0<t_1<\cdots<t_m<u}
 J_q(A(u))
 \prod_{j=1}^mJ_q(A(t_j))
 \,\dd\bm t
 \nonumber\\
 &\qquad=
 \frac{1}{(m+1)!}
 \left(
 \int_0^T J_q(A(s))\,\dd s
 \right)^{m+1}
 =
 \frac{\eta_q^{m+1}}{(m+1)!}
 \le
 \frac{\eta_q^{m+1}}{m!}.
 \label{eq:SM-time-simplex}
\end{align}
Substituting this estimate into Eq.~\eqref{eq:SM-B-integrated} gives
\begin{align}
 \int_0^T L_q(B_{S,i}(u))\,\dd u
 &\le
 \sum_{m\ge0}
 \left(\frac{2d}{q^2}\right)^m
 \eta_q^{m+1}
 \nonumber\\
 &=
 \frac{\eta_q}{1-2d\eta_q/q^2}.
 \label{eq:SM-Bbound}
\end{align}

Finally, the differential equation
\eqref{eq:SM-M-differential} has the Dyson series solution
\begin{align}
 M_{S,i}-\one
 =
 \sum_{k\ge1}(-1)^k
 \int_{0<u_k<\cdots<u_1<T}
 B_{S,i}(u_1)\cdots B_{S,i}(u_k)
 \,\dd\bm u.
 \label{eq:SM-M-Dyson}
\end{align}
Since $L_q$ is submultiplicative, taking $L_q$ and using the symmetry
of the scalar integrand gives
\begin{align}
 L_q(M_{S,i}-\one)
 &\le
 \sum_{k\ge1}\frac{1}{k!}
 \left(
 \int_0^T L_q(B_{S,i}(u))\,\dd u
 \right)^k
 \nonumber\\
 &=
 \exp\left(
 \int_0^T L_q(B_{S,i}(u))\,\dd u
 \right)-1.
 \label{eq:SM-M-Lq-bound}
\end{align}
Together with Eq.~\eqref{eq:SM-Bbound}, this proves
Eq.~\eqref{eq:SM-ratio-est}.  Finally,
\begin{equation}
 L_q(M_{S,i}^\dagger-\one)
 =
 L_q\bigl((M_{S,i}-\one)^\dagger\bigr)
 =
 L_q(M_{S,i}-\one),
\end{equation}
so the same estimate holds for $M_{S,i}^\dagger-\one$.
\end{proof}

The relative operator $M_{S,i}$ need not be Gaussian and may contain
Majorana monomials of arbitrarily high degree. Write
\begin{equation}
 M_{S,i}-\one=\sum_Rm_RP_R,
\end{equation}
{Lemma~\ref{lem:SM-ratio}} at
$q=2$ gives
\begin{equation}
 \sum_R |m_R|2^{|R|}\le r,
 \label{eq:SM-degree-weighted-bound}
\end{equation}
where $r$ is a constant independent of $S$, $i$, and the system size.
For each $s$,
\begin{equation}
 \sum_{|R|=s}|m_R|\le r2^{-s},
\end{equation}
so the total coefficient weight decays exponentially with the degree.

{
For $g\le g_{\mathrm{mix}}$,
Eq.~\eqref{eq:SM-etaq-bound} gives
\begin{equation}
 \eta_2\le 2^dC_VgI_\beta\le\eta_{\mathrm{mix}}.
\end{equation}
The definition of $\eta_{\mathrm{mix}}$ implies
\begin{equation}
 \frac{d\eta_2}{2}\le\frac12,
 \qquad
 2\eta_2\le\log\frac{13}{12}.
\end{equation}
Substituting these inequalities into
Eq.~\eqref{eq:SM-ratio-est} at $q=2$ gives
\begin{equation}
 \sum_R|m_R|2^{|R|}
 \le
 L_2(M_{S,i}-\one)
 \le\frac1{12}.
\end{equation}
The same argument applies to $M_{S,i}^\dagger-\one$.  This proves
Lemma~\ref{lem:SM-relative-coefficient}.
{Proposition~\ref{prop:SM-hard-decomposition} therefore gives
Eq.~\eqref{eq:SM-hard-mean} for $g\le g_{\mathrm{mix}}$.}
}

\subsection{\texorpdfstring{Gaussianity of $G\sigma G$}{Gaussianity of G sigma G and normalization}}
\label{sec:SM-Gaussian-closure}

\begin{proof}[Proof of Lemma~\ref{lem:SM-Gaussian-sandwich}]
We first show that $G\sigma G$ is Gaussian for every nonzero
nonnegative Gaussian operator $\sigma$.

Suppose first that $\sigma$ is invertible.  It can then be written,
up to a positive scalar, as the exponential of a quadratic Majorana
operator.  The same is true of $G$.  Quadratic Majorana operators are

closed under commutators.  Their exponentials therefore form a
group of invertible Gaussian operators, and the product $G\sigma G$
is Gaussian.  It is also
nonnegative, since
\begin{equation}
 \langle\psi|G\sigma G|\psi\rangle
 =
 \langle G\psi|\sigma|G\psi\rangle
 \ge0
\end{equation}
for every $|\psi\rangle$.

Now let $\sigma$ be noninvertible.  By the definition of
noninvertible Gaussian states, there is a sequence of invertible
nonnegative Gaussian operators $\sigma^{(\ell)}$ such that
\begin{equation}
 \sigma^{(\ell)}\longrightarrow\sigma
\end{equation}
in operator norm.  The argument above shows that
$G\sigma^{(\ell)}G$ is Gaussian and nonnegative for every $\ell$.
Moreover,
\begin{equation}
 G\sigma^{(\ell)}G\longrightarrow G\sigma G.
\end{equation}
Since $G$ is invertible and $\sigma$ is nonzero,
$\Tr(G\sigma G)>0$.  Therefore,
\begin{equation}
 \frac{G\sigma^{(\ell)}G}
 {\Tr(G\sigma^{(\ell)}G)}
 \longrightarrow
 \frac{G\sigma G}
 {\Tr(G\sigma G)}.
\end{equation}
The set $\cG_n$ includes its noninvertible limits, so the operator on
the right belongs to $\cG_n$.

Returning to the convex combination of $B$, omit any terms
with $\sigma_z=0$ and define
\begin{equation}
 \tau_z
 =
 \frac{G\sigma_zG}{\Tr(G\sigma_zG)},
 \qquad
 q_z
 =
 \frac{\pi_z\Tr(G\sigma_zG)}
 {\Tr(GBG)}.
\end{equation}
Each $\tau_z$ belongs to $\cG_n$, while $q_z\ge0$ and
\begin{equation}
 \sum_zq_z
 =
 \frac{\Tr\left(G\sum_z\pi_z\sigma_zG\right)}
 {\Tr(GBG)}
 =1.
\end{equation}
Finally,
\begin{equation}
 \frac{GBG}
 {\Tr(GBG)}
 =
 \sum_zq_z\tau_z,
\end{equation}
which proves the claim.
\end{proof}

\begin{proof}[Proof of
Theorem~\ref{thm:SM-mixture}]
For $g\le g_{\mathrm{mix}}$,
Lemma~\ref{lem:SM-relative-coefficient} gives the coefficient bound
used in Proposition~\ref{prop:SM-hard-decomposition}.  Hence
Eq.~\eqref{eq:SM-hard-mean} writes $XX^\dagger$ as a convex
combination of nonnegative Gaussian operators.

We apply Lemma~\ref{lem:SM-Gaussian-sandwich} with
\begin{equation}
 B=XX^\dagger,
 \qquad
 G=\ee^{-\beta H_0/2}.
\end{equation}
Equation~\eqref{eq:SM-half-contour} gives
\begin{equation}
 GBG=\ee^{-\beta H}.
\end{equation}
Together with Eq.~\eqref{eq:SM-hard-mean}, the lemma therefore gives
\begin{equation}
 \rho_\beta
 =
 \frac{\ee^{-\beta H}}{\Tr\ee^{-\beta H}}
 \in\conv(\cG_n).
\end{equation}
This proves Theorem~\ref{thm:SM-mixture}.
\end{proof}

\section{Sampling tree method}
\label{sec:SM-soft-tree}

Appendix~\ref{sec:SM-mixture-proof} proves the existence of a convex
decomposition into Gaussian operators.  Turning that decomposition
into a sampling algorithm requires three additional ingredients.

First, the pinning operators $\one\pm B$, where
$B=\ii\gamma_a\gamma_b$, have a zero eigenvalue.  The first subsection
replaces them by the strictly positive operators $\one\pm tB$ and
controls the coefficient of the remaining Majorana monomial.  This
keeps every intermediate operator positive and makes its trace
comparable to a Gaussian trace, which is the quantity estimated in
Appendix~\ref{sec:SM-oracle}.

Second, normalizing the Gaussian operators at the leaves changes their
relative probabilities according to their traces.  The second
subsection expresses the required reweighting in terms of trace
weights assigned to the nodes of the tree.  These weights determine
the local transition probabilities used in
Appendix~\ref{sec:SM-finite}.

Finally, the index removal step requires samples of the relative
operators $M_{S,i}$ and $M_{S,i}^\dagger$.  Appendix
\ref{sec:SM-mixture-proof} defines such samples through their collected
Majorana coefficients, but computing those coefficients requires
summing the full expansion.  The third subsection instead samples one
expansion path directly and proves an exponential tail bound for its
length.

The first subsection uses the following lemma, which is proved in the
third subsection.

\begin{lemma}
\label{lem:SM-ratio-sampler}
Suppose that $g\le g_{\mathrm{pin}}$.  For every
$S\subseteq[N]$ and $i\in S$, there is a random operator
\begin{equation}
 Z=0
 \quad\text{or}\quad
 Z=\frac{\zeta}{32}Q^{-|R|}P_R,
 \qquad
 R\subseteq S,
 \qquad
 \zeta\in\{\pm1,\pm\ii\},
 \label{eq:SM-Z}
\end{equation}
such that
\begin{equation}
 \bbE(\one+Z)=M_{S,i}.
\end{equation}
The same statement holds for $M_{S,i}^\dagger$, and the two random
operators may be sampled independently.
\end{lemma}

\subsection{Constructing the sampling tree}

For $B=\ii\gamma_a\gamma_b$, the operator
$(\one+sB)/2$ projects onto the eigenspace $B=s$ and therefore pins
the corresponding Majorana pair. We replace this projector by the
full rank operator
\begin{equation}
 \one+tsB,\qquad 0<t<1.
\end{equation}
We refer to this operation as soft pinning.

Assume $g\le g_{\mathrm{pin}}$, so that
Lemma~\ref{lem:SM-ratio-sampler} applies.  At each stage of the
construction, we write the current term as
\begin{equation}
 \mathcal E_S(C)=X^SC(X^S)^\dagger,
\end{equation}
where
\begin{gather}
 C=(\one+\alpha\Gamma)\sigma,
 \nonumber\\
 \sigma=c\prod_{e=(a_e,b_e)\in\cM}
 (\one+t s_eB_e),
 \qquad
 B_e=\ii\gamma_{a_e}\gamma_{b_e}.
 \label{eq:SM-soft-sigma}
\end{gather}
Here $S$ specifies the Majorana indices retained in $X^S$.
The set $\cM$ consists of disjoint pairs
$e=(a_e,b_e)$ with $a_e,b_e\in S^c$, and hence is a matching on
$S^c$.  For each $e\in\cM$, the sign $s_e\in\{\pm1\}$ specifies the
eigenvalue of $B_e$ favored by soft pinning.  The scalar $c$ is
positive.

Here $\alpha\ge0$, and $\Gamma$ is a Hermitian even Majorana
monomial.  No index in $\supp\Gamma$ belongs to a pair in $\cM$.
When $\alpha=0$, we set $\Gamma=\one$.  At the root,
\begin{equation}
 S=[N],\qquad
 \cM=\varnothing,\qquad
 c=1,\qquad
 \alpha=0,
\end{equation}
so that $C=\sigma=\one$ and
$\mathcal E_{[N]}(C)=XX^\dagger$.

Each step must preserve the average and the positivity of $C$ while
reducing either the active set or the degree of $\Gamma$.  Index
removal decreases $|S|$ by one.  Once
$\supp\Gamma\subseteq S^c$, soft pinning transfers a quadratic pair
from $\Gamma$ to $\sigma$ through the factor
$\one\pm t\ii\gamma_a\gamma_b$.  Repetition reduces $\Gamma$ to a
scalar.  Absorbing the resulting positive scalar into $c$ leaves
$C=\sigma$ Gaussian.

Positivity imposes the invariant that controls both operations.  The
operator $\sigma$ is positive, while $\Gamma^2=\one$ and
$\Gamma$ commutes with $\sigma$.  Hence
\begin{equation}
 (1-\alpha)\sigma
 \preceq C
 \preceq(1+\alpha)\sigma.
 \label{eq:SM-C-sigma-comparison}
\end{equation}
In particular, $\alpha\le1$ is sufficient for positivity.  We will
maintain the stronger bound $\alpha\le1/2$, which also keeps $C$
uniformly comparable to $\sigma$. 

A fixed upper bound on $\alpha$ is not by itself preserved by the
recursion.  Soft pinning replaces $\alpha$ by $\alpha/t$.  It is
applied when removing two indices from
$\supp\Gamma\subseteq S^c$ at each step. Thus the bound before soft pinning must
contain one power of $t$ for every such pair.  Index removal has a
different effect.  The sampled monomials $Z_1,Z_2$ are supported in
$S$, and a nonzero sample of degree $s$ has coefficient of absolute
value $Q^{-s}/32$.  This suggests assigning one power of $Q^{-1}$ to
each index of $\supp\Gamma$ that lies in $S$.

Recall from Eq.~\eqref{eq:SM-Qt} that
\begin{equation}
 Q\ge8,
 \qquad
 t=\frac4{Q^2}\le\frac1{16},
 \qquad
 Q\sqrt t=2.
 \label{eq:SM-sync}
\end{equation}
Write
\begin{equation}
 m=|\supp\Gamma\cap S|,
 \qquad
 \ell=|\supp\Gamma\cap S^c|.
\end{equation}
We maintain the invariant
\begin{equation}
 {\alpha\le\lambda t^{\ell/2}Q^{-m}},
 \qquad
 \lambda=\frac12.
 \label{eq:SM-soft-invariant}
\end{equation}

The two powers in Eq.~\eqref{eq:SM-soft-invariant} reflect the effects
of the two operations.  In a soft pinning step, $\ell$ decreases by
two and $\alpha$ is divided by $t$, so the power $t^{\ell/2}$ changes
in the same way.  The power $Q^{-m}$ matches the degree dependence of
the monomials sampled during index removal.  Moreover, when an index
of $\supp\Gamma$ moves from $S$ to $S^c$, the right hand side of
Eq.~\eqref{eq:SM-soft-invariant} increases by
\begin{equation}
 Q\sqrt t=2.
\end{equation}
This compensates for the possible doubling of $\alpha$ in the index
removal step.

Since $t^{\ell/2}Q^{-m}\le1$, the invariant implies
$\alpha\le1/2$.  Each operator $\one+t s_eB_e$ has eigenvalues
$1-t$ and $1+t$.  Since the pairs in $\cM$ are disjoint, these
operators commute, and $\sigma$ is positive, invertible, and
Gaussian.  Moreover, $\Gamma^2=\one$ and $\Gamma$ commutes with
$\sigma$.  It follows that
\begin{equation}
 \frac12\sigma\preceq C\preceq\frac32\sigma.
\end{equation}
In particular, $C$ is positive and invertible.

We declare a node with $S=\varnothing$ and $\alpha=0$ to be a leaf.
At every other node, the next step is chosen as follows.  If
$\alpha=0$, remove the first index of $S$ in a fixed ordering.  If
$\alpha>0$ and $m>0$, remove the first index in
$\supp\Gamma\cap S$.  If $\alpha>0$ and $m=0$, then
$\supp\Gamma\subseteq S^c$.  Apply soft pinning until $\Gamma$
becomes a scalar, absorb $\one+\alpha\Gamma$ into $c$, and set
$\alpha=0$.  If $S$ is still nonempty, continue with index removal.

We now analyze the two operations, beginning with index removal.

\paragraph{Removing an index from $S$.}
Let $i$ be chosen by the rule above, set
$S'=S\setminus\{i\}$, and write $M=M_{S,i}$.  Since
$X^S=X^{S'}M$,
\begin{equation}
 \mathcal E_S(C)
 =
 X^{S'}MCM^\dagger(X^{S'})^\dagger.
 \label{eq:SM-index-removal-parent}
\end{equation}
The operator $M$ is even and supported in $S$, whereas $\sigma$ is
supported in $S^c$.  Hence both $M$ and $M^\dagger$ commute with
$\sigma$, and
\begin{equation}
 MCM^\dagger
 =
 M(\one+\alpha\Gamma)M^\dagger\sigma.
 \label{eq:SM-index-removal-sigma}
\end{equation}
It remains to express
$M(\one+\alpha\Gamma)M^\dagger$ as an average of operators
$\one+\alpha_w\Gamma_w$ satisfying
Eq.~\eqref{eq:SM-soft-invariant} with $S'$ in place of $S$.

Let $Z_1$ and $Z_2$ be independent random variables supplied by
Lemma~\ref{lem:SM-ratio-sampler}, chosen so that
\begin{equation}
 \bbE(\one+Z_1)=M,
 \qquad
 \bbE(\one+Z_2)=M^\dagger.
\end{equation}
Define
\begin{equation}
 Y=(\one+Z_1)(\one+\alpha\Gamma)(\one+Z_2)
   =\one+\sum_{k=1}^7T_k.
 \label{eq:SM-seven}
\end{equation}
Independence gives
\begin{equation}
 \bbE Y=M(\one+\alpha\Gamma)M^\dagger.
 \label{eq:SM-seven-mean}
\end{equation}

Conditioned on $Z_1,Z_2$, choose an index
$K\in\{1,\ldots,7\}$ with probabilities $p_1,\ldots,p_7$ and set
\begin{equation}
 \operatorname{Herm}(p_K^{-1}T_K)=\alpha_w\Gamma_w,
 \qquad
 \operatorname{Herm}(A)=\frac{A+A^\dagger}{2}.
 \label{eq:SM-selected-correction}
\end{equation}
Then
\begin{equation}
 \bbE_K\!\left[
 \operatorname{Herm}(p_K^{-1}T_K)
 \,\,|\,dle|\,Z_1,Z_2
 \right]
 =
 \operatorname{Herm}(Y-\one).
 \label{eq:SM-selected-conditional-mean}
\end{equation}
Since the operator on the right hand side of
Eq.~\eqref{eq:SM-seven-mean} is Hermitian, averaging also over
$Z_1,Z_2$ gives
\begin{equation}
 \bbE(\one+\alpha_w\Gamma_w)
 =
 M(\one+\alpha\Gamma)M^\dagger.
 \label{eq:SM-selected-mean}
\end{equation}

We assign to this outcome the operator
\begin{equation}
 C_w=(\one+\alpha_w\Gamma_w)\sigma.
\end{equation}
Equations~\eqref{eq:SM-index-removal-sigma} and
\eqref{eq:SM-selected-mean} imply
\begin{equation}
 \bbE C_w=MCM^\dagger,
 \qquad
 \bbE\mathcal E_{S'}(C_w)=\mathcal E_S(C).
 \label{eq:SM-index-removal-average}
\end{equation}
Thus the children produced by this step average to their parent.

Each nonzero $T_k$ is a scalar multiple of an even Majorana
monomial.  Its Hermitian part is therefore either zero or a real
multiple of a Hermitian even Majorana monomial.  We include the sign
in $\Gamma_w$ so that $\alpha_w\ge0$.  If the Hermitian part is zero,
we set $\alpha_w=0$ and $\Gamma_w=\one$.  If it is scalar, we take
$\Gamma_w\in\{\one,-\one\}$.  The outcome $K$ retains its assigned
probability in both cases.

The probabilities in Table~\ref{tab:SM-seven} are chosen so that
every resulting coefficient $\alpha_w$ satisfies
Eq.~\eqref{eq:SM-soft-invariant} with $S'$ in place of $S$.  The
probabilities in the table sum to one, and dividing the corresponding
terms by them gives the coefficient bounds in the last column.  The
lemma below accounts for the support dependence in
Eq.~\eqref{eq:SM-soft-invariant}.

\begin{table}[htbp]
\centering
\renewcommand{\arraystretch}{1.2}
\caption{The seven terms in Eq.~\eqref{eq:SM-seven}, their selection
probabilities, and the resulting coefficient bounds.  Here $s$ is
the degree of the single sampled monomial in a term, and $s_1,s_2$
are the degrees of the monomials in $Z_1,Z_2$ when both occur.
If a selected term vanishes, then $\alpha_w=0$.}
\label{tab:SM-seven}
\begin{tabular}{lcc}
\hline
selected term & selection probability & bound on $\alpha_w$ \\
\hline
$\alpha\Gamma$ & $1/2$ & $2\alpha$ \\
$Z_1$ or $Z_2$ & $1/16$ each & $\tfrac12Q^{-s}$ \\
$Z_1\alpha\Gamma$ or $\alpha\Gamma Z_2$
 & $1/16$ each & $\tfrac12\alpha Q^{-s}$ \\
$Z_1Z_2$ & $1/8$ & $\tfrac1{128}Q^{-(s_1+s_2)}$ \\
$Z_1\alpha\Gamma Z_2$
 & $1/8$ & $\tfrac1{128}\alpha Q^{-(s_1+s_2)}$ \\
\hline
\end{tabular}
\end{table}

The last column of Table~\ref{tab:SM-seven} controls the coefficient
$\alpha_w$ before the support of $\Gamma_w$ is taken into account.
To compare these bounds with
Eq.~\eqref{eq:SM-soft-invariant}, define, for every Majorana monomial
$P$,
\begin{equation}
 w_S(P)=
 t^{|\supp P\cap S^c|/2}
 Q^{-|\supp P\cap S|}.
 \label{eq:SM-wS}
\end{equation}
The invariant in Eq.~\eqref{eq:SM-soft-invariant} can then be written
as
\begin{equation}
 \alpha\le\lambda w_S(\Gamma).
\end{equation}
The following lemma combines the coefficient bounds in the table
with the change in support under index removal.

\begin{lemma}
\label{lem:SM-soft-invariant-preserved}
Suppose the current node satisfies the support conditions stated
after Eq.~\eqref{eq:SM-soft-sigma} and
\begin{equation}
 \alpha\le\lambda w_S(\Gamma).
\end{equation}
Then every outcome of the index removal step has the form
\begin{equation}
 C_w=(\one+\alpha_w\Gamma_w)\sigma
\end{equation}
with $S$ replaced by $S'=S\setminus\{i\}$, and satisfies
\begin{equation}
 \alpha_w\le\lambda w_{S'}(\Gamma_w).
 \label{eq:SM-child-bound}
\end{equation}
\end{lemma}
\begin{proof}
For two Majorana monomials $P_R$ and $P_{R'}$, the canonical
anticommutation relations give
\begin{equation}
 P_RP_{R'}
 =
 \omega_{R,R'}P_{R\mathbin{\triangle}R'},
 \qquad
 |\omega_{R,R'}|=1,
\end{equation}
where $R\mathbin{\triangle}R'$ denotes the symmetric difference.
Thus every index in $R\cap R'$ disappears from the support of the
product.  In $w_S(P_R)w_S(P_{R'})$, each such index contributes
either $Q^{-2}$ or $t$, both of which are at most one .  Removing
these contributions can only increase the weight.  Therefore
\begin{equation}
 w_S(P_RP_{R'})
 \ge
 w_S(P_R)w_S(P_{R'}).
 \label{eq:SM-supermult}
\end{equation}

Replacing $S$ by $S'=S\setminus\{i\}$ changes only the contribution
of the index $i$.  If $i\in\supp P$, this contribution changes from
$Q^{-1}$ to $\sqrt t$.  Since $Q\sqrt t=2$,
\begin{equation}
 w_{S'}(P)
 =
 \begin{cases}
  2w_S(P),& i\in\supp P,\\
  w_S(P),& i\notin\supp P.
 \end{cases}
 \label{eq:SM-weight-change}
\end{equation}
In particular, $w_{S'}(P)\ge w_S(P)$.  Taking the Hermitian part
cannot increase the absolute value of a coefficient.  If the result
is nonzero, it remains proportional to the same Majorana monomial
and therefore has the same support.

We now consider the seven terms in
Table~\ref{tab:SM-seven}.  If the selected term vanishes, then
$\alpha_w=0$, and Eq.~\eqref{eq:SM-child-bound} holds immediately.

Suppose first that the selected term is $\alpha\Gamma$ and is
nonzero.  Then $\alpha>0$, so the choice of $i$ gives
$i\in\supp\Gamma\cap S$.  By
Eq.~\eqref{eq:SM-weight-change},
\begin{equation}
 w_{S'}(\Gamma)=2w_S(\Gamma).
\end{equation}
The first row of Table~\ref{tab:SM-seven} and the invariant at the
current node therefore give
\begin{equation}
 \alpha_w
 =2\alpha
 \le2\lambda w_S(\Gamma)
 =\lambda w_{S'}(\Gamma_w).
\end{equation}

Next suppose that the selected term is $Z_1$ or $Z_2$.  Let $P$ be
its sampled Majorana monomial and let $s=|\supp P|$.  Since
$\supp P\subseteq S$,
\begin{equation}
 w_S(P)=Q^{-s}.
\end{equation}
If the Hermitian part is nonzero, then
$\supp\Gamma_w=\supp P$.  The second row of
Table~\ref{tab:SM-seven} and
$w_{S'}(P)\ge w_S(P)$ give
\begin{equation}
 \alpha_w
 \le\frac12Q^{-s}
 =\lambda w_S(P)
 \le\lambda w_{S'}(\Gamma_w).
\end{equation}
This argument does not require $i\in\supp P$.

Now consider $Z_1\alpha\Gamma$ or
$\alpha\Gamma Z_2$.  Let $P$ denote the sampled monomial and let
$s=|\supp P|$.  For a nonzero Hermitian part,
$\Gamma_w$ has the same support as $P\Gamma$.  Equations
\eqref{eq:SM-supermult} and \eqref{eq:SM-weight-change} imply
\begin{equation}
 w_{S'}(\Gamma_w)
 \ge w_S(P\Gamma)
 \ge w_S(P)w_S(\Gamma)
 =Q^{-s}w_S(\Gamma).
 \label{eq:SM-one-Z-support}
\end{equation}
Using the third row of Table~\ref{tab:SM-seven} and the invariant at
the current node, we obtain
\begin{align}
 \alpha_w
 &\le\frac12\alpha Q^{-s}
 \nonumber\\
 &\le\frac12\lambda w_S(\Gamma)Q^{-s}
 \nonumber\\
 &\le\lambda w_{S'}(\Gamma_w).
\end{align}

It remains to consider the two terms containing both sampled
monomials.  Let $P_1,P_2$ be the monomials in $Z_1,Z_2$, with
degrees $s_1,s_2$.  If the selected term is $Z_1Z_2$, then
\begin{equation}
 w_{S'}(\Gamma_w)
 \ge w_S(P_1P_2)
 \ge Q^{-(s_1+s_2)}.
\end{equation}
The fourth row of Table~\ref{tab:SM-seven} and
$\lambda=1/2$ therefore give
\begin{equation}
 \alpha_w
 \le\frac1{128}Q^{-(s_1+s_2)}
 \le\lambda w_{S'}(\Gamma_w).
\end{equation}

If the selected term is
$Z_1\alpha\Gamma Z_2$, then
\begin{equation}
 w_{S'}(\Gamma_w)
 \ge w_S(P_1\Gamma P_2)
 \ge Q^{-(s_1+s_2)}w_S(\Gamma).
\end{equation}
Using the last row of Table~\ref{tab:SM-seven} and the invariant at
the current node gives
\begin{align}
 \alpha_w
 &\le\frac1{128}\alpha Q^{-(s_1+s_2)}
 \nonumber\\
 &\le\frac1{128}\lambda
 w_S(\Gamma)Q^{-(s_1+s_2)}
 \nonumber\\
 &\le\lambda w_{S'}(\Gamma_w).
\end{align}

If $\alpha=0$, the terms containing $\alpha\Gamma$ vanish.  The
arguments for $Z_1$, $Z_2$, and $Z_1Z_2$ do not use the assumption
$i\in\supp\Gamma$, so they remain valid when $i$ is the first index
of $S$ in the fixed ordering.  This proves
Eq.~\eqref{eq:SM-child-bound} for every outcome.

It remains to check the support conditions.  The sampled monomials
are even, so every nonzero $\Gamma_w$ is even.  They are supported
in $S$, whereas every pair in $\cM$ lies in $S^c$.  Since
$\Gamma$ is already disjoint from these pairs, $\Gamma_w$ remains
disjoint from them.  The matching $\cM$ is unchanged during index
removal, and
\begin{equation}
 S^c\subseteq(S')^c.
\end{equation}
Thus all pairs in $\cM$ still lie outside $S'$, and all the required
support conditions are preserved.
\end{proof}
\paragraph{Soft pinning.}
Suppose $\alpha>0$ and $m=0$, so that
$\supp\Gamma\subseteq S^c$.  Soft pinning reduces the support of
$\Gamma$ by two indices at a time and adds the corresponding
quadratic operator to $\sigma$.

If $\Gamma=\varepsilon\one$ with
$\varepsilon\in\{\pm1\}$, then
\begin{equation}
 C=(1+\varepsilon\alpha)\sigma.
\end{equation}
The invariant gives $\alpha\le\lambda=1/2$, so
$1+\varepsilon\alpha>0$.  We replace $c$ by
$c(1+\varepsilon\alpha)$ and set
$\alpha=0$ and $\Gamma=\one$.  This leaves $C$ unchanged.

Now suppose that $\Gamma$ is not scalar.  Since $\Gamma$ is even and
supported in $S^c$, its degree is $\ell=2r$ for some $r\ge1$.  Pair
the indices in its support in a fixed order and write
\begin{equation}
 \Gamma=\varepsilon Q_1\cdots Q_r,
 \qquad
 \varepsilon\in\{\pm1\},
 \qquad
 Q_j=\ii\gamma_{a_j}\gamma_{b_j}.
 \label{eq:SM-defect-pair}
\end{equation}
The pairs $(a_j,b_j)$ are disjoint, all their indices lie in $S^c$,
and none of these indices belongs to a pair in $\cM$.  The operators
$Q_j$ are commuting Hermitian involutions.  Since $m=0$, the
invariant becomes
\begin{equation}
 \alpha\le\lambda t^r.
 \label{eq:SM-pinning-bound}
\end{equation}

Set
\begin{equation}
 A=\varepsilon Q_1\cdots Q_{r-1}.
\end{equation}
Then $\Gamma=AQ_r$, and
\begin{equation}
 \one+\alpha AQ_r
 =
 \frac12
 \left(\one+\frac{\alpha}{t}A\right)(\one+tQ_r)
 +
 \frac12
 \left(\one-\frac{\alpha}{t}A\right)(\one-tQ_r).
 \label{eq:SM-Walsh-step}
\end{equation}
Accordingly, define
\begin{equation}
 C_\pm
 =
 \left(\one\pm\frac{\alpha}{t}A\right)
 (\one\pm tQ_r)\sigma.
 \label{eq:SM-soft-pinning-children}
\end{equation}
Equation~\eqref{eq:SM-Walsh-step} gives
\begin{equation}
 C=\frac12C_++\frac12C_-.
 \label{eq:SM-soft-pinning-average}
\end{equation}
Thus the two operators $C_+$ and $C_-$ are chosen with probability
$1/2$ each.

For the outcome with sign $s\in\{\pm1\}$, the new data are
\begin{equation}
 \alpha'=\frac{\alpha}{t},
 \qquad
 \Gamma'=sA,
 \qquad
 \sigma'=(\one+stQ_r)\sigma.
 \label{eq:SM-soft-pinning-update}
\end{equation}
The set $S$ is unchanged.  The pair $(a_r,b_r)$ is added to $\cM$
with sign $s$.  The support of $\Gamma'$ has
\begin{equation}
 m'=0,
 \qquad
 \ell'=2(r-1).
\end{equation}
By Eq.~\eqref{eq:SM-pinning-bound},
\begin{equation}
 \alpha'
 =\frac{\alpha}{t}
 \le\lambda t^{r-1}
 =\lambda t^{\ell'/2}Q^{-m'}.
\end{equation}
Hence soft pinning preserves
Eq.~\eqref{eq:SM-soft-invariant}.

The new pair is disjoint from $\supp A$ and from every pair already
in $\cM$.  Therefore $Q_r$ commutes with both $A$ and $\sigma$.
Moreover, $\one+stQ_r$ has eigenvalues $1-t$ and $1+t$, while
\begin{equation}
 \alpha'
 \le\lambda t^{r-1}
 \le\frac12.
\end{equation}
It follows that both $\one+\alpha'\Gamma'$ and
$\one+stQ_r$ are positive and invertible.  Thus each $C_\pm$ is
positive and invertible, and $\sigma'$ is again Gaussian.

Repeating the same step removes all indices from $\supp\Gamma$.
When $\Gamma$ becomes scalar, the scalar rule above sets
$\alpha=0$.  Iterating Eq.~\eqref{eq:SM-Walsh-step} gives the
equivalent identity
\begin{equation}
 \one+\alpha\varepsilon\prod_{j=1}^rQ_j
 =
 2^{-r}\sum_{\bm s\in\{\pm1\}^r}
 \left(
  1+\frac{\alpha\varepsilon}{t^r}\prod_{j=1}^rs_j
 \right)
 \prod_{j=1}^r(\one+ts_jQ_j).
 \label{eq:SM-Walsh}
\end{equation}
Equation~\eqref{eq:SM-pinning-bound} implies
\begin{equation}
 \frac12
 \le
 1+\frac{\alpha\varepsilon}{t^r}\prod_{j=1}^rs_j
 \le
 \frac32.
\end{equation}
Thus every scalar appearing after the last soft pinning step is
positive.  The recursive construction uses
Eq.~\eqref{eq:SM-Walsh-step} one pair at a time, so every soft
pinning step has two children.

Together with
Lemma~\ref{lem:SM-soft-invariant-preserved}, this proves that both
operations preserve the form of $C$, its positivity, and
Eq.~\eqref{eq:SM-soft-invariant}.  Repeating them eventually gives
$S=\varnothing$ and $\alpha=0$, so every leaf has the form
\begin{equation}
 C=\sigma,
\end{equation}
with $\sigma$ positive, invertible, and Gaussian.

Sampling $Z_1,Z_2$ still involves continuous time variables, and an
expansion path can have arbitrary length.
Appendix~\ref{sec:SM-finite} truncates these paths and discretizes
the time variables.

\subsection{Node weights}

The index-removal and soft-pinning choices form a rooted tree whose
leaves are unnormalized Gaussian operators.
Normalizing the leaves changes their relative weights according to
their traces.  To perform this reweighting using only local
information, rather than enumerating all leaves, we assign an operator
and a trace weight to every node.  For a node $v$, define
\begin{equation}
 C_v=(\one+\alpha_v\Gamma_v)\sigma_v,
 \qquad
 E_v=GX^{S_v}C_v(X^{S_v})^\dagger G,
 \qquad
 W_v=\Tr E_v.
 \label{eq:SM-Ev}
\end{equation}
At the root, $S_{\mathrm{root}}=[N]$ and
$C_{\mathrm{root}}=\one$, while a leaf $z$ has
$S_z=\varnothing$ and $\alpha_z=0$.

The following lemma shows that the operator at a node is the
conditional average of the operators at its children.  Taking the
trace will then identify $W_v$ as the expected trace of the leaf
reached from $v$, which is the quantity needed for local reweighting. 
\begin{lemma}
\label{lem:SM-tree-martingale}
For every internal node $v$,
\begin{equation}
 E_v=\bbE(E_w\,|\, v),
 \label{eq:SM-martingale}
\end{equation}
where $w$ is the child selected according to the transition
probabilities at $v$.  Moreover,
\begin{equation}
 E_{\mathrm{root}}=\ee^{-\beta H},
\end{equation}
and the operator $E_z$ at every leaf $z$ is positive, invertible, and
Gaussian.
\end{lemma}

\begin{proof}
First consider an index removal step.  Set
$S'=S_v\setminus\{i\}$ and $M=M_{S_v,i}$, so that
$X^{S_v}=X^{S'}M$.  The index-removal identity gives
\begin{equation}
 \bbE(C_w\,|\, v)=MC_vM^\dagger.
\end{equation}
Since every child has $S_w=S'$,
\begin{align}
 \bbE(E_w\,|\, v)
 &=
 GX^{S'}\bbE(C_w\,|\, v)(X^{S'})^\dagger G
 \nonumber\\
 &=
 GX^{S'}MC_vM^\dagger(X^{S'})^\dagger G
 \nonumber\\
 &=
 GX^{S_v}C_v(X^{S_v})^\dagger G
 =E_v.
\end{align}

For a soft pinning step, $S_w=S_v$.
Equation~\eqref{eq:SM-Walsh-step}, together with the definition of
the two children, gives
\begin{equation}
 \bbE(C_w\,|\, v)=C_v.
\end{equation}
Conjugating this identity by $GX^{S_v}$ proves
Eq.~\eqref{eq:SM-martingale} for a soft pinning step as well.

At the root,
\begin{equation}
 S_{\mathrm{root}}=[N],
 \qquad
 C_{\mathrm{root}}=\one,
 \qquad
 X^{S_{\mathrm{root}}}=X.
\end{equation}
Therefore
\begin{equation}
 E_{\mathrm{root}}
 =GXX^\dagger G
 =\ee^{-\beta H}.
\end{equation}

At a leaf $z$, we have
$S_z=\varnothing$ and $\alpha_z=0$.  Hence
$X^{S_z}=\one$ and $C_z=\sigma_z$, so
\begin{equation}
 E_z=L_z=G\sigma_zG.
 \label{eq:SM-leaf}
\end{equation}
The leaf construction makes $\sigma_z$ positive, invertible, and
Gaussian.  Lemma~\ref{lem:SM-Gaussian-sandwich}
therefore gives the same properties for $L_z$.
\end{proof}
Let $p(w\,|\,v)$ denote the probability of selecting a child $w$ at
a node $v$.  For a leaf $z$ below $v$, let $p(z\,|\,v)$ be the
probability of reaching $z$ starting from $v$, and set
$p_z=p(z\,|\,\mathrm{root})$.  Iterating
Eq.~\eqref{eq:SM-martingale} gives
\begin{equation}
 E_v=\sum_{z\succeq v}p(z\,|\,v)L_z,
 \label{eq:SM-descendant-decomposition}
\end{equation}
where $z\succeq v$ means that $z$ is a leaf below $v$.  At the root,
this becomes
\begin{equation}
 \ee^{-\beta H}=\sum_zp_zL_z,
 \qquad
 L_z=G\sigma_zG.
 \label{eq:SM-sampler-native-decomposition}
\end{equation}

For a leaf $z$, define
\begin{equation}
 F_z=\Tr L_z,
 \qquad
 \tau_z=\frac{L_z}{F_z}.
\end{equation}
Taking the trace of
Eq.~\eqref{eq:SM-descendant-decomposition} gives
\begin{equation}
 W_v=\sum_{z\succeq v}p(z\,|\,v)F_z.
 \label{eq:SM-descendant-weight}
\end{equation}
Thus $W_v$ is the expected value of $F_z$ conditioned on starting
from $v$.  In particular,
\begin{equation}
 W_{\mathrm{root}}=\Tr\ee^{-\beta H},
\end{equation}
and Eq.~\eqref{eq:SM-sampler-native-decomposition} gives
\begin{equation}
 \rho_\beta
 =
 \sum_zq_z\tau_z,
 \qquad
 q_z=\frac{p_zF_z}{W_{\mathrm{root}}}.
\end{equation}

This distribution can be generated locally if the node weights are
known.  Taking the trace of Eq.~\eqref{eq:SM-martingale} gives
\begin{equation}
 W_v=\sum_w p(w\,|\,v)W_w,
 \label{eq:SM-node-weight-recursion}
\end{equation}
where the sum is over the children of $v$.  Therefore
\begin{equation}
 q(w\,|\,v)
 =
 \frac{p(w\,|\,v)W_w}{W_v}
 \label{eq:SM-reweighted-child}
\end{equation}
defines a probability distribution over these children.

Indeed, for a path
\begin{equation}
 v_0=\mathrm{root},\quad
 v_1,\ldots,v_h=z,
\end{equation}
the resulting probability of reaching $z$ is
\begin{align}
 \prod_{j=0}^{h-1}q(v_{j+1}\,|\, v_j)
 &=
 \prod_{j=0}^{h-1}
 p(v_{j+1}\,|\, v_j)\frac{W_{v_{j+1}}}{W_{v_j}}
 \nonumber\\
 &=
 \frac{p_zW_z}{W_{\mathrm{root}}}
 =
 \frac{p_zF_z}{W_{\mathrm{root}}}
 =q_z.
\end{align}
Thus the normalized leaf distribution can be sampled without
enumerating all leaves, provided that the weights $W_v$ can be
obtained at the nodes visited by the sampler.

The local probabilities in
Eq.~\eqref{eq:SM-reweighted-child} depend on the node weights
\begin{equation}
 W_v
 =
 \Tr\left[
 GX^{S_v}
 (\one+\alpha_v\Gamma_v)\sigma_v
 (X^{S_v})^\dagger G
 \right].
 \label{eq:SM-Wv-expanded}
\end{equation}
The operator $\sigma_v$ is Gaussian, but
$(\one+\alpha_v\Gamma_v)\sigma_v$ need not be.  The oracle in
Appendix~\ref{sec:SM-oracle} is constructed instead for
\begin{equation}
 Z_v
 =
 \Tr\left[
 GX^{S_v}\sigma_v(X^{S_v})^\dagger G
 \right].
 \label{eq:SM-Zv}
\end{equation}
We now compare this quantity with $W_v$.

At the node $v$, Eq.~\eqref{eq:SM-soft-invariant} reads
\begin{equation}
 \alpha_v
 \le
 \frac12
 t^{\ell_v/2}Q^{-m_v}
 \le\frac12,
\end{equation}
where
\begin{equation}
 m_v=|\supp\Gamma_v\cap S_v|,
 \qquad
 \ell_v=|\supp\Gamma_v\cap S_v^c|.
\end{equation}
Since $\Gamma_v^2=\one$,
\begin{equation}
 \frac12\one
 \preceq
 \one+\alpha_v\Gamma_v
 \preceq
 \frac32\one.
\end{equation}
The operator $\Gamma_v$ commutes with the positive operator
$\sigma_v$, and therefore
\begin{equation}
 \frac12\sigma_v
 \preceq
 (\one+\alpha_v\Gamma_v)\sigma_v
 \preceq
 \frac32\sigma_v.
\end{equation}
Conjugating by $GX^{S_v}$ and taking the trace gives
\begin{equation}
 \frac12Z_v
 \le W_v
 \le\frac32Z_v.
 \label{eq:SM-live-sandwich}
\end{equation}
Thus $Z_v$ approximates $W_v$ within fixed multiplicative factors.  An estimate of
$Z_v$ therefore provides the control of $W_v$ needed for local
reweighting.

The construction above uses
Lemma~\ref{lem:SM-ratio-sampler} as an input.  We now prove that lemma
by sampling one path through the expansion of $M_{S,i}-\one$.  This
expansion path records the choices made within a single
relative factor sample and is distinct from a path through the
sampling tree.

\subsection{\texorpdfstring{Sampling $M_{S,i}$}{Sampling M(S,i)}}

We now prove Lemma~\ref{lem:SM-ratio-sampler}.  Appendix B writes
\begin{equation}
 M_{S,i}-\one=\sum_R m_RP_R
 \label{eq:SM-collected-ratio}
\end{equation}
and bounds the weighted sum of the collected coefficients,
\begin{equation}
 \sum_R |m_R|q^{|R|}.
 \label{eq:SM-collected-mass}
\end{equation}
This estimate suffices for the convex decomposition in
Appendix~\ref{sec:SM-mixture-proof}.  There, $P_R$ is selected with
probability proportional to $|m_R|q^{|R|}$.  Thus the coefficients
$m_R$ have to be computed before sampling.  Here we instead sample an
expansion path directly.

Let $\mathfrak H_{S,i}$ denote the space of expansion paths described
below.  We use counting measure for their discrete choices and
Lebesgue measure for their time variables, and denote the resulting
measure by $\nu_{S,i}$.  Each path $h$ determines a coefficient $a(h)$
and a final monomial $P_{R(h)}$, so that
\begin{equation}
 M_{S,i}-\one
 =
 \int_{\mathfrak H_{S,i}}
 a(h)P_{R(h)}\,\nu_{S,i}(\dd h).
 \label{eq:SM-history-expansion}
\end{equation}
For a fixed $R$, summing the paths that end at $P_R$ gives
\begin{equation}
 m_R
 =
 \int_{\{h:R(h)=R\}}
 a(h)\,\nu_{S,i}(\dd h).
\end{equation}
Consequently,
\begin{equation}
 \sum_R|m_R|q^{|R|}
 \le
 \int_{\mathfrak H_{S,i}}
 |a(h)|q^{|R(h)|}\,\nu_{S,i}(\dd h).
 \label{eq:SM-collected-pathwise-comparison}
\end{equation}
The left-hand side takes the absolute value after the paths ending at
the same monomial have been summed.  The right-hand side takes the
absolute value of each path coefficient before summing.  Appendix~\ref{sec:SM-mixture-proof}
controls the former.  To define nonnegative probabilities for the
paths, we need to control the latter.

To bound the right-hand side of
Eq.~\eqref{eq:SM-collected-pathwise-comparison}, consider first the
expansion of $A(u)$.  Write $V=\sum_Tv_TP_T$.  Expanding the
determinants in Eq.~\eqref{eq:SM-exterior-power} gives
\begin{equation}
 \ee^{uH_0}P_T\ee^{-uH_0}
 =
 \sum_{\varphi:T\hookrightarrow[N]}
 c_{T,\varphi}(u)P_{\varphi(T)},
 \qquad
 |c_{T,\varphi}(u)|
 =
 \prod_{a\in T}|Q(u)_{a,\varphi(a)}|.
 \label{eq:SM-labeled-exterior-power}
\end{equation}
Here $\varphi$ ranges over injective maps.  The sign needed to put the
output labels in canonical order is included in $c_{T,\varphi}(u)$. 
For a fixed output index $b$, we sum the coefficient magnitudes  of the
terms whose final monomial contains $b$.  Define
\begin{equation}
 \widehat J_q(u)
 :=
 \max_b
 \sum_T |v_T|q^{|T|}
 \sum_{\substack{\varphi:T\hookrightarrow[N]\\
                  b\in\varphi(T)}}
 |c_{T,\varphi}(u)|.
 \label{eq:SM-uncollected-J}
\end{equation}
Unlike $J_q(A(u))$, this quantity is evaluated before terms with the
same final monomial are summed.  Hence
$J_q(A(u))\le\widehat J_q(u)$.

Every injective map whose image contains $b$ has a unique $a\in T$
such that $\varphi(a)=b$.  After dropping the injectivity condition
on the remaining outputs, we obtain
\begin{align}
 \widehat J_q(u)
 &\le
 \max_b\sum_T |v_T|q^{|T|}
 \sum_{a\in T}|Q(u)_{ab}|
 \prod_{c\in T\setminus\{a\}}
 \left(\sum_j|Q(u)_{cj}|\right)
 \nonumber\\
 &\le
 q^d\ee^{(d-1)C_Hu}
 \max_b\sum_a|Q(u)_{ab}|
 \sum_{T\ni a}|v_T|
 \nonumber\\
 &\le q^dC_Vg\,\ee^{dC_Hu}.
 \label{eq:SM-uncollected-Jbound}
\end{align}
This is the same estimate as in
Lemma~\ref{lem:SM-interaction-coefficients}, now applied before the
permutation terms are summed.  Set
\begin{equation}
 \widehat\eta_q
 :=
 \int_0^T\widehat J_q(u)\,\dd u.
 \label{eq:SM-etahat}
\end{equation}

We can now specify the expansion paths introduced above.  A path first
chooses a term of order $k$ in the Dyson series
Eq.~\eqref{eq:SM-M-Dyson}, together with its ordered times
$u_1,\ldots,u_k$.  For the $j$th occurrence of $B_{S,i}(u_j)$, it
chooses a term with $m_j$ nested commutators in
Eq.~\eqref{eq:SM-B-nested}, the corresponding times, and one labeled
determinant term from each occurrence of $D_{S,i}$ and $A^{S'}$.
Multiplying the selected Majorana monomials and putting their indices
in canonical order determines $P_{R(h)}$ and the coefficient $a(h)$. 

\begin{lemma}
\label{lem:SM-uncollected-ratio}
Suppose that $2d\widehat\eta_q/q^2<1$.  Then every expansion path
satisfies $R(h)\subseteq S$, and the expansion in
Eq.~\eqref{eq:SM-history-expansion} converges absolutely in $L_q$.
Moreover,
\begin{equation}
 \mathfrak M_q(M_{S,i}-\one)
 :=
 \int_{\mathfrak H_{S,i}}
 |a(h)|q^{|R(h)|}\,\nu_{S,i}(\dd h)
 \le
 \exp\left(
 \frac{\widehat\eta_q}
 {1-2d\widehat\eta_q/q^2}
 \right)-1.
 \label{eq:SM-history-majorant}
\end{equation}
In particular, the expansion converges in operator norm.  The same
statements hold for $M_{S,i}^\dagger-\one$.
\end{lemma}

\begin{proof}
Set $S'=S\setminus\{i\}$, $U(u)=X^{S'}(u)$, and
$D(u)=D_{S,i}(u)$.  As in Appendix~\ref{sec:SM-mixture-proof}, expand
\begin{equation}
 B_{S,i}(u)=U(u)^{-1}D(u)U(u)
\end{equation}
using Eq.~\eqref{eq:SM-B-nested}, and then substitute this expansion
into the Dyson series~\eqref{eq:SM-M-Dyson}.  Expanding every
determinant as in Eq.~\eqref{eq:SM-labeled-exterior-power}, without
summing terms with the same output monomial, gives
Eq.~\eqref{eq:SM-history-expansion}.  Every monomial selected from
$D_{S,i}$ or $A^{S'}$ is supported in $S$.  Their products are
therefore supported in $S$ as well, so $R(h)\subseteq S$.

It remains to bound the weighted coefficient sum in
Eq.~\eqref{eq:SM-history-majorant}.  Consider first a term with $m$
nested commutators in the expansion of $B_{S,i}(u)$.  Before the
$j$th commutator, the degree of the current monomial is at most
\begin{equation}
 d+(j-1)(d-2)\le jd.
\end{equation}
Let $P_R$ be a monomial selected from $A^{S'}(t_j)$ and let $P_T$ be
the current monomial.  A nonzero commutator requires
$R\cap T\ne\varnothing$, and
\begin{equation}
 q^{|R\triangle T|}
 \le q^{|R|+|T|-2}.
\end{equation}
Assign the pair $(R,T)$ to the first index in $R\cap T$ according to a
fixed ordering.  Summing over the at most $jd$ indices in $T$ and
using the definition of $\widehat J_q(t_j)$ shows that the $j$th
commutator increases the weighted coefficient sum by at most
\begin{equation}
 \frac{2jd}{q^2}\widehat J_q(t_j)
\end{equation}
relative to the current monomial.  Here the coefficient $2$ is the
magnitude  of a nonzero commutator.  Over all $m$ commutators, the
degree bounds contribute
\begin{equation}
 \prod_{j=1}^m\frac{2jd}{q^2}
 =
 m!\left(\frac{2d}{q^2}\right)^m.
 \label{eq:SM-path-commutator-bound}
\end{equation}

The labeled terms selected from $D(u)$ are supported in $S$ and
contain $i$, so their total $q$-weighted coefficient magnitude is at
most $\widehat J_q(u)$.  At the $j$th commutator, the corresponding
sum for the terms selected from $A^{S'}(t_j)$ and containing the
chosen shared index is at most $\widehat J_q(t_j)$.  It follows that
the contribution from terms with $m$ nested commutators is bounded by
\begin{align}
 &m!\left(\frac{2d}{q^2}\right)^m
 \int_0^T\dd u
 \int_{0<t_1<\cdots<t_m<u}
 \widehat J_q(u)
 \prod_{j=1}^m\widehat J_q(t_j)\,\dd\bm t
 \nonumber\\
 &\qquad=
 \frac{1}{m+1}\widehat\eta_q
 \left(\frac{2d\widehat\eta_q}{q^2}\right)^m
 \nonumber\\
 &\qquad\le
 \widehat\eta_q
 \left(\frac{2d\widehat\eta_q}{q^2}\right)^m.
 \label{eq:SM-one-B-path-bound}
\end{align}
The equality follows by symmetrizing the $m+1$ time variables.
Summing Eq.~\eqref{eq:SM-one-B-path-bound} over $m$, we obtain
\begin{equation}
 K_q
 :=
 \widehat\eta_q
 \sum_{m\ge0}
 \left(\frac{2d\widehat\eta_q}{q^2}\right)^m
 =
 \frac{\widehat\eta_q}
 {1-2d\widehat\eta_q/q^2}.
 \label{eq:SM-inner-history-mass}
\end{equation}
Thus $K_q$ bounds the weighted coefficient sum for one occurrence of
$B_{S,i}$, including its time integral.

For a term of order $k$ in Eq.~\eqref{eq:SM-M-Dyson}, the ordered
$u$ variables give
\begin{equation}
 \mathfrak M_q(M_{S,i}-\one)
 \le
 \sum_{k\ge1}\frac{K_q^k}{k!}
 =
 \ee^{K_q}-1.
\end{equation}
This proves Eq.~\eqref{eq:SM-history-majorant} and the absolute
convergence in $L_q$.  Since $q\ge1$, it also implies convergence in
operator norm.  Taking the adjoint reverses the order of every product
and conjugates its coefficient, neither of which changes the weighted
coefficient sum.  The same bound therefore holds for
$M_{S,i}^\dagger-\one$.
\end{proof}
We now use this bound to construct the random variable in
Lemma~\ref{lem:SM-ratio-sampler}.

\begin{proof}[Proof of Lemma~\ref{lem:SM-ratio-sampler}]
Define
\begin{equation}
 \eta_Q:=Q^dC_VgI_\beta,
 \qquad
 K_Q:=\frac{\eta_Q}{1-2d\eta_Q/Q^2}.
 \label{eq:SM-KQ}
\end{equation}
Equation~\eqref{eq:SM-uncollected-Jbound} gives
$\widehat\eta_Q\le\eta_Q$.  Since $g\le g_{\mathrm{pin}}$,
Eq.~\eqref{eq:SM-gpin} implies
\begin{equation}
 \eta_Q\le\frac{Q^2}{4d},
 \qquad
 2\eta_Q\le\log(1+\kappa_0),
 \qquad
 \kappa_0=\frac1{32\sqrt2}.
\end{equation}
In particular,
\begin{equation}
 \frac{\widehat\eta_Q}
 {1-2d\widehat\eta_Q/Q^2}
 \le K_Q
 \le2\eta_Q
 \le\log(1+\kappa_0).
\end{equation}
Lemma~\ref{lem:SM-uncollected-ratio} therefore gives
\begin{equation}
 \mathfrak M_Q(M_{S,i}-\one)
 \le\ee^{K_Q}-1
 \le\kappa_0.
 \label{eq:SM-path-mass-kappa}
\end{equation}

For $\zeta\in\{1,-1,\ii,-\ii\}$, define the nonnegative functions
\begin{equation}
 a_1(h)=(\Re a(h))_+,
 \qquad
 a_{-1}(h)=(\Re a(h))_-,
 \qquad
 a_{\ii}(h)=(\Im a(h))_+,
 \qquad
 a_{-\ii}(h)=(\Im a(h))_-.
 \label{eq:SM-phase-components}
\end{equation}
They satisfy
\begin{equation}
 a(h)
 =
 \sum_{\zeta\in\{1,-1,\ii,-\ii\}}
 \zeta a_\zeta(h)
\end{equation}
and
\begin{align}
 &\sum_{\zeta\in\{1,-1,\ii,-\ii\}}
 \int_{\mathfrak H_{S,i}}
 a_\zeta(h)Q^{|R(h)|}\,\nu_{S,i}(\dd h)
 \nonumber\\
 &\qquad=
 \int_{\mathfrak H_{S,i}}
 \bigl(|\Re a(h)|+|\Im a(h)|\bigr)
 Q^{|R(h)|}\,\nu_{S,i}(\dd h)
 \nonumber\\
 &\qquad\le
 \sqrt2\,\mathfrak M_Q(M_{S,i}-\one)
 \le\frac1{32}.
 \label{eq:SM-phase-mass}
\end{align}

For each $\zeta\in\{1,-1,\ii,-\ii\}$, assign the outcomes
$(h,\zeta)$ the density
\begin{equation}
 32a_\zeta(h)Q^{|R(h)|}
\end{equation}
with respect to $\nu_{S,i}$.  Return zero with the remaining
probability.  If $(h,\zeta)$ is selected, set
\begin{equation}
 Z=\frac{\zeta}{32}Q^{-|R|}P_R,
 \qquad
 R=R(h).
\end{equation}
Equation~\eqref{eq:SM-phase-mass} shows that the total assigned
probability is at most one.  Moreover,
\begin{align}
 \bbE Z
 &=
 \sum_{\zeta\in\{1,-1,\ii,-\ii\}}
 \int_{\mathfrak H_{S,i}}
 32a_\zeta(h)Q^{|R(h)|}
 \frac{\zeta}{32}Q^{-|R(h)|}P_{R(h)}
 \,\nu_{S,i}(\dd h)
 \nonumber\\
 &=
 \int_{\mathfrak H_{S,i}}
 a(h)P_{R(h)}\,\nu_{S,i}(\dd h)
 =
 M_{S,i}-\one.
\end{align}
Thus
\begin{equation}
 \bbE(\one+Z)=M_{S,i}.
\end{equation}
Since $R(h)\subseteq S$, every nonzero outcome has the form required
in Eq.~\eqref{eq:SM-Z}.  Applying the same construction to the
expansion of $M_{S,i}^\dagger-\one$ gives the second random variable,
and the two may be sampled independently.
\end{proof}
The proof above specifies the probability assigned to each expansion
path.  It remains to show that one sample can be generated without
enumerating all paths or computing the collected coefficients $m_R$.

\begin{lemma}
\label{lem:SM-ratio-generation}
The random variables for $M_{S,i}$ and $M_{S,i}^\dagger$ in
Lemma~\ref{lem:SM-ratio-sampler} can be generated without forming the
collected coefficients of either operator.  Suppose a proposed path
uses a term of order $k$ in Eq.~\eqref{eq:SM-M-Dyson}, and the $j$th
occurrence of $B_{S,i}$ uses $m_j$ nested commutators in
Eq.~\eqref{eq:SM-B-nested}.  Define its length by
\begin{equation}
 \ell_{\mathrm{path}}
 :=
 k+\sum_{j=1}^k m_j.
 \label{eq:SM-path-length}
\end{equation}
This count also includes proposals that eventually return zero.  There
are constants $c_0,C_0>0$, depending only on the fixed model
parameters, such that
\begin{equation}
 \bbP(\ell_{\mathrm{path}}>L)
 \le C_0\ee^{-c_0L}.
 \label{eq:SM-path-length-tail}
\end{equation}
Conditional on $\ell_{\mathrm{path}}=w$, the procedure uses
$\order(dw)$ Majorana labels and can be carried out in
$n^{\order(d)}\poly(w)$ arithmetic operations.
\end{lemma}
\begin{proof}
It is enough to consider $M_{S,i}-\one$.  The construction for
$M_{S,i}^\dagger-\one$ is obtained by reversing the products and
conjugating their coefficients.  If $\eta_Q=0$, return zero.  Assume
below that $\eta_Q>0$.

Let
\begin{equation}
 j_Q(u):=Q^dC_Vg\,\ee^{dC_Hu},
 \qquad
 \int_0^Tj_Q(u)\,\dd u=\eta_Q.
\end{equation}
For a labeled term $(T,\varphi)$ selected at time $u$, write
\begin{equation}
 z_{T,\varphi}(u)=v_Tc_{T,\varphi}(u).
\end{equation}
Suppose that a complete path $h$ selects $w=k+\sum_jm_j$ such terms,
with degrees $s_1,\ldots,s_w$ and times
$\tau_1,\ldots,\tau_w$.  If $m_{\mathrm{tot}}=\sum_jm_j$, set
\begin{equation}
 \Lambda(h)
 =
 2^{m_{\mathrm{tot}}}
 Q^{\sum_{\ell=1}^ws_\ell-2m_{\mathrm{tot}}}
 \prod_{\ell=1}^w
 |z_{T_\ell,\varphi_\ell}(\tau_\ell)|.
 \label{eq:SM-explicit-Lambda}
\end{equation}
Every nonzero commutator removes at least two repeated Majorana labels.
Equations~\eqref{eq:SM-uncollected-Jbound}
and~\eqref{eq:SM-path-commutator-bound} therefore give
\begin{equation}
 Q^{|R(h)|}|a(h)|
 \le
 \Lambda(h)
 \label{eq:SM-proposal-domination}
\end{equation}
and
\begin{equation}
 \int_{\mathfrak H_{S,i}}
 \Lambda(h)\,\nu_{S,i}(\dd h)
 \le
 \ee^{K_Q}-1
 \le
 \kappa_0.
 \label{eq:SM-proposal-mass}
\end{equation}
Assign the remaining mass to the zero outcome and normalize by
$\kappa_0$.  The resulting proposal distribution has density
$\Lambda(h)/\kappa_0$ on nonzero paths.

Given a nonzero proposed path $h$, retain it with probability
\begin{equation}
 \frac{
 Q^{|R(h)|}
 \bigl(|\Re a(h)|+|\Im a(h)|\bigr)
 }{
 \sqrt2\Lambda(h)
 }
 \le1.
 \label{eq:SM-proposal-acceptance}
\end{equation}
If the path is retained, choose its real or imaginary part with
probability proportional to the corresponding absolute value.  Use
the sign of that part to choose
$\zeta\in\{1,-1,\ii,-\ii\}$, and return
\begin{equation}
 Z=\frac{\zeta}{32}Q^{-|R(h)|}P_{R(h)}.
\end{equation}
If the path is not retained, return zero.  Since
$\sqrt2\kappa_0=1/32$, the real part contributes
\begin{equation}
 \frac{\Lambda(h)}{\kappa_0}
 \frac{Q^{|R(h)|}|\Re a(h)|}
 {\sqrt2\Lambda(h)}
 \frac{\sgn(\Re a(h))}{32}
 Q^{-|R(h)|}P_{R(h)}
 =
 \Re a(h)P_{R(h)}
\end{equation}
to the expectation.  The imaginary part similarly contributes
$\ii\Im a(h)P_{R(h)}$.  Integrating over $h$ gives
\begin{equation}
 \bbE Z
 =
 \int_{\mathfrak H_{S,i}}
 a(h)P_{R(h)}\,\nu_{S,i}(\dd h)
 =
 M_{S,i}-\one.
\end{equation}
Thus this acceptance step gives the distribution constructed in the
proof of Lemma~\ref{lem:SM-ratio-sampler}.

It remains to generate the proposal distribution without enumerating
the expansion paths.  Set
\begin{equation}
 K_{\max}:=\log(1+\kappa_0),
 \qquad
 r_Q:=\frac{2d\eta_Q}{Q^2}\le\frac12.
\end{equation}
Then
\begin{equation}
 K_Q=\frac{\eta_Q}{1-r_Q}\le K_{\max}.
\end{equation}
Choose the order $k\ge1$ in Eq.~\eqref{eq:SM-M-Dyson} with
probability
\begin{equation}
 \bbP(k)
 =
 \frac{K_{\max}^k}{\kappa_0k!}.
 \label{eq:SM-proposal-k}
\end{equation}
These probabilities sum to one because
$\ee^{K_{\max}}-1=\kappa_0$.

For each of the $k$ occurrences of $B_{S,i}$, continue with
probability $K_Q/K_{\max}$ and otherwise return zero.  Conditional on
continuing, choose the number $m$ of nested commutators with
probability
\begin{equation}
 \bbP(m)
 =
 (1-r_Q)r_Q^m.
 \label{eq:SM-proposal-m}
\end{equation}
Draw $m+1$ independent times from the density
$j_Q(t)/\eta_Q$, with one designated as the time $u$ of
$D_{S,i}(u)$.  Return zero unless $u$ is the largest sampled time.
Otherwise, sort the remaining times as
$t_1<\cdots<t_m$.  For fixed $m$, the resulting unnormalized density
on $0<t_1<\cdots<t_m<u<T$ is
\begin{equation}
 m!\left(\frac{2d}{Q^2}\right)^m
 j_Q(u)\prod_{j=1}^m j_Q(t_j).
 \label{eq:SM-inner-proposal-density}
\end{equation}
Indeed,
\begin{equation}
 K_Q(1-r_Q)=\eta_Q,
 \qquad
 r_Q=\frac{2d\eta_Q}{Q^2},
\end{equation}
while sorting the $m$ commutator times accounts for their $m!$
possible orders.  Thus Eq.~\eqref{eq:SM-inner-proposal-density}
is the bound used in Eq.~\eqref{eq:SM-one-B-path-bound}.

We next generate the discrete choices.  At the time $u$ associated
with $D_{S,i}(u)$, choose a labeled term whose output contains $i$ with
probability proportional to
$Q^{|T|}|z_{T,\varphi}(u)|$.  The bound
in Eq.~\eqref{eq:SM-uncollected-Jbound} supplies the normalization; any
unused probability returns zero.  Return zero also unless the output
support is contained in $S$.  The retained outcomes are exactly the
terms in $D_{S,i}(u)$.

Apply the nested commutators from the inside out.  Before the $j$th
application, the current monomial has at most $jd$ indices.  Choose
uniformly from $jd$ positions, letting an unused position return zero,
and then choose a labeled term whose output contains the selected
index.  Return zero unless the output support is contained in $S'$, the
selected index is the first shared index in the fixed ordering, and the
commutator is nonzero.  Each term in the nested commutator is then
counted exactly once.

All these choices use the nonnegative sums defining
$\widehat J_Q$.  Equation~\eqref{eq:SM-uncollected-Jbound} bounds their
total weight by $j_Q(t)$, with the remaining probability assigned to
zero.  The $jd$ possible positions and the coefficient $2$ of a
nonzero commutator give the multiplier
\begin{equation}
 \frac{2jd}{Q^2}
\end{equation}
in Eq.~\eqref{eq:SM-path-commutator-bound}.  Multiplying the selected
Majorana monomials in the nested order and then placing their indices
in canonical order determines $R(h)$ and $a(h)$.  The probabilities
above give precisely the value of $\Lambda(h)$ in
Eq.~\eqref{eq:SM-explicit-Lambda}.

Generate the $k$ occurrences independently and sort their $u$ times,
carrying the remaining data with them.  This accounts for the $k!$
orderings of the time variables in
Eq.~\eqref{eq:SM-M-Dyson}.  Together with
Eq.~\eqref{eq:SM-proposal-k}, the resulting density on nonzero paths
is $\Lambda(h)/\kappa_0$.  Thus the preceding acceptance step produces
the required random variable without computing the coefficients
$m_R$.

Finally, we bound the length of the generated path.  For any fixed
$0<\vartheta<\log2$, let
\begin{equation}
 G_\vartheta
 :=
 \sup_{0\le r\le1/2}
 \sum_{m\ge0}(1-r)r^m\ee^{\vartheta(1+m)}
 \le
 \frac{\ee^\vartheta}{1-\ee^\vartheta/2}
 <\infty.
\end{equation}
Generating an $m_j$ for every one of the $k$ occurrences, even when a
zero outcome has already occurred, can only increase
$\ell_{\mathrm{path}}$.  Therefore,
\begin{equation}
 \bbE\ee^{\vartheta\ell_{\mathrm{path}}}
 \le
 \sum_{k\ge1}
 \frac{K_{\max}^k}{\kappa_0k!}G_\vartheta^k
 =
 \frac{\exp(K_{\max}G_\vartheta)-1}{\kappa_0}
 <\infty.
 \label{eq:SM-generation-mgf}
\end{equation}
Markov's inequality proves
Eq.~\eqref{eq:SM-path-length-tail}.

Every selected interaction monomial has degree at most $d$.  A path
with $\ell_{\mathrm{path}}=w$ therefore contains $\order(dw)$ Majorana
labels.  At each occurrence, the interaction monomial and its output
labels can be sampled by enumerating at most $n^{\order(d)}$ choices;
the remaining operations are polynomial in $w$.  The total running
time for a path of length $w$ is therefore $n^{\order(d)}\poly(w)$.
Truncation and finite precision are treated in
Appendix~\ref{sec:SM-finite}.
\end{proof}

\section{Estimating the node weights}
\label{sec:SM-oracle}

The sampling tree in Appendix~\ref{sec:SM-soft-tree} assigns a weight
$W_v$ to each node.  This weight is the expected trace of the Gaussian
operator obtained by continuing from that node to a leaf.  Computing it
by summing over all descendant leaves would require enumerating the
tree.  We instead estimate the trace
\begin{equation}
 Z_{S,\sigma}
 =\Tr\left[GX^S\sigma(X^S)^\dagger G\right],
 \label{eq:SM-oracle-Z}
\end{equation}
where
\begin{equation}
 \sigma=\prod_{e=(a_e,b_e)\in\cM}
 (\one+t s_e\ii\gamma_{a_e}\gamma_{b_e}),
 \label{eq:SM-oracle-sigma}
\end{equation}
the pairs in $\cM$ are disjoint, and $s_e\in\{\pm1\}$.
Equation~\eqref{eq:SM-live-sandwich} gives
$Z_{S_v,\sigma_v}/2\le W_v\le3Z_{S_v,\sigma_v}/2$.
Thus a constant additive error in $\log Z_{S_v,\sigma_v}$ gives the
constant multiplicative estimate of $W_v$ used in
Appendix~\ref{sec:SM-finite}.  The estimate must hold uniformly over
the active set and all matchings that can occur in the tree.

\begin{theorem}
\label{thm:SM-oracle}
For fixed $0<\beta<\infty$, $D$, and $r_0$, there are constants
$0<t_0<1$ and $g_{\mathrm{log}}>0$, independent of $n$, with the following
property.  For $0\le t\le t_0$, $g\le g_{\mathrm{log}}$, any active set
$S$, and any disjoint matching, $Z_{S,\sigma}>0$, and a randomized
algorithm returns $\widehat L$ such that
\begin{equation}
 \bbP\left(
 |\widehat L-\log Z_{S,\sigma}|\le\eta
 \right)\ge1-\delta
 \label{eq:SM-oracle-guarantee}
\end{equation}
for $0<\eta,\delta\le1$.  Assume that $\beta$, $t$, and the local
coefficients can be evaluated to precision $2^{-p}$ in time polynomial
in $n$ and $p$.  Its running time is
\begin{equation}
 n^{\order(d)}\eta^{-2}
 \polylog\left(\frac{n}{\eta\delta}\right).
 \label{eq:SM-oracle-runtime}
\end{equation}
The constants are uniform in $S$ and the matching geometry.
For $\eta>1$, the algorithm may be run with $\eta=1$.
\end{theorem}

At $\beta=0$, $G=X^S=\one$ and $Z_{S,\sigma}=2^n$, so its
logarithm can be computed directly.
We first compute the trace and bound the two-point function of
$G\sigma G$.  We then expand $\log Z_{S,\sigma}$ about this Gaussian
operator, construct an estimator for the truncated expansion, and
choose the precision and number of samples needed for the theorem.

\subsection{\texorpdfstring{Trace and covariance of $G\sigma G$}{Trace and covariance of G sigma G}}
\label{sec:SM-oracle-Gaussian}

Set
\begin{equation}
 L^S=GX^SG^{-1},
 \qquad
 Z_\sigma=\Tr(G\sigma G),
 \qquad
 \omega_\sigma(Y)=\frac{\Tr(G\sigma GY)}{Z_\sigma}.
\end{equation}
Since $GX^S=L^SG$, cyclicity of the trace gives
\begin{equation}
 Z_{S,\sigma}
 =
 Z_\sigma\,
 \omega_\sigma\bigl((L^S)^\dagger L^S\bigr).
 \label{eq:SM-oracle-factorization}
\end{equation}
This separates the Gaussian operator $G\sigma G$ from the interaction
contained in $L^S$.  The present subsection computes $Z_\sigma$ and
bounds the absolute row and column sums of the two-point function of
$\omega_\sigma$.  The next subsection uses these two quantities to
expand the second term in Eq.~\eqref{eq:SM-oracle-factorization}.

The expansion in the next subsection contains Gaussian expectations
of products of Majorana operators.  By Wick's theorem, each such
expectation is a Pfaffian whose entries are two-point functions of
$\omega_\sigma$.  Since the Majorana labels are subsequently summed,
bounds on the individual entries are not sufficient.  We need uniform
bounds on the absolute row and column sums of the two-point function,
valid for every matching that can occur in the sampling tree.

To compute $Z_\sigma$ and the two-point function of $\omega_\sigma$,
we record both the conjugation action of $G\sigma G$ on the Majorana
operators and its scalar multiplier. For any invertible Gaussian operator $W$, conjugation by $W$ acts
linearly on the Majorana operators.  We denote the corresponding
matrix by $R(W)$:
\begin{equation}
 W^{-1}\bm\gamma W=R(W)\bm\gamma.
\end{equation}
Since $R(cW)=R(W)$ for every nonzero scalar $c$, the scalar multiplier
must be recorded separately.For an antisymmetric matrix $K$, set
\begin{equation}
 \cQ(K)=\frac14\bm\gamma^{\mathsf T}K\bm\gamma.
\end{equation}
The CAR give
\begin{equation}
 [\cQ(K),\bm\gamma]=-K\bm\gamma,
\end{equation}
and therefore
\begin{equation}
 \ee^{-\cQ(K)}\bm\gamma\ee^{\cQ(K)}
 =
 \ee^K\bm\gamma.
\end{equation}
It follows that
\begin{equation}
 R(\ee^{\cQ(K)})=\ee^K.
\end{equation}
The definition of $R(W)$ also gives
\begin{equation}
 R(W_1W_2)=R(W_1)R(W_2).
 \label{eq:SM-transfer}
\end{equation}
After writing $\sigma$ as a scalar times a product of quadratic
exponentials, these identities determine $R(G\sigma G)$.

We now write each soft pinning operator as a scalar times a quadratic
exponential. Let
\begin{equation}
 B_e=\ii\gamma_{a_e}\gamma_{b_e},
 \qquad
 \theta=\operatorname{arctanh}t.
\end{equation}
Since $B_e^2=\one$,
\begin{align}
 \ee^{\theta s_eB_e}
 &=
 \cosh\theta\,\one+s_e\sinh\theta\,B_e
 \nonumber\\
 &=
 \frac{1}{\sqrt{1-t^2}}
 \left(\one+t s_eB_e\right).
\end{align}
Therefore,
\begin{equation}
 \one+t s_eB_e
 =
 \sqrt{1-t^2}\,\ee^{\theta s_eB_e}.
 \label{eq:SM-soft-exp}
\end{equation}
Define
\begin{equation}
 \widetilde\sigma
 :=
 \prod_{e\in\cM}\ee^{\theta s_eB_e}.
\end{equation}
Then
\begin{equation}
 \sigma
 =
 (1-t^2)^{|\cM|/2}\widetilde\sigma.
 \label{eq:SM-sigma-exponential}
\end{equation}
Since scalar multipliers do not affect the transfer matrix, set
\begin{equation}
 S_\cM:=R(\widetilde\sigma),
 \qquad
 E:=R(G)=\ee^{-TK_0},
 \qquad
 T=\frac\beta2.
\end{equation}
The product rule in Eq.~\eqref{eq:SM-transfer} gives
\begin{equation}
 R(G\sigma G)=ES_\cM E.
 \label{eq:SM-Gsigma-transfer}
\end{equation}

Each exponential in $\widetilde\sigma$ acts only on the two Majorana
operators in its pair.  Since the pairs in $\cM$ are disjoint,
$S_\cM$ is block diagonal, with one $2\times2$ block for each pair
and the identity on the unpaired labels.  The block associated with a
pair has eigenvalues $\ee^{2\theta}$ and $\ee^{-2\theta}$, so
$S_\cM$ is positive definite. The scalar $(1-t^2)^{|\cM|/2}$ contributes
$\frac{|\cM|}{2}\log(1-t^2)$ to $\log Z_\sigma$.

\begin{lemma}
\label{lem:SM-baseline-logtrace}
The Gaussian trace $Z_\sigma$ satisfies
\begin{equation}
 \log Z_\sigma
 =
 \frac{|\cM|}{2}\log(1-t^2)
 +\frac12\log\det(\one+ES_\cM E).
 \label{eq:SM-baseline-logdet}
\end{equation}
Moreover,
\begin{equation}
 \ee^{-\beta C_H-2\theta}\one
 \preceq
 ES_\cM E
 \preceq
 \ee^{\beta C_H+2\theta}\one,
 \qquad
 \theta=\operatorname{arctanh}t.
 \label{eq:SM-baseline-spectrum}
\end{equation}
In particular, $\one+ES_\cM E$ is positive definite, with eigenvalues in the interval
\begin{equation}
 \left[
 1+\ee^{-\beta C_H-2\theta},
 \,
 1+\ee^{\beta C_H+2\theta}
 \right].
\end{equation}
For $t\le t_0<1$, this interval depends only on the fixed parameters
and is uniform in $\cM$.
\end{lemma}

\begin{proof}
Let
\begin{equation}
 Y=G\widetilde\sigma G.
\end{equation}
This is a positive Gaussian operator, and
Eq.~\eqref{eq:SM-transfer} gives
\begin{equation}
 R(Y)=ES_\cM E.
\end{equation}
The Gaussian trace identity gives
\begin{equation}
 (\Tr Y)^2
 =
 \det(\one+R(Y))
 =
 \det(\one+ES_\cM E).
\end{equation}
Since $Y$ is positive, $\Tr Y>0$, and hence
\begin{equation}
 \log\Tr Y
 =
 \frac12\log\det(\one+ES_\cM E).
\end{equation}
Using
\begin{equation}
 Z_\sigma
 =
 (1-t^2)^{|\cM|/2}\Tr Y
\end{equation}
now proves Eq.~\eqref{eq:SM-baseline-logdet}.

It remains to prove the spectral bound.  Each $2\times2$ block of
$S_\cM$ has eigenvalues $\ee^{2\theta}$ and
$\ee^{-2\theta}$, while $S_\cM$ acts as the identity on the unpaired
labels.  Therefore,
\begin{equation}
 \ee^{-2\theta}\one
 \preceq
 S_\cM
 \preceq
 \ee^{2\theta}\one.
\end{equation}
Since $E$ is Hermitian, multiplying each term on the left and right
by $E$ preserves the inequalities and gives
\begin{equation}
 \ee^{-2\theta}E^2
 \preceq
 ES_\cM E
 \preceq
 \ee^{2\theta}E^2.
\end{equation}
The spectrum of $K_0$ lies in $[-C_H,C_H]$, and therefore
\begin{equation}
 \ee^{-\beta C_H}\one
 \preceq
 E^2
 \preceq
 \ee^{\beta C_H}\one.
\end{equation}
Combining these inequalities gives
\begin{equation}
 \ee^{-\beta C_H-2\theta}\one
 \preceq
 ES_\cM E
 \preceq
 \ee^{\beta C_H+2\theta}\one,
\end{equation}
which is Eq.~\eqref{eq:SM-baseline-spectrum}.

It follows that the eigenvalues of $\one+ES_\cM E$ lie in
\begin{equation}
 \left[
 1+\ee^{-\beta C_H-2\theta},
 \,
 1+\ee^{\beta C_H+2\theta}
 \right].
\end{equation}
For $t\le t_0<1$, we have
$\theta\le\operatorname{arctanh}t_0$, so this interval depends only on
the fixed parameters and is uniform in $\cM$.
\end{proof}

We next turn to the two-point function of $\omega_\sigma$.  Define
\begin{equation}
 (C_\sigma)_{ab}
 :=
 \omega_\sigma(\gamma_a\gamma_b).
 \label{eq:SM-Csigma-def}
\end{equation}
Let
\begin{equation}
 R:=ES_\cM E=R(G\sigma G).
\end{equation}
The transfer relation
\begin{equation}
 (G\sigma G)^{-1}\bm\gamma(G\sigma G)=R\bm\gamma
\end{equation}
and cyclicity of the trace give
\begin{equation}
 RC_\sigma=C_\sigma^{\mathsf T}.
\end{equation}
On the other hand, the CAR imply
\begin{equation}
 C_\sigma+C_\sigma^{\mathsf T}=2\one.
\end{equation}
Combining these identities gives
\begin{equation}
 (\one+R)C_\sigma=2\one,
\end{equation}
and hence
\begin{equation}
 C_\sigma
 =
 2(\one+ES_\cM E)^{-1}.
 \label{eq:SM-Csigma}
\end{equation}

The expansion in the next subsection contains sums over the Majorana
labels in these two-point functions.  The entrywise bound
\begin{equation}
 |(C_\sigma)_{ab}|\le1
\end{equation}
is not sufficient for this purpose, since it only gives
\begin{equation}
 \sum_b|(C_\sigma)_{ab}|\le N.
\end{equation}
We instead need bounds on the absolute row and column sums that are
independent of the system size.  For a matrix $A$, define
\begin{equation}
 \norm{A}_{\mathrm{Schur}}
 :=
 \max\left\{
 \max_a\sum_b|A_{ab}|,
 \,
 \max_b\sum_a|A_{ab}|
 \right\}.
 \label{eq:SM-Schur-norm}
\end{equation}

The matching $\cM$ may connect Majorana labels at arbitrarily large
distances. We begin with the case
$\cM=\varnothing$.  In this case,
\begin{equation}
 S_\cM=\one,
 \qquad
 C_\sigma=2F_0,
 \qquad
 F_0=(\one+\ee^{-\beta K_0})^{-1}.
\end{equation}
The following lemma derives spatial decay of $F_0$ and, as a
consequence, a bound on $\norm{F_0}_{\mathrm{Schur}}$ that is
independent of the system size.  We will then compare the general matching with this case.

\begin{lemma}
\label{lem:SM-free-covariance-decay}
There are constants $C_\beta,\mu_\beta,L_\beta>0$, depending only on
the fixed model parameters, such that
\begin{equation}
 |(F_0)_{ab}|
 \le
 C_\beta
 \ee^{-\mu_\beta\operatorname{dist}(a,b)},
 \qquad
 \norm{F_0}_{\mathrm{Schur}}
 \le
 L_\beta,
 \label{eq:SM-Lbeta}
\end{equation}
where
\begin{equation}
 F_0=(\one+\ee^{-\beta K_0})^{-1}.
\end{equation}
In particular, the constants are independent of the system size.
\end{lemma}

\begin{proof}
Let
\begin{equation}
 f(x)=\frac{1}{1+\ee^{-\beta x}},
\end{equation}
so that $F_0=f(K_0)$.  The function $f$ is analytic in a complex
neighborhood of the spectral interval $[-C_H,C_H]$.  Because $f$ is analytic in a complex neighborhood of
$[-C_H,C_H]$, there are degree-$k$ polynomials $p_k$ and constants
$C,\mu>0$ such that
\begin{equation}
 \sup_{|x|\le C_H}|f(x)-p_k(x)|
 \le
 C\ee^{-\mu k}.
 \label{eq:SM-F0-poly-approx}
\end{equation}
Since $K_0$ is Hermitian, write
\begin{equation}
 K_0
 =
 U\operatorname{diag}(\lambda_1,\ldots,\lambda_N)U^\dagger,
 \qquad
 \lambda_j\in[-C_H,C_H].
\end{equation}
Then
\begin{equation}
 F_0-p_k(K_0)
 =
 U\operatorname{diag}\bigl(
 f(\lambda_j)-p_k(\lambda_j)
 \bigr)_{j=1}^N U^\dagger.
\end{equation}
It follows directly from Eq.~\eqref{eq:SM-F0-poly-approx} that
\begin{equation}
 \norm{F_0-p_k(K_0)}_{\mathrm{op}}
 \le
 C\ee^{-\mu k}.
 \label{eq:SM-F0-matrix-approx}
\end{equation}

Suppose first that $r_0>0$, and set $R_0=2r_0$.  The matrix $K_0$
has range at most $R_0$.  Consequently, $K_0^j$ has range at most
$jR_0$, and hence
\begin{equation}
 (p_k(K_0))_{ab}=0
 \qquad\text{if}\qquad
 \operatorname{dist}(a,b)>kR_0.
 \label{eq:SM-polynomial-range}
\end{equation}
Let
\begin{equation}
 r=\operatorname{dist}(a,b)>R_0
\end{equation}
and choose
\begin{equation}
 k=\left\lceil\frac{r}{R_0}\right\rceil-1.
\end{equation}
Then $kR_0<r$, so Eq.~\eqref{eq:SM-polynomial-range} and
Eq.~\eqref{eq:SM-F0-matrix-approx} give
\begin{align}
 |(F_0)_{ab}|
 &=
 |(F_0-p_k(K_0))_{ab}|
 \nonumber\\
 &\le
 \norm{F_0-p_k(K_0)}_{\mathrm{op}}
 \nonumber\\
 &\le
 C\ee^{-\mu k}
 \le
 C\ee^\mu
 \exp\left(-\frac{\mu}{R_0}r\right).
\end{align}
Choose $C_\beta$ also so that
$C_\beta\ee^{-\mu_\beta R_0}\ge1$.  Since
$\norm{F_0}_{\mathrm{op}}\le1$, the same bound holds for $r\le R_0$. Thus
\begin{equation}
 |(F_0)_{ab}|
 \le
 C_\beta
 \ee^{-\mu_\beta\operatorname{dist}(a,b)}
\end{equation}
for all $a,b$.

If $r_0=0$, then $K_0$ and $F_0=f(K_0)$ are block diagonal with
respect to the lattice sites.  The same conclusion follows
immediately after adjusting the constants.

It remains to sum the entrywise bound.  The number of Majorana labels
at lattice distance $r$ from a fixed label is at most \(C_D(r+1)^{D-1}\), where $C_D$ is independent of the system size.  Therefore,
\begin{align}
 \sum_b |(F_0)_{ab}|
 \le
 C_\beta C_D
 \sum_{r\ge0}
 (r+1)^{D-1}\ee^{-\mu_\beta r}
 \nonumber:=L_\beta.
\end{align}
The series converges.  Since $F_0$ is Hermitian, the same bound holds
for the absolute column sums.  Hence
\begin{equation}
 \norm{F_0}_{\mathrm{Schur}}\le L_\beta.
\end{equation}
\end{proof}

We now compare the two-point function for a general matching with the
empty-matching case.  On the two-dimensional subspace associated with
a pair $e$, the corresponding block of $S_\cM$ has diagonal entries
$\cosh(2\theta)$ and off-diagonal entries of absolute value
$\sinh(2\theta)$.  Hence the absolute row and column sums of the
corresponding block of $S_\cM-\one$ are
\begin{equation}
 \cosh(2\theta)-1+\sinh(2\theta)
 =
 \ee^{2\theta}-1
 =
 \frac{2t}{1-t}.
\end{equation}
Since the pairs in $\cM$ are disjoint, these blocks do not overlap.
Therefore,
\begin{equation}
 \norm{S_\cM-\one}_{\mathrm{Schur}}
 \le
 \frac{2t}{1-t}.
 \label{eq:SM-matching-Schur}
\end{equation}
In particular, this bound is independent of the number and separation
of the pairs.

Using
\begin{equation}
 F_0=(\one+E^2)^{-1},
\end{equation}
we can write
\begin{equation}
 \one+ES_\cM E
 =
 (\one+E^2)
 \left[
 \one+F_0E(S_\cM-\one)E
 \right].
 \label{eq:SM-covariance-factorization}
\end{equation}
The Schur norm is submultiplicative.  Moreover, the bound on the
absolute row and column sums of $K_0$ gives
\begin{equation}
 \norm{E}_{\mathrm{Schur}}
 =
 \norm{\ee^{-TK_0}}_{\mathrm{Schur}}
 \le
 \ee^{TC_H}.
\end{equation}
Combining this estimate with
Eqs.~\eqref{eq:SM-Lbeta} and
\eqref{eq:SM-matching-Schur}, we obtain
\begin{align}
 \norm{
 F_0E(S_\cM-\one)E
 }_{\mathrm{Schur}}
 \nonumber
 \le
 L_\beta\ee^{\beta C_H}
 \frac{2t}{1-t}.
 \label{eq:SM-matching-correction}
\end{align}

Choose $t_0>0$ such that
\begin{equation}
 L_\beta\ee^{\beta C_H}
 \frac{2t_0}{1-t_0}
 \le
 \frac12.
 \label{eq:SM-t0}
\end{equation}
Then, for $t\le t_0$,
\begin{equation}
 \norm{
 F_0E(S_\cM-\one)E
 }_{\mathrm{Schur}}
 \le\frac12.
\end{equation}
The Neumann series therefore gives
\begin{equation}
 \norm{
 \left[
 \one+F_0E(S_\cM-\one)E
 \right]^{-1}
 }_{\mathrm{Schur}}
 \le2.
\end{equation}

Equations~\eqref{eq:SM-Csigma} and
\eqref{eq:SM-covariance-factorization} now give
\begin{equation}
 C_\sigma
 =
 2\left[
 \one+F_0E(S_\cM-\one)E
 \right]^{-1}F_0.
\end{equation}
Using Eq.~\eqref{eq:SM-Lbeta} once more,
\begin{equation}
 \norm{C_\sigma}_{\mathrm{Schur}}
 \le
 4L_\beta.
 \label{eq:SM-C-Schur}
\end{equation}
Thus the absolute sum over either Majorana index in $(C_\sigma)_{ab}$ is bounded independently of the system size and of
the number and separation of the matching pairs.  This is the
two-point function bound used in the expansion of the next
subsection.

\subsection{\texorpdfstring{Expansion of $\log Z_{S,\sigma}$}
{Expansion of log Z}}
\label{sec:SM-oracle-expansion}

Recall that
\begin{equation}
 Z_{S,\sigma}
 =
 Z_\sigma\,
 \omega_\sigma\bigl((L^S)^\dagger L^S\bigr),
 \qquad
 L^S=GX^SG^{-1}.
 \label{eq:SM-contour-base}
\end{equation}
Therefore,
\begin{equation}
 \log Z_{S,\sigma}
 =
 \log Z_\sigma
 +
 \log\omega_\sigma\bigl((L^S)^\dagger L^S\bigr).
 \label{eq:SM-log-factorization}
\end{equation}
The Gaussian formulas above determine $\log Z_\sigma$ and bound the
two-point function of $\omega_\sigma$.  It remains to expand
$\log\omega_\sigma((L^S)^\dagger L^S)$ in the interaction.

Define
\begin{equation}
 A_{\mathrm L}^S(u)
 :=
 G^{-1}A^S(u)^\dagger G,
 \qquad
 A_{\mathrm R}^S(u)
 :=
 GA^S(u)G^{-1}.
 \label{eq:SM-two-branch-interactions}
\end{equation}
The subscripts $\mathrm L$ and $\mathrm R$ refer to the two sides of
the product $(L^S)^\dagger L^S$.  Using
$L^S=GX^SG^{-1}$ and the time-ordered exponential defining $X^S$, we
obtain
\begin{align}
 (L^S)^\dagger
 =
 \widetilde\cT
 \exp\left[
 -\int_0^T A_{\mathrm L}^S(u)\,\dd u
 \right],
 \nonumber
 L^S=
 \cT
 \exp\left[
 -\int_0^T A_{\mathrm R}^S(u)\,\dd u
 \right].
 \label{eq:SM-two-branch-Dyson}
\end{align}
Here $\widetilde\cT$ orders earlier times to the left, whereas
$\cT$ orders later times to the left.
To treat the two ordered exponentials in a single expansion, introduce
two copies of the time interval,
\begin{equation}
 \mathfrak C
 =
 [0,T]_{\mathrm L}
 \sqcup
 [0,T]_{\mathrm R},
 \qquad
 |\mathfrak C|=2T=\beta.
 \label{eq:SM-two-branch-domain}
\end{equation}
For $z\in\mathfrak C$, define
\begin{equation}
 A(z)
 =
 \begin{cases}
  A_{\mathrm L}^S(u),
   & z=(\mathrm L,u),\\
  A_{\mathrm R}^S(u),
   & z=(\mathrm R,u).
 \end{cases}
 \label{eq:SM-branch-interaction}
\end{equation}
Write its Majorana expansion as
\begin{equation}
 A(z)=\sum_R v_R(z)P_R.
\end{equation}
Every monomial in this expansion has degree at most $d$.  Moreover,
there is a constant
\begin{equation}
 J_*=J_*(\beta,D,r_0)
\end{equation}
such that
\begin{equation}
 \sup_{z\in\mathfrak C}
 \max_a
 \sum_{R\ni a}|v_R(z)|
 \le
 J_*g.
 \label{eq:SM-contour-local}
\end{equation}
Indeed, restricting $A(u)$ to supports contained in $S$ cannot
increase the sum on the left.  Conjugation by $G$ or $G^{-1}$ acts
linearly on the Majorana operators, with absolute row and column sums
bounded by $\ee^{TC_H}$.  The same determinant estimate as in
Eq.~\eqref{eq:SM-uncollected-Jbound} then gives
Eq.~\eqref{eq:SM-contour-local}.

Introduce a scalar parameter $\lambda$ multiplying the interactions
on both sides of $(L^S)^\dagger L^S$, and define
\begin{align}
 F(\lambda)
 :=
 Z_\sigma\,
 \omega_\sigma\Bigg[
 \widetilde\cT
 \exp\left(
 -\lambda\int_0^T A_{\mathrm L}^S(u)\,\dd u
 \right)
 \times
 \cT
 \exp\left(
 -\lambda\int_0^T A_{\mathrm R}^S(u)\,\dd u
 \right)
 \Bigg].
 \label{eq:SM-Flambda}
\end{align}
For finite $n$, $F(\lambda)$ is entire.  At the two values relevant
here,
\begin{equation}
 F(0)=Z_\sigma,
 \qquad
 F(1)=Z_{S,\sigma}.
\end{equation}
For real $\lambda$, the two ordered exponentials in
Eq.~\eqref{eq:SM-Flambda} are adjoints of one another.  Their product
is positive, and hence
\begin{equation}
 F(\lambda)>0.
\end{equation}
In particular, $F$ has no zeros in a sufficiently small neighborhood
of the origin.  We may therefore choose an analytic logarithm in this
neighborhood and write
\begin{equation}
 \log F(\lambda)
 =
 \log Z_\sigma+\sum_{s\ge1}c_s\lambda^s.
 \label{eq:SM-log-Taylor}
\end{equation}

We next identify the coefficients $c_s$.  Expand the two ordered
exponentials in Eq.~\eqref{eq:SM-Flambda}, and fix $s$ interaction
terms at $z_1,\ldots,z_s\in\mathfrak C$, with Majorana monomials
$P_{R_1},\ldots,P_{R_s}$.  Their order is inherited from
Eq.~\eqref{eq:SM-Flambda}: the terms on the left branch come first in
increasing time order, followed by the terms on the right branch in
decreasing time order.  Within each Majorana monomial, the indices are
placed in canonical order.  Equal times on the same branch are ordered
by the vertex labels.

For a partition $\pi$ of $[s]$, let $|\pi|$ denote its number of
blocks.  The ordered cumulant of these interaction terms is
\begin{equation}
 \kappa_s(P_{R_1},\ldots,P_{R_s})
 =
 \sum_{\pi}
 (-1)^{|\pi|-1}(|\pi|-1)!
 \prod_{B\in\pi}
 \omega_\sigma\left(
 \prod_{j\in B}^{\longrightarrow}P_{R_j}
 \right).
 \label{eq:SM-connected-cumulant}
\end{equation}
The arrow indicates that each product inherits the order specified
above.  This combination subtracts the products of moments associated
with all nontrivial partitions of the vertices.  After multiplying by
the interaction coefficients, summing the monomial labels, and
integrating the variables $z_1,\ldots,z_s$, these cumulants give the
coefficient $c_s$.

By Wick's theorem, each moment in
Eq.~\eqref{eq:SM-connected-cumulant} is a Pfaffian of two-point
functions.  To isolate its connected part, introduce interpolation
variables $w_{ij}$ multiplying the covariance entries between
vertices $i$ and $j$, while leaving the covariance entries within each
vertex unchanged.  When the variables $w_{ij}$ disconnect the
vertices into separate groups, the Pfaffian factorizes over those
groups.  No additional exchange sign appears in this factorization
because every interaction monomial has even degree.

Applying the BKAR forest formula to
Eq.~\eqref{eq:SM-connected-cumulant} removes the disconnected
contributions and expresses the cumulant as a sum over spanning
trees~\cite{BrydgesKennedy1987,AbdesselamRivasseau1995}.

For each vertex $i$, write $d_i=|R_i|$ and
\begin{equation}
R_i=\{r_{i,1}<\cdots<r_{i,d_i}\},
\qquad
P_{R_i}
=
\ii^{d_i/2}
\gamma_{r_{i,1}}\cdots\gamma_{r_{i,d_i}}.
\end{equation}
Thus $\alpha\in[d_i]$ labels a position in the ordered monomial at
vertex $i$.  Let $\prec$ denote the ordering of all these positions
fixed above, and write \(r(i,\alpha)=r_{i,\alpha}\).

Let $\mathsf{Tree}_s$ be the set of labeled trees on $[s]$, and fix
$\mathcal T\in\mathsf{Tree}_s$.  For every edge $e=\{i,j\}$, choose
one position $\alpha_{i,e}$ at vertex $i$ and one position
$\alpha_{j,e}$ at vertex $j$.  Positions chosen by different edges
incident to the same vertex must be distinct.  Let
\begin{equation}
\Phi(\mathcal T;\bm R),
\qquad
\bm R=(R_1,\ldots,R_s),
\end{equation}
denote the set of all such choices.

For $\phi\in\Phi(\mathcal T;\bm R)$, an edge $e=\{i,j\}$ selects the
two positions \((i,\alpha_{i,e})\) and \((j,\alpha_{j,e})\). Write them as $x_e\prec y_e$ according to the ordering fixed above,
and define
\begin{equation}
C_{\phi(e)}
=
(C_\sigma)_{r(x_e),r(y_e)}.
\label{eq:SM-selected-covariance}
\end{equation}

We next define the matrix formed by the positions not selected by the
tree edges.  For
$\bm u=(u_e)_{e\in\mathcal T}\in[0,1]^{s-1}$, set
\begin{equation}
w_{ij}^{\mathcal T}(\bm u)
=
\begin{cases}
1,
&i=j,\\[1mm]
\displaystyle
\min_{e\in\operatorname{path}_{\mathcal T}(i,j)}u_e,
&i\ne j,
\end{cases}
\label{eq:SM-BKAR-interpolation}
\end{equation}
where $\operatorname{path}_{\mathcal T}(i,j)$ is the unique path from
$i$ to $j$ in $\mathcal T$.

Remove the two positions selected by every tree edge.  The remaining positions, with their induced ordering, index the
antisymmetric matrix $K^{\mathcal T,\phi}(\bm u)$.  For two remaining
positions $x=(i,\alpha)\prec y=(j,\beta)$, define
\begin{equation}
\left(
K^{\mathcal T,\phi}(\bm u)
\right)_{x,y}
=
w_{ij}^{\mathcal T}(\bm u)
(C_\sigma)_{r(i,\alpha),r(j,\beta)}.
\label{eq:SM-K-entry}
\end{equation}
The entries below the diagonal are fixed by antisymmetry.

With these definitions, the BKAR forest formula gives
\begin{equation}
\begin{aligned}
\kappa_s(P_{R_1},\ldots,P_{R_s})
={}&
\sum_{\mathcal T\in\mathsf{Tree}_s}
\int_{[0,1]^{s-1}}\dd\bm u
\sum_{\phi\in\Phi(\mathcal T;\bm R)}
\varepsilon(\mathcal T,\phi) \\
&\times
\left(\prod_{e\in\mathcal T}C_{\phi(e)}\right)
\operatorname{Pf}K^{\mathcal T,\phi}(\bm u).
\end{aligned}
\label{eq:SM-BKAR}
\end{equation}
Here $\varepsilon(\mathcal T,\phi)$ denotes the product of the sign
determined by the selected positions in the Pfaffian expansion and
the phases $\ii^{d_i/2}$ in the monomials $P_{R_i}$; hence
$|\varepsilon(\mathcal T,\phi)|=1$.

Positions at different vertices may carry the same physical Majorana
label.  The definition of $\Phi(\mathcal T;\bm R)$ only prevents one
position from being selected by two different tree edges.

For fixed $\mathcal T$, $\bm z$, and $\bm u$, define
\begin{equation}
\begin{aligned}
\cS_{s,\mathcal T}(\bm z,\bm u)
:={}&
\sum_{\bm R}
\left(\prod_{j=1}^s v_{R_j}(z_j)\right)
\sum_{\phi\in\Phi(\mathcal T;\bm R)}
\varepsilon(\mathcal T,\phi) \\
&\times
\left(\prod_{e\in\mathcal T}C_{\phi(e)}\right)
\operatorname{Pf}K^{\mathcal T,\phi}(\bm u).
\end{aligned}
\label{eq:SM-tree-integrand}
\end{equation}
After summing over the trees and integrating the interaction
variables, the coefficient in
Eq.~\eqref{eq:SM-log-Taylor} is
\begin{equation}
c_s
=
\frac{(-1)^s}{s!}
\sum_{\mathcal T\in\mathsf{Tree}_s}
\int_{\mathfrak C^s}\dd\bm z
\int_{[0,1]^{s-1}}\dd\bm u\,
\cS_{s,\mathcal T}(\bm z,\bm u).
\label{eq:SM-cs-outer}
\end{equation}
For $s=1$, $\mathsf{Tree}_1$ contains one tree with no edges, and the
integral over $\bm u$ is absent.  The sign $(-1)^s$ comes from the
minus sign associated with each interaction insertion.

The tree edges contribute the explicit two-point entries
$C_{\phi(e)}$, which will be summed using
Eq.~\eqref{eq:SM-C-Schur}.  The other two-point entries form
$\operatorname{Pf}K^{\mathcal T,\phi}(\bm u)$.
The following lemma bounds this Pfaffian directly.
\begin{lemma}
\label{lem:SM-replicated-pfaffian}
Let $\omega$ be an even quasifree state of Majorana operators
$\gamma_1,\ldots,\gamma_N$, and for $f\in\mathbb R^N$ write
$\gamma(f)=\sum_a f_a\gamma_a$.  Let
$w=(w_{ij})_{i,j=1}^s$ be a real positive semidefinite matrix with
$w_{ii}=1$.

Consider a finite ordered set $\mathcal I$ of positions
$x=(i,\alpha)$, where $i\in[s]$, and assign a vector
$f_x\in\mathbb R^N$ with $\norm{f_x}_2\le1$ to each position.
For $x=(i,\alpha)\prec y=(j,\beta)$, define the upper-triangular
entries of an antisymmetric matrix $K$ by
\begin{equation}
K_{x,y}
=
w_{ij}\,
\omega\left(\gamma(f_x)\gamma(f_y)\right).
\label{eq:SM-replicated-K}
\end{equation}
Then every even principal submatrix $K_{\mathcal J}$, with
$\mathcal J\subseteq\mathcal I$, satisfies
$|\operatorname{Pf}K_{\mathcal J}|\le1$.
\end{lemma}

\begin{proof}
Since $w$ is positive semidefinite, it is a Gram matrix: there are
vectors $u_1,\ldots,u_s\in\mathbb R^r$ such that
$w_{ij}=\langle u_i,u_j\rangle$.  Since $w_{ii}=1$, each $u_i$ is a
unit vector.

Consider the enlarged Majorana system with one-particle space
$\mathbb R^r\otimes\mathbb R^N$.  Let $\widehat\omega$ be the even
Gaussian state whose two-point matrix is the direct sum of $r$ copies
of that of $\omega$, and write $\widehat\gamma(h)$ for the Majorana
operator associated with
$h\in\mathbb R^r\otimes\mathbb R^N$.  Then
\begin{equation}
\widehat\omega\left(
\widehat\gamma(u\otimes f)\widehat\gamma(v\otimes g)
\right)
=
\langle u,v\rangle\,
\omega\left(\gamma(f)\gamma(g)\right).
\label{eq:SM-replicated-two-point}
\end{equation}

For a position $x=(i,\alpha)$, set
$\widehat f_x=u_i\otimes f_x$.  If
$x=(i,\alpha)\prec y=(j,\beta)$, then
Eq.~\eqref{eq:SM-replicated-two-point} gives
\begin{equation}
\widehat\omega\left(
\widehat\gamma(\widehat f_x)
\widehat\gamma(\widehat f_y)
\right)
=
w_{ij}\,
\omega\left(\gamma(f_x)\gamma(f_y)\right)
=
K_{x,y}.
\end{equation}
Thus every principal submatrix $K_{\mathcal J}$ is the ordered
two-point matrix of the operators
$\widehat\gamma(\widehat f_x)$ with $x\in\mathcal J$.
The CAR imply
$\norm{\widehat\gamma(\widehat f_x)}_{\mathrm{op}}
=\norm{\widehat f_x}_2
=\norm{u_i}_2\norm{f_x}_2\le1$.

Let $\mathcal J\subseteq\mathcal I$ have even cardinality, with the
order inherited from $\mathcal I$.  Wick's theorem gives
\begin{equation}
\begin{aligned}
\left|\operatorname{Pf}K_{\mathcal J}\right|
&=
\left|
\widehat\omega\left(
\prod_{x\in\mathcal J}^{\longrightarrow}
\widehat\gamma(\widehat f_x)
\right)
\right| \\
&\le
\norm{
\prod_{x\in\mathcal J}^{\longrightarrow}
\widehat\gamma(\widehat f_x)
}_{\mathrm{op}}
\le
\prod_{x\in\mathcal J}
\norm{\widehat\gamma(\widehat f_x)}_{\mathrm{op}}
\le1.
\end{aligned}
\end{equation}
\end{proof}
We apply the lemma to the matrix
$K^{\mathcal T,\phi}(\bm u)$.  To verify its assumption on
$w^{\mathcal T}(\bm u)$, for $\tau\in[0,1]$ let
$\mathcal T_\tau$ contain the edges $e$ of $\mathcal T$ satisfying
$u_e\ge\tau$.  Then
\begin{equation}
w_{ij}^{\mathcal T}(\bm u)
=
\int_0^1
\boldsymbol{1}\{i\text{ and }j
\text{ belong to the same component of }\mathcal T_\tau\}
\,\dd\tau.
\label{eq:SM-BKAR-positive}
\end{equation}
For fixed $\tau$, two vertices contribute $1$ to the matrix inside
the integral exactly when they belong to the same connected component
of $\mathcal T_\tau$.  After grouping the vertices by connected
components, this matrix is block diagonal.  The block corresponding
to a component $B$ is
$\boldsymbol{1}_B\boldsymbol{1}_B^{\mathsf T}$ and is therefore
positive semidefinite. Hence $w^{\mathcal T}(\bm u)$ is positive
semidefinite and has unit diagonal.

For every position $(i,\alpha)$ not selected by a tree edge, take
$f_{i,\alpha}$ to be the standard basis vector corresponding to the
Majorana label $r(i,\alpha)$.  Lemma
\ref{lem:SM-replicated-pfaffian} then gives
\begin{equation}
\left|
\operatorname{Pf}K^{\mathcal T,\phi}(\bm u)
\right|
\le1.
\label{eq:SM-pf-bound}
\end{equation}
Different positions may be assigned the same standard basis vector,
so this also covers repeated physical Majorana labels.  If no
positions remain, the same statement follows from
$\operatorname{Pf}(\varnothing)=1$.

For fixed $\mathcal T$ and $\bm z$, define
\begin{equation}
\cA_{\mathcal T}(\bm z)
:=
\sum_{\bm R}
\left(\prod_{j=1}^s |v_{R_j}(z_j)|\right)
\sum_{\phi\in\Phi(\mathcal T;\bm R)}
\prod_{e\in\mathcal T}|C_{\phi(e)}|.
\label{eq:SM-tree-absolute-sum}
\end{equation}
The Pfaffian bound in Eq.~\eqref{eq:SM-pf-bound} gives
\begin{equation}
\left|
\cS_{s,\mathcal T}(\bm z,\bm u)
\right|
\le
\cA_{\mathcal T}(\bm z).
\label{eq:SM-tree-integrand-bound}
\end{equation}

Set $J=J_*g$ and $L=4L_\beta$.  Root $\mathcal T$ at an arbitrary
vertex and sum the Majorana labels from the leaves toward the root.
Consider an edge joining a vertex to its parent, and fix the Majorana
label $a$ selected at the parent vertex.  Equation
\eqref{eq:SM-C-Schur} gives
$\sum_b |(C_\sigma)_{ab}|\le L$, while
Eq.~\eqref{eq:SM-contour-local} gives
$\sum_{R\ni b}|v_R(z)|\le J$ for each label $b$ at the other endpoint.
Since every monomial has degree at most $d$, we use the bound $d^2$
for the choices of positions at the two endpoints.  Thus the sum over
a nonroot vertex and the edge joining it to its parent is bounded by
$d^2LJ$.

At the root,
$\sum_R|v_R(z)|\le\sum_a\sum_{R\ni a}|v_R(z)|\le NJ$.
Dropping the condition that the selected positions at a vertex be
distinct only enlarges the sum.  It follows that
\begin{equation}
\cA_{\mathcal T}(\bm z)
\le
NJ(d^2LJ)^{s-1}.
\label{eq:SM-tree-normalizer}
\end{equation}

The measure of $\mathfrak C$ is $\beta$, while the integral over
$\bm u$ has measure one.  For $s\ge2$, there are $s^{s-2}$ labeled
trees on $[s]$.  Combining Eqs.~\eqref{eq:SM-cs-outer},
\eqref{eq:SM-tree-integrand-bound}, and
\eqref{eq:SM-tree-normalizer} gives
\begin{equation}
|c_s|
\le
\frac{s^{s-2}}{s!}\,
\beta^s NJ(d^2LJ)^{s-1}.
\label{eq:SM-cs-first-bound}
\end{equation}
Using $s^{s-2}/s!\le2\ee^s$ and treating $s=1$ directly, we obtain
\begin{equation}
|c_s|
\le
2\ee N\beta J\,\rho^{s-1},
\qquad
\rho=\ee\beta d^2LJ.
\label{eq:SM-cs}
\end{equation}

Since $J=J_*g$, choose $g_{\mathrm{log}}>0$ so that
\begin{equation}
\ee\beta d^2LJ_*g_{\mathrm{log}}
\le\frac14.
\label{eq:SM-glog}
\end{equation}
Then $g\le g_{\mathrm{log}}$ implies $\rho\le1/4$, and the series
\begin{equation}
H(\lambda)
:=
\sum_{s\ge1}c_s\lambda^s
\end{equation}
converges absolutely for $|\lambda|<\rho^{-1}$.  In particular, its
domain of convergence contains $\lambda=1$.

Near $\lambda=0$, the definition of the coefficients $c_s$ gives
\begin{equation}
F(\lambda)
=
Z_\sigma\exp H(\lambda).
\end{equation}
Both sides are analytic for $|\lambda|<\rho^{-1}$ and agree in a
neighborhood of zero.  The identity theorem therefore extends this
equality throughout that disk.  Since $F(\lambda)>0$ for
$0\le\lambda\le1$, the branch of $\log F$ continued from
$\lambda=0$ agrees with the real logarithm on this interval.
Evaluating at $\lambda=1$ gives
\begin{equation}
\log Z_{S,\sigma}
=
\log Z_\sigma+\sum_{s\ge1}c_s.
\label{eq:SM-log-series}
\end{equation}

\subsection{\texorpdfstring{
The truncated expansion of $\log Z_{S,\sigma}$}
{The truncated expansion of log Z}}
\label{sec:SM-truncated-series}

We approximate the series in Eq.~\eqref{eq:SM-log-series} by retaining
terms through order $S_{\max}$.  Equation~\eqref{eq:SM-cs} gives
\begin{equation}
\left|
\sum_{s>S_{\max}}c_s
\right|
\le
\frac{2\ee N\beta J}{1-\rho}\,
\rho^{S_{\max}}.
\label{eq:SM-log-tail}
\end{equation}
Since $\rho\le1/4$, choosing
\begin{equation}
S_{\max}
=
\order\left(\log\frac{N}{\eta}\right)
\label{eq:SM-Smax}
\end{equation}
makes the truncation error at most $\eta/4$.

It remains to estimate the finite sum
$\sum_{s=1}^{S_{\max}}c_s$.  Even for fixed $s$,
Eq.~\eqref{eq:SM-cs-outer} contains an integral and a large sum over
trees and label assignments.  We replace the integrals by midpoint
sums, approximate the interaction coefficients and the entries of
$C_\sigma$, and construct an unbiased estimator for the resulting
finite sum.  The next subsection bounds the approximation error.

Fix $s$, a tree $\mathcal T$, and the variables $\bm z,\bm u$ in
Eq.~\eqref{eq:SM-cs-outer}.  A local configuration $\xi_i$ at vertex
$i$ consists of a monomial label $R_i$ and a choice of a distinct
position in $R_i$ for every tree edge incident to $i$.  Let $\Omega_i$
be the set of these configurations.  If the degree of $i$ in
$\mathcal T$ is larger than $d$, then $\Omega_i$ is empty; otherwise
$|\Omega_i|\le C_dN^d$ for a constant depending only on $d$.

A tuple $\bm\xi=(\xi_1,\ldots,\xi_s)$ determines the monomial labels
$\bm R$ and the position assignment $\phi$.  Define the interaction
coefficient at vertex $i$ and the entry of $C_\sigma$ selected by an
edge $e=\{i,j\}$ by
\begin{equation}
V_i(\xi_i)=v_{R_i}(z_i),
\qquad
E_e(\xi_i,\xi_j)=C_{\phi(e)}.
\label{eq:SM-local-weights}
\end{equation}
For the same configuration, define
\begin{equation}
\Psi_{\mathcal T}(\bm\xi,\bm u)
=
\varepsilon(\mathcal T,\phi)
\operatorname{Pf}K^{\mathcal T,\phi}(\bm u).
\label{eq:SM-residual-Psi}
\end{equation}
By Eq.~\eqref{eq:SM-pf-bound},
$|\Psi_{\mathcal T}(\bm\xi,\bm u)|\le1$.  The integrand in
Eq.~\eqref{eq:SM-cs-outer} can now be written as
\begin{equation}
\cS_{s,\mathcal T}(\bm z,\bm u)
=
\sum_{\bm\xi\in\prod_i\Omega_i}
\left(\prod_{i=1}^s V_i(\xi_i)\right)
\left(\prod_{e=\{i,j\}\in\mathcal T}
E_e(\xi_i,\xi_j)\right)
\Psi_{\mathcal T}(\bm\xi,\bm u).
\label{eq:SM-ticket-target}
\end{equation}
The quantities $V_i$ and $E_e$ depend only on one vertex or one edge,
while $\Psi_{\mathcal T}$ is bounded in absolute value by one.  Direct
sampling from the absolute value of the product would require dividing
by its total weight.  This weight can be arbitrarily small, so additive
approximations of the local terms would not control the normalized
probabilities.  We therefore approximate the local terms first and
construct an unbiased estimator of the resulting unnormalized sum.

Fix a precision parameter $p$ and set $D_p=2^p$.  For each complex
number $\zeta$ used below, compute integers $A_p(\zeta)$ and
$B_p(\zeta)$ such that
\begin{equation}
\widehat\zeta
=
\frac{A_p(\zeta)+\ii B_p(\zeta)}{D_p},
\qquad
|\widehat\zeta-\zeta|\le2^{-p}.
\label{eq:SM-ticket-round}
\end{equation}
These approximations can be computed in polynomial time at precision
$2^{-p}$.  Apply the rule to every value of $V_i$ and $E_e$, and denote the
approximations by $\widehat V_i$ and $\widehat E_e$.
The numerator $A_p(\zeta)+\ii B_p(\zeta)$ is a sum of
$|A_p(\zeta)|$ terms equal to $\sgn A_p(\zeta)$ and
$|B_p(\zeta)|$ terms equal to $\ii\sgn B_p(\zeta)$.  The number of
these terms is
$n_p(\zeta)=|A_p(\zeta)|+|B_p(\zeta)|$.

After multiplication by $D_p$, every rounded local term is a finite
sum of phases.  Expanding their product gives a finite collection of
local choices whose total weight can be computed by the tree recursion
below.  For a given $\bm\xi$, choose one term from each of the $s$ rounded
coefficients and $s-1$ rounded entries, and let $\chi$ be the product
of the chosen phases.  Expanding the product gives
\begin{equation}
\sum_{\text{phase choices for }\bm\xi}\chi
=
D_p^{2s-1}
\left(\prod_{i=1}^s\widehat V_i(\xi_i)\right)
\left(\prod_{e=\{i,j\}\in\mathcal T}
\widehat E_e(\xi_i,\xi_j)\right).
\label{eq:SM-phase-product-sum}
\end{equation}
The number of phase choices associated with $\bm\xi$ is
\begin{equation}
N_{\mathcal T}(\bm\xi)
=
\left(\prod_{i=1}^s n_p(V_i(\xi_i))\right)
\left(\prod_{e=\{i,j\}\in\mathcal T}
n_p(E_e(\xi_i,\xi_j))\right).
\label{eq:SM-ticket-multiplicity}
\end{equation}
The total number of configuration and phase choices is
\begin{equation}
M_{\mathcal T}
=
\sum_{\bm\xi\in\prod_i\Omega_i}
N_{\mathcal T}(\bm\xi).
\label{eq:SM-ticket-count}
\end{equation}

By Eq.~\eqref{eq:SM-phase-product-sum}, replacing $V_i$ and $E_e$ in
Eq.~\eqref{eq:SM-ticket-target} by their rounded values produces
$D_p^{-(2s-1)}$ times a sum over the $M_{\mathcal T}$ configuration
and phase choices.  The estimator below samples one of these choices
uniformly.  It therefore needs $M_{\mathcal T}$ and a way to recover
the choice indexed by an integer in
$\{0,\ldots,M_{\mathcal T}-1\}$ without listing all tuples $\bm\xi$.

Equation~\eqref{eq:SM-ticket-multiplicity} is a product of quantities
depending on individual vertices and edges.  This product structure
makes the sum over $\bm\xi$ recursive.  Root $\mathcal T$ at an
arbitrary vertex $r$, and let
\begin{equation}
U_i(\xi_i)=n_p(V_i(\xi_i)),
\qquad
U_{ij}(\xi_i,\xi_j)
=
n_p(E_{\{i,j\}}(\xi_i,\xi_j)).
\label{eq:SM-local-counts}
\end{equation}
For a nonroot vertex $i$ with parent $j$, let
$m_{i\to j}(\xi_j)$ count the choices in the subtree rooted at $i$
when the configuration at $j$ is $\xi_j$.  For each $\xi_i$, the
choices at vertex $i$ and on the edge $\{i,j\}$ contribute
$U_i(\xi_i)U_{ij}(\xi_i,\xi_j)$.  The choices in the child subtrees
contribute the product of their messages.  Summing over $\xi_i$ gives
\begin{equation}
m_{i\to j}(\xi_j)
=
\sum_{\xi_i\in\Omega_i}
U_i(\xi_i)U_{ij}(\xi_i,\xi_j)
\prod_{k\in\operatorname{ch}(i)}
m_{k\to i}(\xi_i),
\label{eq:SM-ticket-DP}
\end{equation}
where $\operatorname{ch}(i)$ is the set of children of $i$.  At the
root there is no parent edge, so
\begin{equation}
M_{\mathcal T}
=
\sum_{\xi_r\in\Omega_r}
U_r(\xi_r)
\prod_{k\in\operatorname{ch}(r)}
m_{k\to r}(\xi_r).
\label{eq:SM-ticket-DP-root}
\end{equation}

The recursion also identifies each of the $M_{\mathcal T}$ choices.
Order the configurations in every $\Omega_i$ and the phase terms in
every rounded value.  In Eq.~\eqref{eq:SM-ticket-DP-root}, each root
configuration corresponds to a consecutive block of integers whose
size is the associated term.  The message recursion divides this
block in the same way among the choices in the child subtrees.
Proceeding from the root to the leaves maps every integer in
$\{0,\ldots,M_{\mathcal T}-1\}$ to a unique configuration tuple and one phase term for every vertex and edge.

If $M_{\mathcal T}=0$, set $Y_{\mathcal T}=0$.  Otherwise let
$k_{\mathcal T}=\lceil\log_2M_{\mathcal T}\rceil$ and draw an integer
$q$ uniformly from $\{0,\ldots,2^{k_{\mathcal T}}-1\}$.  Set
$Y_{\mathcal T}=0$ when $q\ge M_{\mathcal T}$.  When $q<M_{\mathcal T}$, follow the construction from the root to the leaves to determine a unique configuration tuple $\bm\xi$ and one phase term for each vertex and edge.  Let $\chi_q\in\{\pm1,\pm\ii\}$ be the product of these phase terms.

Let $\widehat K^{\mathcal T,\phi}(\bm u)$ be the matrix obtained from
$K^{\mathcal T,\phi}(\bm u)$ using the rounded entries, and set
\begin{equation}
\widehat\Psi_{\mathcal T}(\bm\xi,\bm u)
=
\varepsilon(\mathcal T,\phi)
\operatorname{Pf}\widehat K^{\mathcal T,\phi}(\bm u).
\label{eq:SM-rounded-Psi}
\end{equation}
For $q<M_{\mathcal T}$, define
\begin{equation}
Y_{\mathcal T}
=
2^{k_{\mathcal T}}D_p^{-(2s-1)}
\chi_q\widehat\Psi_{\mathcal T}(\bm\xi,\bm u).
\label{eq:SM-ticket-output}
\end{equation}
The next subsection specifies the precision and bounds the cost of
evaluating this Pfaffian.

Every valid configuration and phase choice is selected with
probability $2^{-k_{\mathcal T}}$, which cancels the factor
$2^{k_{\mathcal T}}$ in Eq.~\eqref{eq:SM-ticket-output}.  Therefore
\begin{equation}
\bbE\left(
Y_{\mathcal T}\,|\, s,\mathcal T,\bm z,\bm u
\right)
=
\sum_{\bm\xi\in\prod_i\Omega_i}
\left(\prod_{i=1}^s\widehat V_i(\xi_i)\right)
\left(\prod_{e=\{i,j\}\in\mathcal T}
\widehat E_e(\xi_i,\xi_j)\right)
\widehat\Psi_{\mathcal T}(\bm\xi,\bm u).
\label{eq:SM-ticket-unbiased}
\end{equation}
Thus $Y_{\mathcal T}$ is unbiased for the rounded integrand, without
dividing by the small normalizing sum considered above.

We now sample the remaining variables in
Eq.~\eqref{eq:SM-cs-outer}.  For $1\le s\le S_{\max}$, use probability
$\pi_s=2^{-s}$, and assign the unused probability
$2^{-S_{\max}}$ to the value zero.  For $s=1$, use the tree with one
vertex and set $\tau_1=\widehat\tau_1=1$.
For $s\ge2$, labeled trees on $[s]$ are in bijection with Pr\"ufer
words of length $s-2$, so their number is $\tau_s=s^{s-2}$.  Set
$\ell_s=\lceil\log_2s\rceil$ and
$\widehat\tau_s=2^{\ell_s(s-2)}$.  Draw $\ell_s$ bits for each symbol
of the word.  Values from $0$ to $s-1$ encode the labels in $[s]$; if
any value is at least $s$, set the estimator to zero.  Every valid
tree is then selected with probability $\widehat\tau_s^{-1}$.

The construction above gives an estimator for
$\cS_{s,\mathcal T}(\bm z,\bm u)$ once
$s,\mathcal T,\bm z$, and $\bm u$ have been specified.  To obtain
$c_s$, we must also average over $\bm z$ and $\bm u$.  Since the
algorithm uses finitely many bits, we replace these integrals by
midpoint sums on the grid
\begin{equation}
\mathcal G_p
=
\left\{
\left(m+\frac12\right)2^{-p}:
m=0,\ldots,2^p-1
\right\}.
\label{eq:SM-midpoint-grid}
\end{equation}
For each $z_j$, draw the branch label uniformly from
$\{\mathrm L,\mathrm R\}$ and draw $r_j$ uniformly from
$\mathcal G_p$, giving $z_j=(b_j,Tr_j)$.  Draw every interpolation
variable $u_e$ independently from the same grid.

These averages give the midpoint approximation to the normalized
integrals in Eq.~\eqref{eq:SM-cs-outer}.  The interpolation domain has
volume one, while the time domain has volume $\beta^s$.  Choose a
positive approximation $b_{s,p}$ satisfying
\begin{equation}
|b_{s,p}-\beta^s|\le2^{-p},
\qquad
b_{s,p}\le2\beta^s.
\label{eq:SM-contour-prefactor-round}
\end{equation}
The next subsection chooses $p$ so that these conditions hold for all
$s\le S_{\max}$.

These choices define one sample of the complete estimator.  If the
order sampler returns zero, the Pr\"ufer word is invalid, or
$M_{\mathcal T}=0$, set $Y_{\mathrm{out}}=0$.  Otherwise set
\begin{equation}
Y_{\mathrm{out}}
=
\frac{(-1)^s\widehat\tau_s b_{s,p}}
{\pi_s s!}
Y_{\mathcal T}.
\label{eq:SM-complete-estimator}
\end{equation}
Let $\widehat c_{s,p}$ denote the midpoint approximation to $c_s$ in
which the interaction coefficients, covariance entries, and residual
Pfaffian are replaced by their approximations above.  The sampling
probabilities cancel the prefactors in
Eq.~\eqref{eq:SM-complete-estimator}, and hence
\begin{equation}
\bbE Y_{\mathrm{out}}
=
\sum_{s=1}^{S_{\max}}\widehat c_{s,p}.
\label{eq:SM-complete-unbiased}
\end{equation}

\begin{lemma}
\label{lem:SM-complete-envelope}
Suppose $\widehat J$ bounds the rooted absolute sums of the rounded
interaction coefficients and
$\norm{\widehat C_\sigma}_{\mathrm{Schur}}\le\widehat L$.  Then
\begin{equation}
D_p^{-(2s-1)}M_{\mathcal T}
\le
2^sN\widehat J
(d^2\widehat L\widehat J)^{s-1}.
\label{eq:SM-ticket-envelope}
\end{equation}
Set
$\widehat\rho=2\ee\beta d^2\widehat L\widehat J$.
If $\widehat\rho\le1/4$ and
$|\widehat\Psi_{\mathcal T}(\bm\xi,\bm u)|\le2$ for every
configuration under consideration, then
\begin{equation}
|Y_{\mathrm{out}}|
\le
16\ee N\beta\widehat J.
\label{eq:SM-complete-O-N}
\end{equation}
In particular, $|Y_{\mathrm{out}}|\le C_{\beta,D,r_0}N$ when
$\widehat J$ is bounded in terms of $\beta,D,r_0$.
\end{lemma}

\begin{proof}
For every rounded value $\widehat\zeta$,
\begin{equation}
\frac{n_p(\zeta)}{D_p}
=
|\Re\widehat\zeta|+|\Im\widehat\zeta|
\le
\sqrt2|\widehat\zeta|.
\label{eq:SM-phase-count-bound}
\end{equation}
The product defining $N_{\mathcal T}(\bm\xi)$ contains $s$ rounded
interaction coefficients and $s-1$ rounded entries of $C_\sigma$.
Applying Eq.~\eqref{eq:SM-phase-count-bound} to these $2s-1$
quantities contributes at most
$(\sqrt2)^{2s-1}\le2^s$.

The remaining sum has the same vertex and edge structure as the
absolute sum in Eq.~\eqref{eq:SM-tree-normalizer}.  Applying the same
summation from the leaves, with $J,L$ replaced by
$\widehat J,\widehat L$, gives
\begin{equation}
D_p^{-(2s-1)}M_{\mathcal T}
\le
2^sN\widehat J
(d^2\widehat L\widehat J)^{s-1}.
\end{equation}
If $M_{\mathcal T}=0$, then $Y_{\mathcal T}=0$.  Otherwise,
$2^{k_{\mathcal T}}<2M_{\mathcal T}$.  Equation
\eqref{eq:SM-ticket-output} and the bound
$|\widehat\Psi_{\mathcal T}|\le2$ therefore give
\begin{equation}
|Y_{\mathcal T}|
\le
4D_p^{-(2s-1)}M_{\mathcal T}.
\label{eq:SM-fixed-tree-envelope}
\end{equation}
For $s\ge2$, the padded tree count satisfies
$\widehat\tau_s\le2^{s-2}\tau_s$, while
$\tau_s/s!\le\ee^s$.  Combining these inequalities with
Eqs.~\eqref{eq:SM-contour-prefactor-round},
\eqref{eq:SM-ticket-envelope}, and
\eqref{eq:SM-fixed-tree-envelope} yields
\begin{equation}
|Y_{\mathrm{out}}|
\le
8\ee N\beta\widehat J
\frac{(2\widehat\rho)^{s-1}}{\pi_s}.
\label{eq:SM-output-envelope}
\end{equation}
The same inequality holds for $s=1$ by direct calculation.  Since
$\pi_s=2^{-s}$ and $\widehat\rho\le1/4$,
\[
\frac{(2\widehat\rho)^{s-1}}{\pi_s}
=
2(4\widehat\rho)^{s-1}
\le2.
\]
Substituting this bound into
Eq.~\eqref{eq:SM-output-envelope} proves
Eq.~\eqref{eq:SM-complete-O-N}.
\end{proof}

\subsection{Precision and running time}
\label{sec:SM-oracle-precision}

The estimator in Eq.~\eqref{eq:SM-complete-unbiased} has expectation
$\sum_{s=1}^{S_{\max}}\widehat c_{s,p}$, whereas the target is
$\sum_{s=1}^{S_{\max}}c_s$.  We choose $p$ so that the rounded data
satisfy Lemma~\ref{lem:SM-complete-envelope} and the difference
between these two sums is at most $\eta/4$.  We then choose the number
of samples needed to control the remaining random error.

Recall that $J=J_*g$ and $L=4L_\beta$.   Decrease
$g_{\mathrm{log}}$ if necessary so that
\begin{equation}
2\ee\beta d^2LJ\le\frac1{16}
\label{eq:SM-exact-rounded-rho}
\end{equation}
whenever $g\le g_{\mathrm{log}}$.  There are at most $N^d$
interaction coefficients relevant to any evaluation.  If every
coefficient and matrix entry is computed with absolute error at most
$2^{-p}$, then
\begin{equation}
\widehat J\le J+\nu_p,
\qquad
\widehat L\le L+\nu_p,
\qquad
\nu_p=C_*N^d2^{-p},
\label{eq:SM-rounded-envelope-error}
\end{equation}
where $C_*=C_*(\beta,D,r_0,d)$.  Choose $\nu_0>0$ so that
\begin{equation}
2\ee\beta d^2(L+\nu_0)(J+\nu_0)\le\frac14
\label{eq:SM-rounded-rho-margin}
\end{equation}
for every permitted value of $J$.  The condition $\nu_p\le\nu_0$
then implies $\widehat\rho\le1/4$.

We must also keep the rounded Pfaffian bounded.  For a matrix $A$,
write $\norm A_{\max}=\max_{a,b}|A_{ab}|$.  If $A$ and $B$ are
antisymmetric $r\times r$ matrices and $r\ge2$ is even, expansion over
perfect matchings gives
\begin{equation}
|\operatorname{Pf}A-\operatorname{Pf}B|
\le
\frac r2(r-1)!!
\max\{\norm A_{\max},\norm B_{\max}\}^{r/2-1}
\norm{A-B}_{\max}.
\label{eq:SM-pf-perturb}
\end{equation}
Indeed, the two Pfaffians contain $(r-1)!!$ products of $r/2$
entries, and the difference of two such products is bounded by
telescoping one entry at a time.
The matrix $K^{\mathcal T,\phi}(\bm u)$ has size at most
$dS_{\max}$.  Since its Pfaffian has absolute value at most one by
Eq.~\eqref{eq:SM-pf-bound}, Eq.~\eqref{eq:SM-pf-perturb} shows that
$|\widehat\Psi_{\mathcal T}|\le2$ whenever
\begin{equation}
C_*\exp\left(
C_*dS_{\max}\log(dS_{\max}+1)
\right)\nu_p
\le1.
\label{eq:SM-rounded-pf-condition}
\end{equation}
The empty Pfaffian is equal to one and requires no approximation.

Choose $p_0=p_0(S_{\max},N)$ so that
\begin{equation}
\nu_{p_0}\le\nu_0,
\qquad
C_*\exp\left(
C_*dS_{\max}\log(dS_{\max}+1)
\right)\nu_{p_0}\le1,
\qquad
2^{-p_0}\le\min\{1,\beta^{S_{\max}}\}.
\label{eq:SM-precision-floor}
\end{equation}
The first two conditions give the assumptions of
Lemma~\ref{lem:SM-complete-envelope}.  The third guarantees that the
finite-precision approximation $b_{s,p}$ in
Eq.~\eqref{eq:SM-contour-prefactor-round} can be chosen for every
$s\le S_{\max}$.  These conditions require only
\begin{equation}
p_0(S_{\max},N)
=
\order_{\beta,D,r_0,d}\left(
d\log(N+1)+S_{\max}\log(S_{\max}+1)+1
\right).
\label{eq:SM-floor-size}
\end{equation}
The preceding conditions keep each sample bounded.  We next choose
$p$ large enough to control the bias caused by rounding and midpoint
summation.

\begin{lemma}
\label{lem:SM-rounding-stability}
Suppose the coefficients and matrix entries used in the estimator are
computed with absolute error at most $2^{-p}$, and the integrals are
replaced by the midpoint sums in
Eq.~\eqref{eq:SM-midpoint-grid}.  If $p\ge p_0(S_{\max},N)$, then
\begin{equation}
\left|
\sum_{s=1}^{S_{\max}}c_s
-
\sum_{s=1}^{S_{\max}}\widehat c_{s,p}
\right|
\le
CN^{d+1}
\exp\left(
CS_{\max}\log(S_{\max}+1)
\right)2^{-p},
\label{eq:SM-total-rounding-bound}
\end{equation}
where $C=C(\beta,D,r_0,d)$ is independent of the active set and the
matching.  Consequently, the bias is at most $\eta/4$ if
\begin{equation}
p\ge
\max\left\{
p_0(S_{\max},N),
C_{\beta,D,r_0,d}\left[
S_{\max}\log(S_{\max}+1)
+\log\frac{N+1}{\eta}+1
\right]
\right\}.
\label{eq:SM-sufficient-p}
\end{equation}
Since $S_{\max}=\order(\log(N/\eta))$, it is enough to take
\begin{equation}
p
=
\order_{\beta,D,r_0,d}\left(
\log^2\frac{N}{\eta}+1
\right).
\label{eq:SM-ticket-precision}
\end{equation}
\end{lemma}

\begin{proof}
We first compare the exact and rounded expressions at the grid
points.  Telescope the product of the $s$ interaction coefficients
and the $s-1$ selected entries of $C_\sigma$.  Each replacement changes
a rooted absolute sum by at most $\nu_p$.  The other coefficients and
matrix entries can be summed from the leaves as in
Eq.~\eqref{eq:SM-tree-normalizer}.  Equation
\eqref{eq:SM-pf-perturb} controls the change in the Pfaffian.  Since
its matrix has size at most $dS_{\max}$, these errors are bounded by
the right-hand side of Eq.~\eqref{eq:SM-total-rounding-bound}.

It remains to compare the integrals with their midpoint sums.  The
coefficients $v_R(z)$ satisfy
\begin{equation}
\sup_z\max_a\sum_{R\ni a}
|\partial_r v_R(z)|
\le
J_*'g,
\label{eq:SM-contour-derivative}
\end{equation}
where $r$ is the normalized time coordinate and
$J_*'=J_*'(\beta,D,r_0,d)$.  This follows by differentiating the
quadratic propagators defining $v_R(z)$ and applying the same
exterior-power bound used for Eq.~\eqref{eq:SM-contour-local}.
The interpolation entries $w_{ij}^{\mathcal T}(\bm u)$ are Lipschitz
functions of $\bm u$, because each is the minimum of the variables on
a path in $\mathcal T$.
On every region with a fixed ordering of the time variables, these
bounds control the midpoint error.  The grid cells intersecting an
equal-time hyperplane on one branch have total normalized volume
$\order(S_{\max}^22^{-p})$; their contribution is bounded using
Eq.~\eqref{eq:SM-tree-normalizer}.  Including the approximation of
$\beta^s$ by $b_{s,p}$ and summing the labeled trees with
$\tau_s/s!\le\ee^s$ proves
Eq.~\eqref{eq:SM-total-rounding-bound}.
\end{proof}

We now bound the cost of one sample.  For a sampled tree, every set
$\Omega_i$ has size at most $C_dN^d$.  The message recursion, the
selection of a configuration and its phases, the rounded Pfaffian, and
the sampling of the tree and grid points require
$N^{\order(d)}\poly(S_{\max},p)$ arithmetic operations at precision
$2^{-p}$, in addition to the cost of evaluating the input data at that
precision.

\begin{proposition}
\label{prop:SM-completed-log-oracle}
Choose $S_{\max}$ as in Eq.~\eqref{eq:SM-Smax} and choose $p$ according
to Eq.~\eqref{eq:SM-sufficient-p}.  Let
\begin{equation}
R_{\mathrm{MC}}
=
\order\left(
N^2\eta^{-2}\log\frac2\delta
\right),
\label{eq:SM-ticket-replicates}
\end{equation}
and let $\overline Y$ be the average of $R_{\mathrm{MC}}$ independent
samples of $Y_{\mathrm{out}}$.  Compute $\widehat L_0$ using
Lemma~\ref{lem:SM-baseline-logtrace} so that
$|\widehat L_0-\log Z_\sigma|\le\eta/4$, and return
\begin{equation}
\widehat L
=
\widehat L_0+\Re\overline Y.
\label{eq:SM-final-log-estimator}
\end{equation}
Then $\widehat L$ satisfies
Eqs.~\eqref{eq:SM-oracle-guarantee} and
\eqref{eq:SM-oracle-runtime}.
\end{proposition}

\begin{proof}
The truncation error in Eq.~\eqref{eq:SM-log-tail}, the bias in
Lemma~\ref{lem:SM-rounding-stability}, and the error in
$\widehat L_0$ are each at most $\eta/4$.  The precision choice also
gives the assumptions of Lemma~\ref{lem:SM-complete-envelope}, so
$|Y_{\mathrm{out}}|\le C_{\beta,D,r_0,d}N$.  Hoeffding's inequality
therefore gives
\begin{equation}
\left|
\Re\overline Y
-
\bbE\Re Y_{\mathrm{out}}
\right|
\le\frac\eta4
\label{eq:SM-MC-error}
\end{equation}
except with probability at most $\delta$.  Adding the four errors
proves Eq.~\eqref{eq:SM-oracle-guarantee}.

A single sample requires
$N^{\order(d)}\poly(S_{\max},p)$ arithmetic operations at precision
$2^{-p}$, apart from the cost of evaluating the input data.  Substituting
$S_{\max}=\order(\log(N/\eta))$ and
$p=\order_{\beta,D,r_0,d}(\log^2(N/\eta)+1)$, and multiplying by
$R_{\mathrm{MC}}$, gives Eq.~\eqref{eq:SM-oracle-runtime}.
\end{proof}

The theorem was stated for $\sigma$ without its positive scalar
coefficient.  At a node $v$, write
\begin{equation}
\sigma_v=c_v\overline\sigma_v,
\qquad
\overline\sigma_v
=
\prod_{e\in\mathcal M_v}
(\one+t s_eB_e).
\end{equation}
Then
\begin{equation}
\log Z_{S_v,\sigma_v}
=
\log c_v+\log Z_{S_v,\overline\sigma_v}.
\end{equation}
We therefore apply the estimator to $\overline\sigma_v$ and add
$\log c_v$.

If the resulting estimate $\widehat L_v$ satisfies
$|\widehat L_v-\log Z_{S_v,\sigma_v}|\le\eta$, then
$\widehat Z_v=\ee^{\widehat L_v}$ satisfies
\begin{equation}
\ee^{-\eta}\widehat Z_v
\le
Z_{S_v,\sigma_v}
\le
\ee^\eta\widehat Z_v.
\end{equation}
Combining this with Eq.~\eqref{eq:SM-live-sandwich} gives
\begin{equation}
\frac12\ee^{-\eta}\widehat Z_v
\le
W_v
\le
\frac32\ee^\eta\widehat Z_v.
\end{equation}
Thus a constant choice of $\eta$ gives the constant-factor estimates
of the node weights used in Appendix~\ref{sec:SM-finite}.

\section{Efficient classical sampling}
\label{sec:SM-finite}

Appendix~\ref{sec:SM-soft-tree} constructs an exact sampling tree, and
Appendix~\ref{sec:SM-oracle} gives estimates of its node weights.  Two
further steps are needed for the finite algorithm in
Theorem~\ref{thm:SM-sampler}.  First, an index removal step uses
continuous time variables and expansion paths of unbounded length, so
its child distribution must be replaced by a finite one.  Second, the
native leaf probabilities do not include the traces of the Gaussian
leaf operators.  We first approximate the sampling tree and control
the resulting operator error.  We then sample its leaves with the
required trace weights and convert the selected leaf into a normalized
Gaussian state.

\subsection{Truncation and discretization}
\label{sec:SM-finite-tree}

We will replace the child distribution at each index removal node by a
finite distribution.  To control how these changes accumulate
along the tree, we first compare the weights of adjacent nodes.

\begin{lemma}
\label{lem:SM-edge-bound}
For every edge $v\to w$ in the sampling tree,
\begin{equation}
\frac13
\le
\frac{W_w}{W_v}
\le
3.
\label{eq:SM-edge}
\end{equation}
\end{lemma}

\begin{proof}
We consider index removal and soft pinning separately.  For any
$S$ and $\sigma$, define the positive linear functional
\begin{equation}
\Phi_{S,\sigma}(Y)
=
\Tr\left[
GX^S\sigma^{1/2}Y\sigma^{1/2}(X^S)^\dagger G
\right].
\label{eq:SM-edge-functional}
\end{equation}

Suppose first that $v\to w$ is an index removal edge.  Put
$S'=S_v\setminus\{i\}$, $M=M_{S_v,i}$, and
\begin{equation}
P_v=\one+\alpha_v\Gamma_v,
\qquad
P_w=\one+\alpha_w\Gamma_w.
\end{equation}
The operator $\sigma_v$ is unchanged by index removal and commutes
with $M$, $P_v$, and $P_w$.  Hence, with
$\Phi=\Phi_{S',\sigma_v}$,
\begin{equation}
W_v=\Phi(MP_vM^\dagger),
\qquad
W_w=\Phi(P_w).
\label{eq:SM-index-edge-functional}
\end{equation}

At an index removal node, either $\alpha_v=0$ or
$\supp\Gamma_v\cap S_v\ne\varnothing$.  The invariant in
Eq.~\eqref{eq:SM-soft-invariant} therefore gives
\begin{equation}
\alpha_v\le\frac{1}{2Q}\le\frac1{16},
\qquad
\alpha_w\le\frac12.
\label{eq:SM-index-edge-alpha}
\end{equation}
Lemma~\ref{lem:SM-ratio}, applied with $q=Q$, and
Eq.~\eqref{eq:SM-gpin} give
\begin{equation}
\norm{M-\one}_{\mathrm{op}}
\le
L_Q(M-\one)
\le
\kappa_0
<
\frac1{32}.
\label{eq:SM-index-edge-M}
\end{equation}
It follows that
\begin{equation}
(1-\alpha_v)(1-\kappa_0)^2\one
\preceq
MP_vM^\dagger
\preceq
(1+\alpha_v)(1+\kappa_0)^2\one
\label{eq:SM-parent-order}
\end{equation}
and
\begin{equation}
(1-\alpha_w)\one
\preceq
P_w
\preceq
(1+\alpha_w)\one.
\label{eq:SM-child-order}
\end{equation}
Applying $\Phi$ and using Eq.~\eqref{eq:SM-index-edge-alpha} gives
\begin{equation}
\frac{1-\alpha_w}
{(1+\alpha_v)(1+\kappa_0)^2}
\le
\frac{W_w}{W_v}
\le
\frac{1+\alpha_w}
{(1-\alpha_v)(1-\kappa_0)^2}.
\label{eq:SM-index-edge-ratio}
\end{equation}
The lower bound is larger than $1/3$, and the upper bound is smaller
than $3$.

Now suppose that $v\to w$ is a soft pinning edge.  Write
$\Gamma_v=AB$, where $A$ and $B$ are commuting Hermitian Majorana
monomials with disjoint supports and $B$ is quadratic.  The active set
does not change.  With $\Phi=\Phi_{S_v,\sigma_v}$, the parent weight is
\begin{equation}
W_v=\Phi(\one+\alpha_vAB),
\end{equation}
whereas the child corresponding to $s\in\{\pm1\}$ has weight
\begin{equation}
W_w
=
\Phi\left[
\left(\one+s\frac{\alpha_v}{t}A\right)
(\one+stB)
\right].
\end{equation}
Since this step is applied to a nonscalar monomial, the invariant gives
$\alpha_v\le t/2$ and $\alpha_v/t\le1/2$.  Since $A$ and $B$ commute
and square to the identity,
\begin{equation}
(1-\alpha_v)\one
\preceq
\one+\alpha_vAB
\preceq
(1+\alpha_v)\one
\end{equation}
and
\begin{equation}
(1-\alpha_v/t)(1-t)\one
\preceq
\left(\one+s\frac{\alpha_v}{t}A\right)(\one+stB)
\preceq
(1+\alpha_v/t)(1+t)\one.
\end{equation}
Applying $\Phi$ gives
\begin{equation}
\frac{(1-\alpha_v/t)(1-t)}{1+\alpha_v}
\le
\frac{W_w}{W_v}
\le
\frac{(1+\alpha_v/t)(1+t)}{1-\alpha_v}.
\label{eq:SM-soft-edge-ratio}
\end{equation}
For $t\le1/16$, the two bounds lie between $5/11$ and $51/31$,
which proves Eq.~\eqref{eq:SM-edge}.
\end{proof}

Soft pinning already has two outcomes, each with probability $1/2$,
and uses no continuous variables.  Only index removal needs to be
approximated.  One index removal draws two independent expansion paths
as in Lemma~\ref{lem:SM-ratio-generation}.  If $\ell$ is the sum of
their lengths, Eq.~\eqref{eq:SM-path-length-tail} and a union bound give
constants $c_0,C_0>0$ such that
\begin{equation}
\bbP(\ell>L)
\le
C_0\ee^{-c_0L}.
\label{eq:SM-step-length-tail}
\end{equation}

Every path from the root to a leaf contains exactly $N=2n$ index
removals.  Between two consecutive removals, soft pinning can occur at
most $n$ times.  Hence
\begin{equation}
h=2n(n+1)
\label{eq:SM-tree-depth}
\end{equation}
is an upper bound on the depth.  For $0<\epsilon\le1$, set
\begin{equation}
\delta
=
\min\left\{
\frac1{48},
\frac{\epsilon}{1000h}
\right\}.
\label{eq:SM-kernel-accuracy}
\end{equation}

At an index removal node $v$, let $P_v$ be the distribution of the
resulting child, whose data are
\begin{equation}
(S_w,\sigma_w,\alpha_w,\Gamma_w,c_w).
\end{equation}
The expansion
paths and continuous times are used only to generate this child and
are then discarded.  We approximate the distribution of the child,
not the distribution of these auxiliary variables.

\begin{lemma}
\label{lem:SM-capped-kernel}
For every index removal node $v$, there is a distribution $P'_v$
supported on finitely many legal children such that
\begin{equation}
\TV(P'_v,P_v)\le\delta.
\label{eq:SM-kernel-TV}
\end{equation}
A child with distribution $P'_v$ can be sampled in
$N^{\order(d)}\poly(\log(1/\delta))$ time under the input assumptions
of Theorem~\ref{thm:SM-sampler}.  If $k$
is the number of children with positive probability under $P'_v$,
then
\begin{equation}
\log k
=
\poly(n,\log(1/\epsilon)).
\label{eq:SM-logk}
\end{equation}
At a soft pinning node, set $P'_v=P_v$.
\end{lemma}

\begin{proof}
We first truncate the two expansion paths used in index removal.
Choose
\begin{equation}
L
\ge
c_0^{-1}\log\frac{4C_0}{\delta}.
\label{eq:SM-path-cutoff}
\end{equation}
If $\ell>L$, use the child obtained by selecting the term
$\alpha_v\Gamma_v$ in Eq.~\eqref{eq:SM-seven}.  When $\alpha_v=0$,
this is the child with no remaining Majorana monomial.  The resulting child is
legal by Lemma~\ref{lem:SM-soft-invariant-preserved}.  Equation~\eqref{eq:SM-step-length-tail} gives
$\bbP(\ell>L)\le\delta/4$, so this replacement changes the child law
by at most $\delta/4$ in total variation.  Applying the cutoff during
the generation of the two paths ensures that no more than $L$
interaction terms are generated.

We next replace the continuous time variables by finitely many values.
For fixed discrete choices $\omega$ and a child $c$, let
$g_{\omega,c}(\bm t)$ be the density of returning $c$ at times
$\bm t$.  For the corresponding path $h$, the four choices
$\zeta=1,-1,\ii,-\ii$ have weights
$(\Re a(h))_+$, $(-\Re a(h))_+$,
$(\Im a(h))_+$, and $(-\Im a(h))_+$, respectively.
If $a(h)$ changes to $a'(h)$, the sum of the absolute changes in these
weights is at most $\sqrt{2}|a(h)-a'(h)|$.

In the proposal probabilities of
Lemma~\ref{lem:SM-ratio-generation}, evaluate the upper bounds at
$g_{\mathrm{pin}}$ and retain the actual interaction coefficients in
the probabilities of the returned terms.  The unused probability is
assigned to zero.  The inequality
$g\le g_{\mathrm{samp}}\le g_{\mathrm{pin}}$ ensures that the resulting
law is unchanged.

Changing one time variable differentiates one propagator entry.  Since
\begin{equation}
\norm{Q'(u)}_{1\to1},
\norm{Q'(u)}_{\infty\to\infty}
\le C_H\ee^{C_Hu},
\end{equation}
the estimates used in
Lemma~\ref{lem:SM-ratio-generation} give
\begin{equation}
\sum_{\omega,c}
\int
\left(
g_{\omega,c}(\bm t)
+
\norm{\nabla g_{\omega,c}(\bm t)}_1
\right)\dd\bm t
\le
\poly(L)C_1^{L+1},
\label{eq:SM-path-density-bound}
\end{equation}
where $C_1=C_1(\beta,D,r_0)$.

Divide each time interval into $2^B$ equal parts and replace every
sampled time by the midpoint of its interval.  If two sampled times lie
in the same subinterval, we use the child selected above.  Otherwise,
the midpoint replacement preserves their order.  The gradient term in
Eq.~\eqref{eq:SM-path-density-bound} then bounds the change in the
output probabilities.  For each pair of time variables, the region in
which both times lie in the same subinterval occupies a fraction
$2^{-B}$ of $[0,T]^2$.  There are at most $\binom{L}{2}$ pairs of time variables, so
Eq.~\eqref{eq:SM-path-density-bound} bounds the probability that some
pair lies in the same subinterval by
$\poly(L)C_1^{L+1}2^{-B}$.

Let $P_v^{(B)}$ be the output distribution after truncating the paths
and replacing the times by their midpoints.  Then
\begin{equation}
\TV(P_v^{(B)},P_v)
\le
\frac{\delta}{4}
+
\poly(L)C_1^{L+1}2^{-B}.
\label{eq:SM-finite-child-error}
\end{equation}

Choose $B$ so that the second term is at most $\delta/2$.  Since
$L=\order(\log(1/\delta))$, it is enough to take
\begin{equation}
B
=
\order\left(
L\log(L+1)+\log(1/\delta)
\right)
\end{equation}
The midpoint distribution can be sampled from its cumulative
probabilities without listing its $2^B$ grid points, using only
polynomially many random choices.

Small interaction coefficients create the remaining numerical
difficulty: their phases are unstable under an absolute approximation.
Let $D_N\le N^{\order(d)}$ bound the number of labeled interaction
terms in one local draw,
and set
\begin{equation}
\xi=\frac{\delta}{64LD_N}.
\end{equation}
At a grid time $u$, write
\begin{equation}
p_a
=
\frac{Q^{|T|}|z_{T,\varphi}(u)|}
{Q^dC_Vg_{\mathrm{pin}}\ee^{dC_Hu}}
\end{equation}
for the proposal probability of a labeled term.  Choose
$\underline p_a$ so that
$0\le p_a-\underline p_a\le\xi$, and assign the term to the zero
outcome when $\underline p_a<\xi$.  The discarded mass is at most
$2D_N\xi$ in one draw and at most $2LD_N\xi\le\delta/32$ along the two
paths.

Every retained term has $p_a\ge\xi$, and hence
$|z_{T,\varphi}(u)|\ge C_Vg_{\mathrm{pin}}\xi$.  A nonzero truncated
path coefficient is therefore bounded below by
$[\min\{1,C_Vg_{\mathrm{pin}}\xi\}]^L$.  The acceptance probability in
Eq.~\eqref{eq:SM-proposal-acceptance} therefore has a denominator
bounded below by the same quantity.  For the phase, choose
$\zeta\in\{1,-1,\ii,-\ii\}$ with probability
\begin{equation}
\frac{a_\zeta(h)}{
\sum_{\zeta'\in\{1,-1,\ii,-\ii\}}a_{\zeta'}(h)},
\end{equation}
where the nonnegative weights $a_\zeta(h)$ are defined in
Eq.~\eqref{eq:SM-phase-components}.  This denominator is at least
$|a(h)|$ and obeys the same lower bound.  Computing these probabilities requires
$\order(L\log(NL/\delta))$ digits of precision.  Choose this precision
so that these probabilities contribute at most $\delta/32$ coupling
error and all other finite choices contribute at most $\delta/8$.
Let $P'_v$ be the resulting child distribution.  Coupling the
successive choices gives
\begin{equation}
\TV(P'_v,P_v^{(B)})
\le
\frac{\delta}{32}+\frac{\delta}{32}+\frac{\delta}{8}
<\frac{\delta}{4}.
\end{equation}
Together with Eq.~\eqref{eq:SM-finite-child-error}, this proves
Eq.~\eqref{eq:SM-kernel-TV}.

Each path contains at most $L$ interaction terms, and each time is
specified by a $B$-bit grid index.  The number of choices and the
precision required for each one are polynomial in $N$ and
$\log(1/\delta)$.  This gives the stated running time, and the same
bounds prove Eq.~\eqref{eq:SM-logk}.
\end{proof}
Replacing every index removal law by $P'_v$ gives a finite tree of
depth at most $h$.  The leaf operators are unchanged; only their
probabilities are modified.  Define $E'_z=E_z$ at a leaf and, at an
internal node,
\begin{equation}
E'_v
=
\sum_wP'_v(w)E'_w.
\label{eq:SM-finite-completion}
\end{equation}
The next lemma compares this operator with the exact conditional
average $E_v$.

\begin{lemma}
\label{lem:SM-capped-perturbation}
Let $h_v$ be the maximum number of remaining edges below $v$.  Then
\begin{equation}
\frac{\norm{E'_v-E_v}_1}{W_v}
\le
(1+6\delta)^{h_v}-1
\le
12h_v\delta.
\label{eq:SM-tree-perturb}
\end{equation}
\end{lemma}

\begin{proof}
The claim is immediate at a leaf.  Put
$a_r=(1+6\delta)^r-1$ and suppose that
$\norm{E'_w-E_w}_1\le a_{h_v-1}W_w$ for every child $w$ of $v$.
Since $\TV(P'_v,P_v)\le\delta$ and
$W_w\le3W_v$ by Lemma~\ref{lem:SM-edge-bound},
\begin{equation}
\sum_w|P'_v(w)-P_v(w)|W_w
\le
6\delta W_v.
\label{eq:SM-one-step-TV}
\end{equation}
Since the exact child weights satisfy
$W_v=\sum_wP_v(w)W_w$, Eq.~\eqref{eq:SM-one-step-TV} also gives
\begin{equation}
\sum_wP'_v(w)W_w
\le
(1+6\delta)W_v.
\end{equation}
Using the definitions of $E'_v$ and $E_v$, together with the induction
hypothesis, we obtain
\begin{align}
\norm{E'_v-E_v}_1
&\le
\sum_wP'_v(w)\norm{E'_w-E_w}_1
+
\left\|
\sum_w\bigl(P'_v(w)-P_v(w)\bigr)E_w
\right\|_1
\nonumber\\
&\le
\left[(1+6\delta)a_{h_v-1}+6\delta\right]W_v
=
a_{h_v}W_v.
\end{align}
This proves the first inequality in
Eq.~\eqref{eq:SM-tree-perturb}.  Finally,
$6h_v\delta\le6h\delta<1/2$, so
$(1+x)^r-1\le2rx$ gives
$a_{h_v}\le12h_v\delta$.
\end{proof}

At the root, write $E'=E'_{\mathrm{root}}$ and
$Z=\Tr\ee^{-\beta H}$.  Equation~\eqref{eq:SM-tree-perturb} gives
\begin{equation}
\norm{E'-\ee^{-\beta H}}_1
\le
aZ,
\qquad
a=12h\delta
\le
0.012\epsilon.
\label{eq:SM-root-perturb}
\end{equation}

\subsection{Subtree trace weights}
\label{sec:SM-subtree-weights}

Under the finite child distributions constructed above, let $p'_v$
be the probability of reaching a node $v$.  At a leaf $z$, define
\begin{equation}
L_z=G\sigma_zG,
\qquad
F_z=\Tr L_z.
\label{eq:SM-finite-leaf}
\end{equation}
The desired reweighting assigns leaf $z$ a weight proportional to
$p'_zF_z$.  To sample this distribution without enumerating the
leaves, we first consider the total leaf weight below each node.

We begin by approximating the leaf traces.  Set
$\xi=\epsilon/1000$.  Applying
Lemma~\ref{lem:SM-baseline-logtrace} to $\sigma_z$, including its
stored scalar, gives a positive number $\widetilde F_z$ such
that
\begin{equation}
\ee^{-\xi}F_z
\le
\widetilde F_z
\le
\ee^\xi F_z.
\label{eq:SM-Ftilde}
\end{equation}
Lemma~\ref{lem:SM-baseline-logtrace} computes
$\widetilde F_z$ in polynomial time at the required precision.
Whenever a leaf is revisited, we use the same value.

For a node $v$ with $p'_v>0$, define
\begin{equation}
r_v
=
\sum_{z\succeq v}
p'(z\,|\,v)\widetilde F_z,
\label{eq:SM-subtree-trace}
\end{equation}
where $p'(z\,|\,v)$ is the probability of reaching $z$ starting from
$v$.  Thus $r_z=\widetilde F_z$ at a leaf, while at an internal node
\begin{equation}
r_v
=
\sum_wP'_v(w)r_w.
\label{eq:SM-subtree-trace-recursion}
\end{equation}
At the root,
$r_{\mathrm{root}}=\sum_zp'_z\widetilde F_z$, and the corresponding
leaf distribution is
\begin{equation}
q'_z
=
\frac{p'_z\widetilde F_z}{r_{\mathrm{root}}}.
\label{eq:SM-approximate-leaf-law}
\end{equation}

Computing $r_v$ directly would require summing over all leaves below
$v$.  We instead compare it with the node weight $W_v$.
Appendix~\ref{sec:SM-oracle} estimates $Z_v$, while
Eq.~\eqref{eq:SM-live-sandwich} compares $Z_v$ with $W_v$.
Define
\begin{equation}
a_v=(1+6\delta)^{h_v}-1,
\label{eq:SM-av}
\end{equation}
where $h_v$ is the remaining depth below $v$.

\begin{lemma}
\label{lem:SM-capped-subtree-ratio}
For every node $v$ of the finite tree,
\begin{equation}
\ee^{-\xi}(1-a_v)W_v
\le
r_v
\le
\ee^\xi(1+a_v)W_v.
\label{eq:SM-capped-sandwich}
\end{equation}
Suppose the oracle in Appendix~\ref{sec:SM-oracle} returns
$\widehat L_v$ satisfying
\begin{equation}
|\widehat L_v-\log Z_v|
\le
\eta_0,
\label{eq:SM-node-oracle-event}
\end{equation}
where $\eta_0$ is a constant independent of the system size.  Then
\begin{equation}
\frac12(1-a_v)\ee^{-(\xi+\eta_0)}
\le
\frac{r_v}{\ee^{\widehat L_v}}
\le
\frac32(1+a_v)\ee^{\xi+\eta_0}.
\label{eq:SM-constant-potential}
\end{equation}
Moreover, if $v\to w$ has positive probability in the finite tree,
then
\begin{equation}
|\log r_w-\log r_v|
\le
\log12.
\label{eq:SM-logvariation}
\end{equation}
\end{lemma}

\begin{proof}
Taking the trace in Eq.~\eqref{eq:SM-finite-completion} gives
\begin{equation}
\Tr E'_v
=
\sum_{z\succeq v}p'(z\,|\,v)F_z.
\label{eq:SM-finite-subtree-trace}
\end{equation}
It follows from Eq.~\eqref{eq:SM-Ftilde} that
\begin{equation}
\ee^{-\xi}\Tr E'_v
\le
r_v
\le
\ee^\xi\Tr E'_v.
\label{eq:SM-rv-Ev-prime}
\end{equation}
Since $\Tr E_v=W_v$, Eq.~\eqref{eq:SM-tree-perturb} also gives
\begin{equation}
(1-a_v)W_v
\le
\Tr E'_v
\le
(1+a_v)W_v.
\label{eq:SM-Ev-prime-trace}
\end{equation}
Combining these inequalities proves
Eq.~\eqref{eq:SM-capped-sandwich}.

Equation~\eqref{eq:SM-live-sandwich} gives
$Z_v/2\le W_v\le3Z_v/2$.  On the event in
Eq.~\eqref{eq:SM-node-oracle-event}, we also have
\begin{equation}
\ee^{-\eta_0}Z_v
\le
\ee^{\widehat L_v}
\le
\ee^{\eta_0}Z_v.
\end{equation}
Substituting these inequalities into
Eq.~\eqref{eq:SM-capped-sandwich} proves
Eq.~\eqref{eq:SM-constant-potential}.

For an edge $v\to w$, Lemma~\ref{lem:SM-edge-bound} and
Eq.~\eqref{eq:SM-capped-sandwich} give
\begin{equation}
\frac{r_w}{r_v}
\le
3\ee^{2\xi}\frac{1+a_w}{1-a_v},
\qquad
\frac{r_v}{r_w}
\le
3\ee^{2\xi}\frac{1+a_v}{1-a_w}.
\label{eq:SM-ratio-edge-proof}
\end{equation}
Since $a_v,a_w\le12h\delta\le0.012$ and $\xi\le0.001$, both
right-hand sides are smaller than $12$.  Hence
$|\log r_w-\log r_v|\le\log12$.
\end{proof}

Thus $\ee^{\widehat L_v}$ approximates $r_v$ within constant
multiplicative factors, and the subtree traces of adjacent nodes are
also multiplicatively comparable.  The next subsection uses these properties to
sample the distribution in
Eq.~\eqref{eq:SM-approximate-leaf-law} without evaluating all children
of a node.

\subsection{Sampling the reweighted leaves}
\label{sec:SM-reweighted-leaves}

If all subtree traces were available, the distribution in
Eq.~\eqref{eq:SM-approximate-leaf-law} could be sampled recursively.
At an internal node $v$, one would choose a child $w$ with probability
\begin{equation}
P'_v(w)\frac{r_w}{r_v}.
\label{eq:SM-ideal-reweighted-child}
\end{equation}
These probabilities sum to one by
Eq.~\eqref{eq:SM-subtree-trace-recursion}.  This rule is not directly
available, since it would require computing $r_w$ for every child of
$v$.  We instead construct a Markov chain on the nodes of the tree.
A downward proposal samples a child from $P'_v$, while the reverse
proposal moves to the parent.  The probability $P'_v(w)$ then cancels
in detailed balance, so the acceptance probability requires only
approximations to the subtree traces at the two adjacent nodes.

We first analyze the chain assuming that positive numbers
$\varphi_v$ have been assigned to the nodes.  At a leaf, set
$\varphi_z=\widetilde F_z$.  At an internal node, suppose
\begin{equation}
A^{-1}r_v
\le
\varphi_v
\le
Ar_v
\label{eq:SM-node-approximation}
\end{equation}
for a constant $A\ge1$.  We also assume that
\begin{equation}
\ee^{-b}
\le
\frac{r_w}{r_v}
\le
\ee^b
\label{eq:SM-adjacent-subtree-weights}
\end{equation}
whenever $v\to w$ and $P'_v(w)>0$.

\begin{lemma}
\label{lem:SM-tree-reweighting}
Let the finite sampling tree have depth at most $h$, and suppose
Eqs.~\eqref{eq:SM-node-approximation} and
\eqref{eq:SM-adjacent-subtree-weights} hold.  For every
$0<\theta<1$, there is an algorithm whose output distribution
$\widetilde q$ on the leaves satisfies
\begin{equation}
\TV(\widetilde q,q')
\le
\theta.
\label{eq:SM-tree-reweighting-error}
\end{equation}
The algorithm requires only samples from $P'_v$, the parent of the
current node, and the values $\varphi_v$.  Its number of steps
is polynomial in $h$, $A$, $\ee^b$, and $\log(1/\theta)$.
\end{lemma}

\begin{proof}
Consider the following lazy Metropolis chain on the nodes that can be
reached from the root.  From a node $v$, stay put with probability
$1/2$.  With probability $1/4$, sample a child from $P'_v$ when $v$ is
internal, and otherwise stay put.  With the remaining probability
$1/4$, propose the parent of $v$ when $v$ is not the root, and
otherwise stay put.  An adjacent proposal $v\to w$ is accepted with
probability
\begin{equation}
\min\left\{1,\frac{\varphi_w}{\varphi_v}\right\}.
\label{eq:SM-node-MH-acceptance}
\end{equation}

Let $p'_v$ be the probability of reaching $v$ under the child
distributions $P'_v$, and define
\begin{equation}
\Pi(v)
=
\frac{p'_v\varphi_v}{\mathcal N},
\qquad
\mathcal N
=
\sum_xp'_x\varphi_x.
\label{eq:SM-node-stationary-law}
\end{equation}
If $w$ is a child of $v$, then $p'_w=p'_vP'_v(w)$, and
\begin{align}
&\Pi(v)\frac{P'_v(w)}4
\min\left\{1,\frac{\varphi_w}{\varphi_v}\right\}
\nonumber\\
&\qquad=
\frac{p'_w}{4\mathcal N}
\min\{\varphi_v,\varphi_w\}
=
\Pi(w)\frac14
\min\left\{1,\frac{\varphi_v}{\varphi_w}\right\}.
\label{eq:SM-tree-detailed-balance}
\end{align}
Thus $\Pi$ is stationary.  The value of $P'_v(w)$ is needed only to
sample the child proposal; it is not evaluated in the acceptance
probability.  Since $\varphi_z=\widetilde F_z$ at every leaf,
conditioning $\Pi$ on the leaf set $\mathsf L$ gives
\begin{equation}
\Pi(z\,|\,\mathsf L)
=
\frac{p'_z\widetilde F_z}
{\sum_{y\in\mathsf L}p'_y\widetilde F_y}
=
q'_z.
\label{eq:SM-leaf-conditional-law}
\end{equation}
The approximations at internal nodes therefore change how long the
chain spends away from the leaves, but not its conditional
distribution on the leaves.

We next bound the time needed for the chain to approach $\Pi$.  Put
$H=h+1$ and define
\begin{equation}
R_v
=
p'_vr_v
=
\sum_{z\succeq v}p'_z\widetilde F_z.
\label{eq:SM-subtree-mass}
\end{equation}
Each leaf contribution appears in $R_v$ only for the nodes on the path
from the root to that leaf, and hence at most $H$ times.  Therefore
\begin{equation}
\sum_vR_v
\le
Hr_{\mathrm{root}},
\qquad
\mathcal N
\le
AHr_{\mathrm{root}}.
\label{eq:SM-total-node-mass}
\end{equation}
It follows that
\begin{equation}
\Pi(\mathsf L)
\ge
\frac1{AH},
\qquad
\Pi(\mathrm{root})
\ge
\frac1{A^2H}.
\label{eq:SM-leaf-root-mass}
\end{equation}

For a nonroot node $w$, let $\mathsf T_w$ consist of $w$ and its
descendants, and let $v$ be its parent.  Applying the same path count
within $\mathsf T_w$ gives
\begin{equation}
\Pi(\mathsf T_w)
\le
\frac{AHR_w}{\mathcal N}.
\label{eq:SM-descendant-mass}
\end{equation}
The stationary flow across the edge $v$--$w$ is
\begin{equation}
\mathcal C(v,w)
=
\frac{p'_w}{4\mathcal N}
\min\{\varphi_v,\varphi_w\}.
\label{eq:SM-edge-flow}
\end{equation}
Equations~\eqref{eq:SM-node-approximation} and
\eqref{eq:SM-adjacent-subtree-weights} imply
\begin{equation}
\mathcal C(v,w)
\ge
\frac{\ee^{-b}R_w}{4A\mathcal N},
\qquad
\frac{\Pi(\mathsf T_w)}{\mathcal C(v,w)}
\le
4A^2\ee^bH.
\label{eq:SM-subtree-flow-ratio}
\end{equation}

Use the unique path in the tree as the canonical path between two
nodes.  Its length is at most $2h$, and the edge $v$--$w$ is used only
when one endpoint lies in $\mathsf T_w$ and the other does not.  The
canonical path bound therefore gives
\begin{equation}
\operatorname{gap}^{-1}
\le
8A^2\ee^b hH
=:
\mathcal R.
\label{eq:SM-tree-gap}
\end{equation}
Together with the lower bound on $\Pi(\mathrm{root})$, this yields
\begin{equation}
\norm{K^t(\mathrm{root},\cdot)-\Pi}_{\mathrm{TV}}
\le
\frac{A\sqrt H}{2}
\exp\left(-\frac{t}{\mathcal R}\right).
\label{eq:SM-tree-mixing}
\end{equation}

Set $\varepsilon_{\mathrm{mix}}=\theta/(16AH)$ and take
\begin{equation}
t
=
\left\lceil
\mathcal R
\log\left(
\frac{8A^2H^{3/2}}{\theta}
\right)
\right\rceil.
\label{eq:SM-tree-burnin}
\end{equation}
If $\mu=K^t(\mathrm{root},\cdot)$, then
$\norm{\mu-\Pi}_{\mathrm{TV}}\le\varepsilon_{\mathrm{mix}}$.
Since $\Pi(\mathsf L)\ge1/(AH)$,
\begin{equation}
\TV\left(
\mu(\,\cdot\,|\,\mathsf L),q'
\right)
\le
\frac{2\varepsilon_{\mathrm{mix}}}
{\Pi(\mathsf L)-\varepsilon_{\mathrm{mix}}}
\le
\frac\theta4.
\label{eq:SM-conditioned-leaf-law}
\end{equation}

Before running the chain, sample one path from the root to a leaf and
store the resulting leaf $z_0$.  Run
\begin{equation}
R
=
\left\lceil
4AH\log\frac2\theta
\right\rceil
\label{eq:SM-leaf-trials}
\end{equation}
independent $t$-step chains from the root, and return the endpoint of
the first chain that ends at a leaf.  If none does, return $z_0$.
Since $\mu(\mathsf L)\ge\Pi(\mathsf L)-\varepsilon_{\mathrm{mix}}$,
the probability of using $z_0$ is at most $\theta/2$.  Together with
Eq.~\eqref{eq:SM-conditioned-leaf-law}, this gives total variation
error at most $3\theta/4$ when the acceptance comparisons are exact.

There are at most $M=Rt$ acceptance decisions.  Evaluate each
acceptance probability so that the sum of their errors is at most
$\theta/4$.  Coupling the
exact and approximate decisions then proves
Eq.~\eqref{eq:SM-tree-reweighting-error} and the stated running time.
\end{proof}

We now construct the values $\varphi_v$ required by the lemma.  At a
leaf, set $\varphi_z=\widetilde F_z$.  At an internal node, run the
oracle in Appendix~\ref{sec:SM-oracle} with a fixed accuracy
$\eta_0$, and set $\varphi_v=\ee^{\widehat L_v}$.  Store this value when the node is first
visited and use it on every later visit.  On the event that the oracle
guarantee holds, Eq.~\eqref{eq:SM-constant-potential} gives
Eq.~\eqref{eq:SM-node-approximation} for a constant $A$, while
Eq.~\eqref{eq:SM-logvariation} gives
Eq.~\eqref{eq:SM-adjacent-subtree-weights} with $b=\log12$.

Let $M_{\mathrm{rw}}=1+Rt$.  This bounds the total number of distinct
internal nodes whose weights can be requested across all the chains.
Set
\begin{equation}
\zeta
=
\frac{\epsilon}{1000},
\qquad
\delta_{\mathrm{or}}
=
\frac{\zeta}{M_{\mathrm{rw}}}.
\label{eq:SM-oracle-failure}
\end{equation}
Use independent randomness for the first oracle call at each new node
and request failure probability at most $\delta_{\mathrm{or}}$.  For
the analysis, choose this randomness independently for every internal
node before running the walk.  Whenever an output fails its guarantee,
replace it by any fixed value that satisfies the guarantee.  Conditional
on the resulting values, Lemma~\ref{lem:SM-tree-reweighting} applies.
Couple the implemented walk to this modified walk using the same
proposals and acceptance decisions.  They agree until the first
inaccurate oracle value is requested.  At the first query of a new
node, its oracle randomness is independent of the preceding trajectory,
so a union bound over at most $M_{\mathrm{rw}}$ queried nodes bounds the
probability that the two walks differ by $\zeta$.  If an oracle call fails,
the algorithm still returns a leaf, but its distribution is left
uncontrolled.  Therefore the unconditional leaf distribution $q''$
satisfies
\begin{equation}
\TV(q'',q')
\le
\theta+\zeta.
\label{eq:SM-TV}
\end{equation}

\subsection{Error and complexity analysis}
\label{sec:SM-output-error}

Let $z$ be the leaf selected by the reweighted tree walk.  The
algorithm returns the normalized state
\begin{equation}
\tau_z
=
\frac{L_z}{F_z}
=
\frac{G\sigma_zG}{\Tr(G\sigma_zG)}.
\label{eq:SM-output}
\end{equation}
Lemma~\ref{lem:SM-Gaussian-sandwich} shows that $\tau_z$ is a positive
Gaussian state.  It is specified by the data
\begin{equation}
\left(K_0,\beta,\cM_z,(s_e)_{e\in\cM_z},t\right),
\label{eq:SM-factor-description}
\end{equation}
and the normalization in Eq.~\eqref{eq:SM-output}.  Any positive
scalar in $\sigma_z$ enters the trace weight $F_z$ and cancels in the
normalized state $\tau_z$.

The same data also determine the two-point function.  Let
$S_{\cM_z}$ be the transfer matrix of the matching at the leaf and
let $E=\ee^{-TK_0}$.  The formulas in
Appendix~\ref{sec:SM-oracle-Gaussian} give
\begin{equation}
R_z
=
ES_{\cM_z}E,
\qquad
C_{\tau_z}
=
2(\one+R_z)^{-1},
\label{eq:SM-output-cov}
\end{equation}
where $(C_{\tau_z})_{ab}=\Tr(\tau_z\gamma_a\gamma_b)$.  Thus all
two-point functions, and hence all Wick correlators, can be computed
from Eq.~\eqref{eq:SM-factor-description} using $N\times N$ matrices.

The factor data in Eq.~\eqref{eq:SM-factor-description} and the
covariance matrix in Eq.~\eqref{eq:SM-output-cov} are equivalent
polynomial-size classical descriptions of the output state.

We now compare the average output with the Gibbs state.  Recall that
$E'=E'_{\mathrm{root}}$, and set
\begin{equation}
Z'
=
\Tr E',
\qquad
q_z^{\mathrm{fin}}
=
\frac{p'_zF_z}{Z'}.
\label{eq:SM-finite-leaf-law}
\end{equation}
The exact trace reweighting of the finite tree gives
\begin{equation}
\sum_zq_z^{\mathrm{fin}}\tau_z
=
\frac{E'}{Z'}.
\label{eq:SM-finite-mixture}
\end{equation}
Replacing $F_z$ by $\widetilde F_z$ changes this law to $q'$.  Let
$q''$ be the leaf law produced by the reweighted tree walk.
The following proposition combines these two changes with the
approximation of the sampling tree.

\begin{proposition}
\label{prop:SM-unconditional-error}
The average output of the sampler satisfies
\begin{equation}
\left\|
\bbE\tau-\rho_\beta
\right\|_1
\le
\frac{2a}{1-a}
+
(\ee^{2\xi}-1)
+
2\theta
+
2\zeta.
\label{eq:SM-master-error}
\end{equation}
For
\begin{equation}
\xi
=
\theta
=
\zeta
=
\frac{\epsilon}{1000},
\label{eq:SM-final-error-parameters}
\end{equation}
the right-hand side is smaller than $0.04\epsilon$, and in particular
smaller than $\epsilon$.
\end{proposition}

\begin{proof}
Let $Z=\Tr\ee^{-\beta H}$.  Equation~\eqref{eq:SM-root-perturb}
implies $|Z'-Z|\le aZ$ and hence $Z'\ge(1-a)Z$.  It follows that
\begin{equation}
\left\|
\frac{E'}{Z'}-\rho_\beta
\right\|_1
\le
\frac{2a}{1-a}.
\label{eq:SM-normalization-error}
\end{equation}

Let
$\widetilde Z'=r_{\mathrm{root}}=\sum_zp'_z\widetilde F_z$.  From
Eq.~\eqref{eq:SM-Ftilde},
\begin{equation}
\ee^{-2\xi}
\le
\frac{q'_z}{q_z^{\mathrm{fin}}}
\le
\ee^{2\xi},
\label{eq:SM-leaf-law-ratio}
\end{equation}
and therefore
\begin{equation}
\sum_z\left|q'_z-q_z^{\mathrm{fin}}\right|
\le
\ee^{2\xi}-1.
\label{eq:SM-leaf-law-distance}
\end{equation}
Since $\norm{\tau_z}_1=1$, Eq.~\eqref{eq:SM-TV} gives
\begin{equation}
\left\|
\sum_zq''_z\tau_z
-
\sum_zq'_z\tau_z
\right\|_1
\le
2(\theta+\zeta).
\label{eq:SM-reweighting-mixture-error}
\end{equation}
Combining these inequalities with Eqs.~\eqref{eq:SM-finite-mixture} and
\eqref{eq:SM-normalization-error} proves
Eq.~\eqref{eq:SM-master-error}.  Finally,
$a\le0.012\epsilon$, Eq.~\eqref{eq:SM-final-error-parameters}, and
$0<\epsilon\le1$ give the stated numerical bound and prove
Eq.~\eqref{eq:SM-main-alg}.
\end{proof}

The finite tree has depth $h=\order(n^2)$, so a node is specified by a
polynomial-length path from the root.  Lemma~\ref{lem:SM-capped-kernel}
samples each child in
$N^{\order(d)}\poly(\log(n/\epsilon))$ time.  In
Lemma~\ref{lem:SM-tree-reweighting}, the constants $A$ and $\ee^b$ are
independent of $n$ and $\epsilon$, and
$\theta=\epsilon/1000$; hence the number of Markov chain steps is
polynomial in $n$ and $\log(1/\epsilon)$.  The walk queries only
polynomially many log weights, each to constant additive accuracy and
with $\log(1/\delta_{\mathrm{or}})=\poly(\log(n/\epsilon))$.
The oracle inputs, leaf traces, and Gaussian matrices are all computed
to polynomial precision from the data specified in
Theorem~\ref{thm:SM-sampler}.  Thus the total time and space are
polynomial in $n$ and $\log(1/\epsilon)$.

\section{Physical quantities accessible with the Gibbs sampler}
\begin{corollary} 
\label{cor:physical-quantities}
Assume the sampling conditions of Theorem~\ref{thm:SM-sampler}.  Run
the sampler with error parameter $\epsilon_{\mathrm{samp}}/2$, so that its Gaussian
output satisfies
\begin{equation}
\norm{\bbE\tau-\rho_\beta}_1\le\frac{\epsilon_{\mathrm{samp}}}{2}.
\label{eq:cor-sampler-bias}
\end{equation}
For statistical accuracy $\delta>0$ and failure probability
$\kappa>0$, the following quantities are classically accessible in
time polynomial in $n$, $1/\delta$, $\log(1/\kappa)$, and
$\log(1/\epsilon_{\mathrm{samp}})$; the moment estimator below also
has polynomial dependence on its order $k$.

\begin{enumerate}
\item
Let $O$ be an operator whose Majorana expansion contains at most
polynomially many nonzero monomials.  Its support is supplied
explicitly, and each nonzero coefficient can be evaluated to relative
error $2^{-p}$ in time polynomial in $n$ and $p$.  The sample average gives an
estimate $\widehat O$ such that, with probability at least $1-\kappa$,
\begin{equation}
 \left|\widehat O-\Tr(O\rho_\beta)\right|
 \le(\epsilon_{\mathrm{samp}}+\delta)\norm O.
 \label{eq:observable-consequence}
\end{equation}

\item
For any subset $A$ of fermionic modes, consider the simultaneous
projective measurement of the occupation operators $\{n_i\}_{i\in A}$.
The corresponding measurement-outcome distribution on $\rho_\beta$
can be sampled with total variation error at most
$\epsilon_{\mathrm{samp}}/2$.
Consequently, one can efficiently sample the full counting statistics of
\begin{equation}
 N_A=\sum_{i\in A}n_i,
\end{equation}
as well as the fermion parity $(-1)^{N_A}$.

\item
Let $A$ be any subsystem and let $\tau_A$ and $\rho_{\beta,A}$ denote
the corresponding reduced states.  For any integer $k\ge2$, repeated
groups of $k$ independent sampler outputs give an estimate
$\widehat m_k$ such that, with probability at least $1-\kappa$,
\begin{equation}
\left|\widehat m_k-\Tr(\rho_{\beta,A}^k)\right|
\le k\epsilon_{\mathrm{samp}}+\delta.
\label{eq:cor-moment-estimate}
\end{equation}
For one group of outputs
$\tau^{(1)},\ldots,\tau^{(k)}$,
\begin{equation}
 \bbE
 \Tr\left(
 \tau_A^{(1)}
 \tau_A^{(2)}
 \cdots
 \tau_A^{(k)}
 \right)
 =
 \Tr\left[(\bbE\tau_A)^k\right],
 \label{eq:cor-moment-mean}
\end{equation}
and
\begin{equation}
 \left|
 \Tr\left[(\bbE\tau_A)^k\right]
 -
 \Tr(\rho_{\beta,A}^k)
 \right|
 \le k\epsilon_{\mathrm{samp}}.
 \label{eq:cor-moment}
\end{equation}

\item
If $A$ is nonempty and contains $\order(\log n)$ fermionic modes, then the subsystem
von Neumann entropy
\begin{equation}
 S(\rho_{\beta,A})
 =
 -\Tr\left(\rho_{\beta,A}\log\rho_{\beta,A}\right)
\end{equation}
can be estimated to inverse-polynomial additive accuracy in polynomial
time.  In particular, the exact sampler mean obeys
\begin{equation}
\begin{aligned}
 \left|
 S(\bbE\tau_A)-S(\rho_{\beta,A})
 \right|
 &\le
 x_{\mathrm{samp}}\log\left(2^{|A|}-1\right)
 +h\left(x_{\mathrm{samp}}\right),\\
 x_{\mathrm{samp}}
 &:=
 \min\left\{
 \frac{\epsilon_{\mathrm{samp}}}{2},
 1-2^{-|A|}
 \right\}.
\end{aligned}
 \label{eq:cor-entropy}
\end{equation}
 where $h$ is the binary entropy function.

\end{enumerate}
\end{corollary}
\begin{proof}
Equation~\eqref{eq:cor-sampler-bias} controls the systematic error of
the Gaussian sampler.  Independent repetitions are used below to
control the additional statistical error.
We prove the four statements separately.

For the first statement, write the Majorana expansion of $O$ as
\begin{equation}
 O=\sum_R o_R P_R,
\end{equation}
where, by assumption, only $K=\poly(n)$ coefficients $o_R$ are
nonzero.  Orthogonality of the Majorana monomials gives
\begin{equation}
\sum_R|o_R|
\le
\sqrt K\left(\sum_R|o_R|^2\right)^{1/2}
\le
\sqrt K\norm O.
\end{equation}
For every Gaussian output $\tau$, Wick's theorem gives
$\Tr(P_R\tau)$ as the Pfaffian of the corresponding submatrix of the
covariance matrix of $\tau$.  Since each such Pfaffian has dimension at
most $2n$, $\Tr(P_R\tau)$ is computable in polynomial time.  Hence
$\Tr(O\tau)$ is computable in polynomial time from the classical
description produced by the sampler.  Together with the runtime of
Theorem~\ref{thm:SM-sampler}, this gives the stated polynomial runtime.
Moreover,
\begin{align}
 \left|\bbE\Tr(O\tau)-\Tr(O\rho_\beta)\right|
 &=
 \left|\Tr\left[O(\bbE\tau-\rho_\beta)\right]\right|
 \nonumber\\
 &\le
 \norm O\,\norm{\bbE\tau-\rho_\beta}_1
 \le\epsilon_{\mathrm{samp}}\norm O.
\end{align}
The random variable $\Tr(O\tau)$ has absolute value at most
$\norm O$.  Applying Hoeffding's inequality to its real and imaginary
parts therefore shows that
$\order(\delta^{-2}\log(1/\kappa))$ independent outputs determine its
mean to accuracy $(\delta/2)\norm O$ with probability at least
$1-\kappa$.  The displayed coefficient bound shows that relative
coefficient accuracy and absolute Pfaffian accuracy of order
$\delta/\sqrt K$ suffice to evaluate each Pfaffian sum to additive
accuracy $(\delta/2)\norm O$.  Combining these two errors with the sampler bias
proves Eq.~\eqref{eq:observable-consequence}.

For the second statement, the occupation operators
$\{n_i\}_{i\in A}$ commute, and their simultaneous projective
measurement is a fermionic Gaussian measurement.  Given the covariance
matrix of a Gaussian state, its occupation outcomes can be sampled
sequentially using the standard Gaussian measurement update rules.
Evaluate their conditional probabilities so that the implemented
measurement law is within $\epsilon_{\mathrm{samp}}/4$ in total
variation distance of the exact measurement law.  Let $q$ denote the
implemented outcome distribution, let $\overline q$ denote the exact
outcome distribution obtained by measuring the sampler output, and let
$p_\beta$ denote the outcome distribution obtained from $\rho_\beta$.
Using the total variation distance
\begin{equation}
 \TV(q,p_\beta)
 :=
 \frac12\sum_{\bm s}
 \left|q(\bm s)-p_\beta(\bm s)\right|,
\end{equation}
contractivity of the trace norm under the measurement channel gives
\begin{equation}
 \TV(\overline q,p_\beta)
 \le
 \frac12\norm{\bbE\tau-\rho_\beta}_1
 \le\frac{\epsilon_{\mathrm{samp}}}{4}.
\end{equation}
Therefore,
\begin{equation}
 \TV(q,p_\beta)
 \le
 \TV(q,\overline q)+\TV(\overline q,p_\beta)
 \le\frac{\epsilon_{\mathrm{samp}}}{2}.
\end{equation}
For an occupation outcome
$\bm s=(s_i)_{i\in A}\in\{0,1\}^{|A|}$, the corresponding value of
$N_A$ is
\begin{equation}
 m(\bm s)=\sum_{i\in A}s_i.
\end{equation}
Hence the full counting statistics of $N_A$ are given by
\begin{equation}
 q_N(m)
 =
 \sum_{\bm s:\,m(\bm s)=m}q(\bm s),
 \qquad
 p_{\beta,N}(m)
 =
 \sum_{\bm s:\,m(\bm s)=m}p_\beta(\bm s),
 \qquad
 m=0,\ldots,|A|.
\end{equation}
Since this is a deterministic classical postprocessing of the occupation
outcomes,
\begin{equation}
 \TV(q_N,p_{\beta,N})
 \le
 \TV(q,p_\beta)
 \le\frac{\epsilon_{\mathrm{samp}}}{2}.
\end{equation}
Likewise, the fermion parity associated with $\bm s$ is
\begin{equation}
 (-1)^{N_A}=(-1)^{m(\bm s)},
\end{equation}
so its two-outcome distribution is obtained from the same samples by the
map $\bm s\mapsto(-1)^{m(\bm s)}$, and therefore also has total variation
error at most $\epsilon_{\mathrm{samp}}/2$.

For the third statement, partial trace preserves Gaussianity, so every
$\tau_A^{(j)}$ is again a fermionic Gaussian state. Recall the
one-particle transfer matrix $R(W)$ defined in
Eq.~\eqref{eq:SM-transfer}. Since
\begin{equation}
 R(W_1W_2)=R(W_1)R(W_2),
\end{equation}
the product
\begin{equation}
 \tau_A^{(1)}\tau_A^{(2)}\cdots\tau_A^{(k)}
\end{equation}
is again a Gaussian operator, with
\begin{equation}
 R\left(
 \tau_A^{(1)}\tau_A^{(2)}\cdots\tau_A^{(k)}
 \right)
 =
 \prod_{j=1}^k R\left(\tau_A^{(j)}\right).
\end{equation}
The transfer matrix does not record the scalar multiplier needed for
the trace.  We retain this scalar and use the ordered Grassmann
representation of each factor.  The Grassmann formulas of
Ref.~\cite{Bravyi2005} express their products and traces as Pfaffians
with a fixed Majorana ordering, including for singular intermediate
Gaussian operators.  Hence
\begin{equation}
 \Tr\left(
 \tau_A^{(1)}
 \tau_A^{(2)}
 \cdots
 \tau_A^{(k)}
 \right)
\end{equation}
is computable in $\poly(k,n)$ time.

Since the $k$ executions are independent,
\begin{align}
 \bbE
 \Tr\left(
 \tau_A^{(1)}
 \tau_A^{(2)}
 \cdots
 \tau_A^{(k)}
 \right)
 &=
 \Tr\left[
 (\bbE\tau_A)
 (\bbE\tau_A)
 \cdots
 (\bbE\tau_A)
 \right]
 \nonumber\\
 &=
 \Tr\left[(\bbE\tau_A)^k\right],
\end{align}
which proves Eq.~\eqref{eq:cor-moment-mean}.  By contractivity of the
trace norm under partial trace,
\begin{equation}
 \norm{\bbE\tau_A-\rho_{\beta,A}}_1
 \le\frac{\epsilon_{\mathrm{samp}}}{2}.
 \label{eq:cor-proof-reduced}
\end{equation}
Using the telescoping identity for the difference of the $k$th powers
and the fact that both $\bbE\tau_A$ and $\rho_{\beta,A}$ are density
operators,
\begin{align}
 \left|
 \Tr\left[(\bbE\tau_A)^k\right]
 -
 \Tr(\rho_{\beta,A}^k)
 \right|
 &\le
 \sum_{j=0}^{k-1}
 \norm{
 (\bbE\tau_A)^j
 (\bbE\tau_A-\rho_{\beta,A})
 \rho_{\beta,A}^{\,k-1-j}
 }_1
 \nonumber\\
 &\le
 k\norm{\bbE\tau_A-\rho_{\beta,A}}_1
 \le k\epsilon_{\mathrm{samp}},
\end{align}
which proves Eq.~\eqref{eq:cor-moment}.  The trace of a product of
$k$ density operators has absolute value at most one.  Repeating the
$k$-sample experiment
$\order(\delta^{-2}\log(1/\kappa))$ times determines its mean to
accuracy $\delta/2$ with probability at least $1-\kappa$.  Evaluate
each trace by the same Pfaffian formulas to additive
accuracy $\delta/2$.
The combined numerical and statistical error is at most $\delta$, which
proves Eq.~\eqref{eq:cor-moment-estimate}.

Finally, suppose that $A$ contains $\order(\log n)$ fermionic modes.
The Majorana monomials on $A$ form an orthogonal operator basis, and
there are
\begin{equation}
 4^{|A|}=\poly(n)
\end{equation}
such monomials.  The reduced state $\bbE\tau_A$ is therefore determined
by the expectation values
\begin{equation}
 \bbE\Tr(P_R\tau_A)
 =
 \bbE\Tr(P_R\tau)
\end{equation}
for Majorana monomials $P_R$ supported in $A$.  Wick's theorem evaluates
each $\Tr(P_R\tau)$ from the sampled covariance matrix.
Since only polynomially many coefficients are required, simultaneous
concentration and the Majorana-basis expansion first give a Hermitian
estimate with trace-norm error at most $\delta/2$ from $\bbE\tau_A$.
Choosing a nearest density operator
$\widehat\rho_A$ in trace norm then gives
\begin{equation}
\norm{\widehat\rho_A-\bbE\tau_A}_1\le\delta
\end{equation}
with probability at least $1-\kappa$, using polynomially many sampler
outputs.

The Hilbert-space dimension of $A$ is
\begin{equation}
 2^{|A|}=\poly(n),
\end{equation}
so $S(\widehat\rho_A)$ can subsequently be computed to the required
precision by polynomial-dimensional matrix diagonalization.  The
systematic contribution follows from Eq.~\eqref{eq:cor-proof-reduced}
and the Fannes--Audenaert inequality:
\begin{equation}
 \left|
 S(\bbE\tau_A)-S(\rho_{\beta,A})
 \right|
 \le
 x_{\mathrm{samp}}
 \log\left(2^{|A|}-1\right)
 +h\left(x_{\mathrm{samp}}\right),
\end{equation}
which proves Eq.~\eqref{eq:cor-entropy}.  Including the statistical
reconstruction error replaces $x_{\mathrm{samp}}$ in the right-hand
side by
\begin{equation}
x
=
\min\left\{
\frac{\epsilon_{\mathrm{samp}}+\delta}{2},
1-2^{-|A|}
\right\}.
\end{equation}
For $|A|=\order(\log n)$ and inverse-polynomial
$\epsilon_{\mathrm{samp}}$ and $\delta$, this bound is also
inverse-polynomial, up to the logarithmic dimension factor.
\end{proof}
\bibliographystyle{unsrt}
\bibliography{references}

\end{document}